\documentclass[10pt,a4paper]{article}
\usepackage[utf8]{inputenc}
\usepackage[T1]{fontenc}
\usepackage[english]{babel}
\usepackage{amsmath, amssymb, amsthm}
\usepackage{graphicx}
\usepackage{booktabs}
\usepackage{tabularx}
\usepackage{hyperref}
\usepackage{geometry}
\usepackage{authblk}
\usepackage{subcaption}
\usepackage{enumitem}
\usepackage{stmaryrd}
\numberwithin{equation}{section}
\usepackage{float}
\usepackage{lipsum}
\usepackage{cmbright}
\usepackage{algorithm}
\usepackage{algorithmic}
\usepackage[authoryear]{natbib}
\usepackage{microtype}
\usepackage{xurl}
\title{\textbf{Post Hoc Inference for Component Attribution in Multivariate Change-Point Detection}}

\author[1,2]{Dhia-Elhaq Ouerfelli \thanks{dhia-el-haq.ouerfelli@universite-paris-saclay.fr}}
\author[1,2]{Sylvain Arlot \thanks{sylvain.arlot@inria.fr}}
\author[1,2]{Kevin Bleakley \thanks{kevin.bleakley@inria.fr}}
\author[1,2]{Patrick Pamphile \thanks{patrick.pamphile@universite-paris-saclay.fr}}
\affil[1]{Université Paris-Saclay, Laboratoire de Mathématiques d'Orsay, France}
\affil[2]{Inria Saclay, France}

\date{} 

\newtheorem{theorem}{Theorem}

\newtheorem{lemma}[theorem]{Lemma}

\newtheorem{remark}{Remark}

\newtheorem{assumption}{Assumption}

\usepackage{xcolor}

\usepackage{comment}

\usepackage{tikz}
\usetikzlibrary{decorations.pathreplacing} 
\usetikzlibrary{calc} 
\usetikzlibrary{positioning,arrows.meta}

\begin{document}

\maketitle
\begin{abstract}
We consider the post-detection analysis of change-points for multivariate time series, with the goal of identifying which coordinates are responsible for a detected change. After a change-point has been located by an offline detection algorithm, we propose \emph{post hoc} statistical procedures to determine whether the change occurs in either of two predefined blocks of coordinates or in both. Our methods rely on two-sample testing procedures with a particular focus on nonparametric tests; we provide theoretical guarantees for Type I error control. Simulations and a real-data experiment demonstrate the strong performance of the proposed procedures.
\end{abstract}

\section{Introduction}

Change-point detection is a fundamental issue in time series analysis. The goal is to identify times when the distribution of observations changes abruptly, whether these changes affect the mean, variance, or more general distributional properties \cite{Basseville1993, Aminikhanghahi2017}. Extensive literature has driven the development of diverse procedures, covering both parametric and nonparametric models \cite{harchaoui2007, matteson2014nonparametric, Arlot2019}, which provide strong theoretical guarantees regarding the estimation of the number of change-points and their temporal locations \cite{Truong_2020, Garreau2018}. However, in multivariate settings, detecting a change-point is often only the first step. Once a structural change is detected in a multivariate time series, practitioners need to identify which specific variables or components of the signal are responsible for the observed distributional shift. This \emph{post hoc} interpretability is crucial for effective decision-making across diverse domains. For example, in industrial process control, one must identify which subsystem or specific group of sensors is responsible for a detected anomaly \cite{Oakland2007, DING2023, TANG2025}. Similarly, in epidemiology, effective intervention requires determining which geographical regions contribute to an observed spike in reported cases \cite{unkel2012statistical, Chen2024}. Furthermore, in finance, practitioners seek to identify which specific sectors or assets are driving a structural break detected in the broader market \cite{Ross2013, McMillan01082011, Malik2021}.

From a statistical perspective, identifying these variables is a problem of \emph{post hoc} inference and explainability. In practice, this inference is often conducted naively by comparing observations before and after the estimated change-point. However, because the samples used for this inference are defined based on a change-point that is itself estimated from the data, a critical selection bias is introduced. Consequently, the observations are no longer independent of the selection procedure, invalidating the guarantees of standard statistical tests \cite{lehmann2005testing, gretton2012}. Ignoring this dependence in nonparametric two-sample testing leads to a severe loss of Type I error control. This issue falls within the broader framework of post-selection inference, where a data-driven selection step affects the distribution of subsequent test statistics.

A substantial literature focuses on hypothesis testing for change-point detection, with the aim of assessing whether a detected change-point corresponds to a statistically significant distributional change. These tests are typically constructed conditionally on the event that a change-point has been detected by a given algorithm \cite{jewell2022testing, bhattacharyya2025theoretical, carrington2025improving, carrington2024post}. Notably, \cite{umezu2017selective} and \cite{sugiyama2021valid} perform inference by conditioning on both the detected change-points and the coordinate support selected by the detection algorithm. However, these works address a different inferential question from ours: they aim to determine whether a detected change-point corresponds to a true change in distribution. In contrast, our goal is not to assess the existence of a change-point, but to identify which variables or coordinates are responsible for the detected change. 

Other works investigate variable selection or interpretability in the context of distributional comparison, often through kernel-based methods \cite{mitsuzawa2023variable, mitsuzawa2024variable, wang2023variable}. These approaches aim to identify which coordinates contribute to differences between two multivariate distributions. However, they typically assume that the two samples being compared are fixed and independent \emph{a priori}, thereby failing to account for the \emph{post hoc} selection bias induced when samples are defined based on an estimated change-point.

Some further contributions address variable attribution explicitly within the context of change-point detection, which is closely related to our setting. For instance, \cite{delaconcha2023} propose an online nonparametric approach for attributing distributional changes in a graph-structured framework, based on likelihood-ratio estimation, with the goal of identifying nodes responsible for a detected change. However, their method does not provide theoretical guarantees on exact block recovery or error control. Another recent work \cite{Chen2024} considers online change-point detection under parametric assumptions, such as mean shifts, and develops inference procedures for change-point locations. They also provide theoretical guarantees for correctly estimating the support of the affected coordinates.

While existing approaches address related objectives, the validity of \emph{post hoc} attribution procedures remains largely unexplored in broader contexts, especially when multiple change-points are involved and the underlying distributional change is not restricted to a specific parametric form.

In this paper, we address these gaps by proposing an offline, nonparametric, and multivariate framework for \emph{post hoc} change-point attribution. Our contributions are explicitly tailored to overcome the limitations of existing approaches. First, we introduce a novel inference procedure, which we call GTST, specifically designed to identify predefined blocks of variables affected by an estimated change-point. Second, alongside GTST, we formalize the mathematical framework and assumptions required to properly apply sample splitting techniques, namely the hold-out method \cite{Cox_1975}, in this \emph{post hoc} context. Crucially, we provide theoretical guarantees for the valid control of the Type I error rate for both the GTST procedure and the hold-out method, effectively overcoming the selection bias induced when samples are defined from a data-driven change-point estimate \cite{Garreau2018}. Formally, the coordinates of the multivariate time series are partitioned into two predefined blocks, and our procedures determine whether the detected general distributional change occurs in the first block, in the second block, or in both. Finally, we support our theoretical results with extensive empirical evaluations---on both synthetic and real-world data---demonstrating that our proposed procedures not only control the Type I error but also achieve high statistical power.

This block-wise approach allows for a precise localization of the change. For instance, in clinical monitoring, variables can be partitioned into Left and Right Hemisphere blocks: a change detected exclusively in one suggests a localized neurological event, whereas a simultaneous change in both indicates a systemic physiological shift. Similarly, in infrastructure monitoring—such as a water network divided into North and South sectors—a change detected in only one sector pinpoints a localized leak, while changes in both suggest a main supply line rupture. Beyond these examples, this framework extends to physical sensor networks (e.g., seismic stations or railway surveillance \cite{delaconcha2023}), human activity recognition using sensors distributed across different body parts \cite{sugiyama2021valid}, and biomedical studies targeting anomalies localized within specific genomic regions \cite{umezu2017selective}.

While our analysis focuses on two blocks of coordinates, the proposed framework extends to multiple blocks through suitable multiple testing corrections. More importantly, the framework is modular: it does not rely on a specific change-point detection algorithm or a particular two-sample testing procedure. The theoretical guarantees only require that the change-point estimator satisfies appropriate localization properties and that the correct number of change-points is either known or consistently estimated. Overall, the proposed approach provides a general and theoretically grounded framework for \emph{post hoc} attribution of distributional changes in multivariate time series.

\paragraph{Outline.} The paper is organized as follows: Section~\ref{sec:background} introduces the background and assumptions needed in the rest of the paper. In Section~\ref{sec:methods}, we present new methods to test whether there is a change in a specific block of coordinates around the estimated change-point, and provide a theoretical analysis for Type I error control. Section~\ref{sec:experiments} reports the results of empirical evaluations on various data sets, including synthetic and real-world data. Section~\ref{sec:conclusion} concludes the paper and discusses future research directions.

\section{Background}
\label{sec:background}

\subsection{Definition of parameters}

Let $\mathcal{X}$ be a measurable space. Consider a multivariate time series $\mathbf{X} := (X_t)_{t \in \llbracket 1, n \rrbracket}$ of length $n \geq 2$, where each $X_t \in \mathcal{X}$. We assume that there exist true change-point instants 
\[
  1 \le \tau_1^\star < \tau_2^\star < \cdots < \tau_{\kappa^\star-1}^\star \le n-1,
\]
where $\kappa^\star \in \mathbb{N}^\star$ denotes the number of segments (i.e., the number of true change-points plus one). We extend this notation by setting $\tau_0^\star := 0$ and $\tau_{\kappa^\star}^\star := n$.

The intervals $\llbracket \tau_{i-1}^\star + 1, \tau_i^\star \rrbracket$, for $i \in \llbracket 1, \kappa^\star \rrbracket$, form a partition of $\llbracket 1, n \rrbracket$ into segments. We assume that all observations $(X_t)_{t \in \llbracket 1, n \rrbracket}$ are mutually independent. Specifically, for each $i \in \llbracket 1, \kappa^\star \rrbracket$, the observations within the segment $\llbracket \tau_{i-1}^\star + 1, \tau_{i}^\star \rrbracket$ are independent and identically distributed (i.i.d.) according to a distribution denoted by $P_{(i)}$. By definition of a true change-point, the distribution changes from one segment to the next, meaning:
\[
  P_{(i)} \neq P_{(i+1)} \qquad \text{for all } i \in \llbracket 1, \kappa^\star-1 \rrbracket.
\]
We define the parameter $\underline{\Lambda}_{\tau^\star}$ as the relative length of the smallest segment between two change-points:
\[
\underline{\Lambda}_{\tau^\star}:= \frac{1}{n}\min_{1 \leq i \leq \kappa^\star} \left(\tau^\star_{i} - \tau^\star_{i-1}\right).
\]
This parameter quantifies the minimal temporal separation between two consecutive change-points as a proportion of the total length of the time series. A small value for $\underline{\Lambda}_{\tau^\star}$ indicates that some change-points are close to each other, making them difficult to identify accurately, whereas a large value ensures they are well-separated and easier to estimate reliably.

\par The Hausdorff distance \cite{Garreau2018}, denoted by $d_H$, can be used to measure the discrepancy between two segmentations. For two segmentations $\tau^1$ and $\tau^2$ of $\llbracket 1, n\rrbracket $, we have
\[
d_H(\tau^1, \tau^2) := \max\left\{d_{\infty}^{(1)}(\tau^1, \tau^2), d_{\infty}^{(1)}(\tau^2, \tau^1)\right\},
\]
where
\[
   d_{\infty}^{(1)}(\tau^1, \tau^2):= \max_{1 \leq i \leq |\tau^1|-2} \left\{ \min_{1 \leq j \leq |\tau^2|-2} \left| \tau_i^1 - \tau_j^2 \right| \right\}.
\]

\subsection{Assumptions and properties related to change-point detection}

Let $\alpha_0 \in (0, 1)$. We assume the existence of a high-probability event related to the initial change-point detection algorithm, which we formalize as follows:

\begin{assumption}[Change-point detection and localization event]
\label{assum:detection_event}
There exists an event $\Omega_{\alpha_0}$ with $\mathbb{P}(\Omega_{\alpha_0}) \geq 1 - \alpha_0$ upon which the change-point detection algorithm provides an estimated segmentation $\hat{\tau}$ (with $|\hat{\tau}| = \widehat{\kappa} + 1$) satisfying the following two properties:
\begin{enumerate}
    \item \textbf{Correct number of change-points:} 
    \begin{equation}\label{eq:kappa_star}
         \widehat{\kappa} = \kappa^\star.
    \end{equation}
    \item \textbf{Localization accuracy:} There exists $\delta_n = \delta_n(\alpha_0) \in \mathbb{N}$ such that:
    \begin{equation}\label{eq:delta_n}
        d_{\infty}^{(1)}(\tau^\star, \hat{\tau}) \leq \delta_n.
    \end{equation}
\end{enumerate}
\end{assumption}

We then introduce an additional structural assumption regarding the minimal spacing between true change-points:

\begin{assumption}[Minimal spacing]
\label{assum:minimal_spacing}
We assume that the parameter $\underline{\Lambda}_{\tau^\star}$ satisfies:
\begin{equation}\label{eq:delta_n'}
    n \underline{\Lambda}_{\tau^\star} >
    \begin{cases}
    2\delta_n & \text{if } \kappa^\star > 2, \\
    \delta_n & \text{if } \kappa^\star = 2.
    \end{cases}
\end{equation}
\end{assumption}
\noindent We refer to inequality \eqref{eq:delta_n'} as the separability condition. To simplify notations, we define
\[
\delta_n' := \delta_n \left( 1 + \mathbf{1}_{\kappa^\star > 2} \right).
\]
Under this notation, the inequality \eqref{eq:delta_n'} is equivalent to $n \underline{\Lambda}_{\tau^\star} > \delta_n'$. 
\begin{remark}[On the detection event and assumptions]
Property \eqref{eq:delta_n} allows us to construct valid confidence intervals for the temporal positions of each true change-point. Furthermore, Assumption~\ref{assum:minimal_spacing} ensures that these confidence regions remain well separated, providing sufficient uncontaminated data to perform \emph{post hoc} inference. Specifically, when there are at least two true change-points ($\kappa^\star > 2$), each must be separated from its neighbors by a margin greater than $2\delta_n$ to guarantee that confidence intervals do not overlap. In the simpler case of a single true change-point ($\kappa^\star = 2$), a margin of $\delta_n$ is sufficient to keep the confidence interval within the observation boundaries. 

Finally, we frequently use $\kappa^\star$ in the remainder of this paper to construct temporal grids and conduct inference. Although $\kappa^\star$ is typically an unknown parameter in practice, conditionally on the high-probability event $\Omega_{\alpha_0}$, we have the exact equality $\widehat{\kappa} = \kappa^\star$ given by \eqref{eq:kappa_star}. Consequently, within this event, the true number of change-points behaves as a known quantity that coincides with the estimated one.
\end{remark}

With these structural conditions in place, the following lemma formalizes the localization properties of the estimated change-points, which are central to our \emph{post hoc} analysis.

\medskip

\begin{lemma}\label{lem:IC}
Under Assumptions~\ref{assum:detection_event} and~\ref{assum:minimal_spacing}, the following properties hold on the event $\Omega_{\alpha_0}$:
\begin{enumerate}
    \item Uniqueness: For all $i \in \llbracket 1,\kappa^\star-1 \rrbracket$, there exists a unique $j \in \llbracket 1,\kappa^\star-1 \rrbracket$ such that $|\tau^\star_i - \hat{\tau}_j| \leq \delta_n$.
    
    \item Confidence intervals: For all $i \in \llbracket 1,\kappa^\star-1 \rrbracket$, we have $|\tau^\star_i - \hat{\tau}_i| \leq \delta_n$, which allows us to define the confidence interval $\text{CI}_i$ of the change-point $\tau^\star_i$ as:
    \[
    \text{CI}_i:= \llbracket \hat{\tau}_i - \delta_n, \hat{\tau}_i + \delta_n \rrbracket.
    \]
    
    \item Hausdorff distance equality: 
    \[
    d_H(\tau^\star, \hat{\tau}) = d_{\infty}^{(1)}(\tau^\star, \hat{\tau}) = d_{\infty}^{(1)}(\hat{\tau}, \tau^\star) = \max_{1 \leq i \leq \kappa^\star-1} \left| \tau_i^\star - \hat{\tau}_i \right|.
    \]
    
    \item Separation of estimated points: For all $i \in \llbracket 0,\kappa^\star-1 \rrbracket$, we have
    \[
    |\hat{\tau}_{i+1} - \hat{\tau}_i| \geq n \underline{\Lambda}_{\tau^\star}  - \delta_n'.
    \]
\end{enumerate}
\end{lemma}
The proof is provided in Appendix~\ref{app:proof_IC}.

\subsection{Post hoc attribution for multivariate change-point detection}

In multivariate time series settings, change-point detection procedures typically provide an estimate of the number of change-points $\widehat{\kappa}$ together with their locations $\hat{\tau} =(\hat{\tau}_1,\dots,\hat{\tau}_{\widehat{\kappa}})$, but do not indicate which coordinates of the multivariate time series are responsible for the detected changes.

This lack of interpretability motivates the need for \emph{post hoc} attribution (often referred to as explainability in multivariate settings), especially in applications where understanding the origin of a structural change is crucial. Often, a detected ``global'' change-point may be driven by modifications occurring only in a subset of coordinates, such as specific sensors, economic indicators, or clinical measurements. Moreover, a detected change may involve complex distributional features, such as changes in the covariance structure, rather than marginal effects alone. To the best of our knowledge, existing literature on \emph{post hoc} attribution for change-point detection with error control is limited to mean shifts \cite{Chen2024}. Here, we focus on the more general setting of distributional changes.

In principle, a complete explanation of a detected change would require testing all possible subsets of coordinates in order to identify those responsible for the change. However, such an exhaustive approach leads to combinatorial explosion, since the number of subsets grows exponentially with the dimension. To make the problem tractable, we restrict our analysis to a structured setting where the coordinates of the multivariate time series are partitioned into two predefined blocks. The goal is then to determine, for a given detected change-point, whether the distributional change occurs in the first block, in the second block, or in both.

A first step in addressing this question is to assess the effect of predefined blocks and conduct inference on these blocks. To this end, let $\phi: \mathcal{X} \to \mathcal{Z}$ be a measurable projection onto one of the blocks, where $\mathcal{Z}$ is a measurable space, and define the projected time series $\mathbf{Z}:=(Z_t)_{t \in \llbracket 1,n \rrbracket}$ by:
\[
Z_t := \phi(X_t).
\]
For example, if $\mathcal{X} \subseteq \mathbb{R}^d$ with $d \geq 2$, the mapping $\phi$ may extract a single coordinate or a specific subset of coordinates of the multivariate time series. The attribution problem for the original process $\mathbf{X}$ can then be reformulated as testing for a statistically significant distributional change in $\mathbf{Z}$ at the estimated change-point. This can be assessed by comparing the distribution of $\mathbf{Z}$ before and after the candidate change-point using local two-sample tests. For each $i \in \llbracket 1, \kappa^\star \rrbracket$, we denote the distribution of $Z_t$ for $t \in \llbracket \tau_{i-1}^\star+1, \tau_{i}^\star \rrbracket$ by $P_{(i)}^\phi$. 

We focus on a single true change-point $\tau^\star_{i_0}$ with $i_0 \in \llbracket 1, \kappa^\star - 1 \rrbracket$. A \emph{post hoc} attribution step for this change-point then amounts to determining whether $P_{(i_0)}^\phi \neq P_{(i_0+1)}^\phi$.

A naive approach to test $P_{(i_0)}^\phi \neq P_{(i_0+1)}^\phi$ would be to perform a classical two-sample test directly around the estimated change-point $\hat{\tau}_{i_0}$. However, such a procedure faces two major theoretical challenges: segment contamination and selection bias.

First, to avoid contamination from neighboring true change-points, we must exploit the localization guarantees of the event $\Omega_{\alpha_0}$ and Assumption~\ref{assum:minimal_spacing}. Thanks to Lemma~\ref{lem:IC}, we can isolate two safe, data-driven segments that are guaranteed to lie within the correct true structural segments:
\[
S_{i_0}^-:= \llbracket \hat{\tau}_{i_0-1} + 1 + \delta_n, \hat{\tau}_{i_0} - \delta_n \rrbracket \subseteq \llbracket \tau^\star_{i_0-1}+1, \tau^\star_{i_0} \rrbracket,
\]
where $Z_t \sim P_{(i_0)}^{\phi}$ for all $t \in S_{i_0}^-$, and:
\[
S_{i_0}^+:= \llbracket \hat{\tau}_{i_0} + 1 + \delta_n, \hat{\tau}_{i_0+1} - \delta_n \rrbracket \subseteq \llbracket \tau^\star_{i_0}+1, \tau^\star_{i_0+1} \rrbracket,
\]
where $Z_t \sim P_{(i_0+1)}^{\phi}$ for all $t \in S_{i_0}^+$.

Second, although constructing these segments successfully resolves the contamination issue, the selection bias remains. Because the boundaries of $S_{i_0}^-$ and $S_{i_0}^+$ depend directly on the estimated change-points $\hat{\tau}$, these segments are non-deterministic. This data-dependent selection step induces a critical dependency between the sampling windows and the observed data, which completely invalidates the Type I error control of standard classical two-sample tests. 

This fundamental limitation motivates the development of specialized \emph{post hoc} inference procedures capable of conditioning on this selection event, which we present in the following section.

\begin{remark}
\leavevmode
If $\delta_n = 0$, the change-point estimation is exact, meaning the segments $S_{i_0}^-$ and $S_{i_0}^+$ become deterministic and the selection bias vanishes. In all that follows, we assume the realistic case where $\delta_n > 0$.
\end{remark}

\section{Methods (post hoc procedures)}
\label{sec:methods}
\subsection{Grid-based Two-Sample Test (GTST)}
\label{subsec:GTST}
\subsubsection{Method description}

To statistically evaluate the change occurring at a true change-point $\tau_{i_0}^\star$, one would ideally perform a two-sample test on the exact preceding and succeeding segments. In practice, however, these true boundaries are unknown. We only have access to their estimates $\widehat\tau_{i_0-1}$, $\widehat\tau_{i_0}$, and $\widehat\tau_{i_0+1}$, each subject to a localization error bounded by $\delta_n$. Intuitively, one could handle this uncertainty by exhaustively exploring all possible segment boundaries within the $\delta_n$-neighborhoods of these estimates. However, running a two-sample test for every single valid configuration requires up to $(2\delta_n+1)^3$ tests per change-point. As $\delta_n$ grows, this exhaustive search quickly becomes computationally prohibitive. To reduce this computational burden while maintaining theoretical validity, we introduce a discretized approach. By restricting the candidate boundaries to a predefined grid, we can capture the necessary local behavior with a fraction of the computational cost.

To this end, we introduce a discretization parameter $\eta \in 
\llbracket 1,\, 2\delta_n \rrbracket$, referred to as the \emph{grid step}, and restrict all candidate segment boundaries to integer multiples of~$\eta$.
This grid-based construction drastically reduces the number of admissible segments, and therefore the total number of two-sample tests to be performed,
while still exploring all configurations that are relevant at the scale~$\delta_n$.
Let
\[
m:= \left\lceil \frac{2\delta_n}{\eta} \right\rceil ,
\]
which determines the number of grid points used to parametrize the segments.
The formal definition of this discretized family is as follows.
For $\beta,\gamma,\zeta\in\{0,\dots,m\}$, we define:
\begin{align*}
a_{\beta}(\eta)
&:= \Biggl(
      \Biggl\lfloor
        \frac{\widehat\tau_{i_0-1}}{\eta}
        - \frac{\delta_n}{\eta}\,\mathbf 1_{\{i_0-1\neq 0\}}
      \Biggr\rfloor
      + (1+\beta)\,\mathbf 1_{\{i_0-1\neq 0\}}
    \Biggr)\eta
    + \mathbf 1_{\{i_0-1=0\}}, \\[0.5ex]
b_{\gamma}(\eta)
&:= \Biggl(
      \Biggl\lfloor
        \frac{\widehat\tau_{i_0}}{\eta}
        - \frac{\delta_n}{\eta}
      \Biggr\rfloor
      + \gamma
    \Biggr)\eta, \\[0.5ex]
c_{\zeta}(\eta)
&:= \max\Biggl(
      \min\Biggl(
        \Biggl(
          \Biggl\lfloor
            \frac{\widehat\tau_{i_0+1}}{\eta}
            - \frac{\delta_n}{\eta}\,\mathbf 1_{\{i_0+1\neq \kappa^\star\}}
          \Biggr\rfloor
          + \zeta\,\mathbf 1_{\{i_0+1\neq \kappa^\star\}}
        \Biggr)\eta,
        \, n
      \Biggr),
      \, n\,\mathbf 1_{\{i_0+1=\kappa^\star\}}
    \Biggr).
\end{align*}
Then, we define:
\[
\widetilde{\mathcal{E}}_{i_0,\eta} := \Bigl\{ \bigl(\widetilde{S}^{-}_{i_0,\eta}, \widetilde{S}^{+}_{i_0,\eta}\bigr) = \bigl( \llbracket a_{\beta}(\eta), b_{\gamma}(\eta) \rrbracket, \llbracket b_{\gamma}(\eta)+\eta, c_{\zeta}(\eta) \rrbracket \bigr), \text{ with } (\beta, \gamma, \zeta) \in \llbracket 0, m \rrbracket^3 \Bigr\}.
\]
We also define the projected oracle segments $S^{\star-}_{i_0}$ and $S^{\star+}_{i_0}$ on the grid as: 
\[
S^{\star-}_{i_0,\eta}:= \left[\!\!\left[\left(\left\lfloor \frac{\tau^\star_{i_0-1}}{\eta} \right\rfloor + \mathbf{1}_{\{i_0 - 1 \neq 0\}} \right) \eta + \mathbf{1}_{\{i_0 - 1 = 0\}}, \left\lfloor \frac{\tau^\star_{i_0}}{\eta} \right\rfloor  \eta\right]\!\!\right]
\]
and
\[
S^{\star+ }_{i_0, \eta}:= \left[\!\!\left[\left(\left\lfloor \frac{\tau^\star_{i_0}}{\eta} \right\rfloor + 1 \right) \eta, \max\left(\left\lfloor \frac{\tau^\star_{i_0 + 1}}{\eta} \right\rfloor  \eta \, , \, n \cdot \mathbf{1}_{\{i_0 + 1 = \kappa^\star\}}\right)\right]\!\!\right].
\]
To provide a comprehensive overview of our approach, Figure~\ref{fig:gtst_full_workflow} illustrates the complete step-by-step workflow of the GTST procedure.
\begin{figure}[tbp]
\centering
\resizebox{\textwidth}{!}{
\begin{tikzpicture}[x=0.85cm,y=1cm]
  \definecolor{dgreen}{RGB}{37,194,105}
  \definecolor{pblue}{RGB}{50,100,200}
  \definecolor{gtstpurple}{RGB}{142,68,173}

  \def\etav{2.0}

  \def\tmstar{2}   
  \def\tstar{9}    
  \def\tpstar{16}  

  \def\thm{3}      
  \def\th{9}       
  \def\thp{15}     

  \def\deltan{2}   

  \draw[draw=gray!50, thick, rounded corners, fill=gray!5] (-2.5, 5.1) rectangle (19.5, 7.3); 
  \draw[draw=gray!50, thick, rounded corners, fill=gray!5] (-2.5, 3.1) rectangle (19.5, 4.9); 
  \draw[draw=gray!50, thick, rounded corners, fill=gray!5] (-2.5, 1.1) rectangle (19.5, 2.9); 
  \draw[draw=gray!50, thick, rounded corners, fill=gray!5] (-2.5, -1.6) rectangle (19.5, 0.9);  

  \node[left, font=\bfseries] at (-0.2,6) {Oracle};
  \draw[->] (0,6) -- (19,6);

  \foreach \x in {0,1,...,19} {
    \draw[gray!60, thin] (\x, 5.85) -- (\x, 6.15);
  }

  \draw[dgreen, line width=4pt] (\tmstar+1,6) -- (\tstar,6) node[midway,above=4pt, text=black] {$S^{\star-}_{i_0}$};
  \draw[dgreen, line width=4pt] (\tstar+1,6) -- (\tpstar,6) node[midway,above=4pt, text=black] {$S^{\star+}_{i_0}$};

  \draw[dashed, thick, dgreen] (\tmstar,6.6) -- (\tmstar,5.4); 
  \node[above left, inner sep=2pt] at (\tmstar,6.6) {$\tau^{\star}_{i_0-1}$};
  
  \draw[dashed, thick, dgreen] (\tmstar+1,6.6) -- (\tmstar+1,5.4); 
  \node[above right, inner sep=2pt] at (\tmstar+1,6.6) {$\tau^{\star}_{i_0-1}+1$};

  \draw[dashed, thick, dgreen] (\tstar,6.6) -- (\tstar,5.4);   
  \node[above left, inner sep=2pt] at (\tstar,6.6) {$\tau^{\star}_{i_0}$};
  
  \draw[dashed, thick, dgreen] (\tstar+1,6.6) -- (\tstar+1,5.4); 
  \node[above right, inner sep=2pt] at (\tstar+1,6.6) {$\tau^{\star}_{i_0}+1$};
  
  \draw[dashed, thick, dgreen] (\tpstar,6.6) -- (\tpstar,5.4); 
  \node[above=2pt] at (\tpstar,6.6) {$\tau^{\star}_{i_0+1}$};

  \node[left, font=\bfseries] at (-0.2,4) {Proj. oracle};
  \draw[->] (0,4) -- (19,4);

  \foreach \x in {0,2,...,18} {
    \draw[gray!90, thick] (\x, 3.8) -- (\x, 4.2);
  }
  \node[right, text=gray] at (18.5, 4.3) {Grid ($\eta=2$)};

  \def\prm{4}    
  \def\prto{8}   
  \def\prtpo{10} 
  \def\prp{16}   

  \draw[pblue, line width=4pt] (\prm,4) -- (\prto,4) node[midway,above=4pt, text=black] {$S^{\star-}_{i_0, \eta}$};
  \draw[pblue, line width=4pt] (\prtpo,4) -- (\prp,4) node[midway,above=4pt, text=black] {$S^{\star+}_{i_0, \eta}$};

  \draw[->, pblue, thin, dashed] (\tmstar+1, 5.8) -- (\prm, 4.2);
  \draw[->, pblue, thin, dashed] (\tstar, 5.8) -- (\prto, 4.2);
  \draw[->, pblue, thin, dashed] (\tstar+1, 5.8) -- (\prtpo, 4.2);
  \draw[->, pblue, thin, dashed] (\tpstar, 5.8) -- (\prp, 4.2);

  \node[left, font=\bfseries] at (-0.2,2) {Data};
  \draw[->] (0,2) -- (19,2);

  \foreach \x in {0,1,...,19} {
    \draw[gray!60, thin] (\x, 1.85) -- (\x, 2.15);
  }

  \draw[red, very thick, |-|] (\thm-\deltan,2.3) -- (\thm+\deltan,2.3) node[midway,above=2pt, scale=0.85, text=black] {$[\widehat\tau_{i_0-1}\pm\delta_n]$};
  \draw[red, very thick, |-|] (\th-\deltan,2.3) -- (\th+\deltan,2.3) node[midway,above=2pt, scale=0.85, text=black] {$[\widehat\tau_{i_0}\pm\delta_n]$};
  \draw[red, very thick, |-|] (\thp-\deltan,2.3) -- (\thp+\deltan,2.3) node[midway,above=2pt, scale=0.85, text=black] {$[\widehat\tau_{i_0+1}\pm\delta_n]$};

  \draw[thick] (\thm,2.5) -- (\thm,1.5) node[below] {$\widehat\tau_{i_0-1}$};
  \draw[thick] (\th,2.5) -- (\th,1.5) node[below] {$\widehat\tau_{i_0}$};
  \draw[thick] (\thp,2.5) -- (\thp,1.5) node[below] {$\widehat\tau_{i_0+1}$};

  \draw[dotted, thick, gray] (\tmstar,5.5) -- (\tmstar,2.5);
  \draw[dotted, thick, gray] (\tstar,5.5) -- (\tstar,2.5);
  \draw[dotted, thick, gray] (\tpstar,5.5) -- (\tpstar,2.5);

  \node[left, font=\bfseries] at (-0.2,0) {GTST};
  \node[left, font=\bfseries] at (-0.2,-0.4) {candidates};
  \draw[->] (0,0) -- (19,0);

  \foreach \x in {0,2,...,18} {
    \draw[gray!90, thick] (\x, -0.2) -- (\x, 0.2);
  }
  \node[right, text=gray] at (18.5, 0.3) {Grid ($\eta=2$)};

  \def\a{2}            
  \def\endZoneOne{6}   

  \def\b{8}            
  \def\bp{10}          
  \def\endZoneTwo{12}  
  
  \def\startZoneThree{14} 
  \def\c{18}           

  \fill[red!10, opacity=0.5] (\a, -0.8) rectangle (\endZoneOne, 2.2);
  \fill[red!10, opacity=0.5] (\b, -0.8) rectangle (\endZoneTwo, 2.2);
  \fill[red!10, opacity=0.5] (\startZoneThree, -0.8) rectangle (\c, 2.2); 

  \draw[gtstpurple, line width=4pt] (\a,0) -- (\b,0) node[midway,below=4pt, text=black] {$\widetilde S^-_{i_0, \eta}$};
  \draw[gtstpurple, line width=4pt] (\bp,0) -- (\c,0) node[midway,below=4pt, text=black] {$\widetilde S^+_{i_0, \eta}$};

  \draw[gtstpurple, line width=2pt] (\a,0.25)  -- (\a,-0.25);
  \draw[gtstpurple, line width=2pt] (\b,0.25)  -- (\b,-0.25);
  \draw[gtstpurple, line width=2pt] (\bp,0.25) -- (\bp,-0.25);
  \draw[gtstpurple, line width=2pt] (\c,0.25)  -- (\c,-0.25);

  \node[below, text=gtstpurple, yshift=-12pt] at (\a,-0.25) {$a_\beta(\eta)$};
  \node[below, text=gtstpurple, yshift=-12pt] at (\c,-0.25) {$c_\zeta(\eta)$};
  \node[below, xshift=-16pt, text=gtstpurple, yshift=-12pt] at (\b,-0.25) {$b_\gamma(\eta)$};
  \node[below, xshift=+22pt, text=gtstpurple, yshift=-12pt] at (\bp,-0.25) {$b_\gamma(\eta)+\eta$};

  \draw[->, thick, dashed, red, shorten >=2pt, shorten <=2pt] (\thm-\deltan, 2.2) to[out=-60, in=120] (\a, 0.25);
  \draw[->, thick, dashed, red, shorten >=2pt, shorten <=2pt] (\thm+\deltan, 2.2) to[out=-60, in=120] (\endZoneOne, 0.25);

  \draw[->, thick, dashed, red, shorten >=2pt, shorten <=2pt] (\th-\deltan, 2.2) to[out=-60, in=120] (\b, 0.25);
  \draw[->, thick, dashed, red, shorten >=2pt, shorten <=2pt] (\th+\deltan, 2.2) to[out=-60, in=120] (\endZoneTwo, 0.25);

  \draw[->, thick, dashed, red, shorten >=2pt, shorten <=2pt] (\thp-\deltan, 2.2) to[out=-60, in=120] (\startZoneThree, 0.25);
  \draw[->, thick, dashed, red, shorten >=2pt, shorten <=2pt] (\thp+\deltan, 2.2) to[out=-60, in=120] (\c, 0.25);

\end{tikzpicture}
}
\caption{Complete workflow of the GTST procedure. \textbf{Top (Oracle):} The true segments bounded by the unknown change-points. \textbf{Second (Proj. Oracle):} The theoretical ideal segments are projected onto the grid $\eta$. \textbf{Third (Data):} The change-points are estimated, defining uncertainty neighborhoods of radius $\delta_n$ (red intervals). \textbf{Bottom (GTST Candidates):} The practical test segments are formed by evaluating discrete boundaries ($a, b, b+\eta, c$) constrained to the grid $\eta$, which are derived from the red uncertainty neighborhoods.}
\label{fig:gtst_full_workflow}
\end{figure}
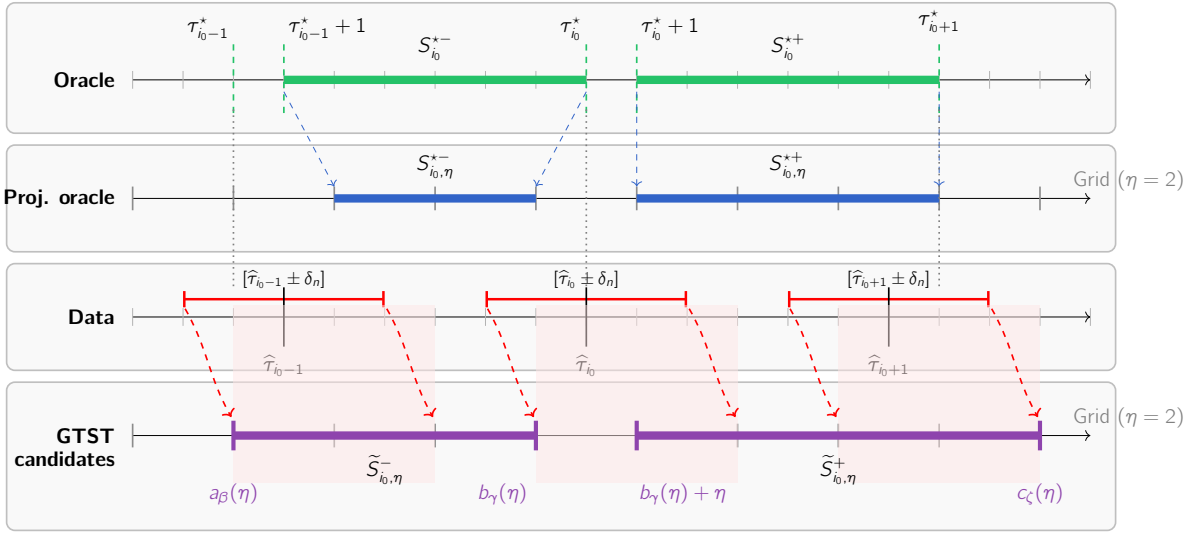

\subsubsection{Theoretical analysis}
\begin{lemma}\label{lem:theorM1_bis}
We assume that Assumptions~\ref{assum:detection_event} and~\ref{assum:minimal_spacing} hold. 
Define
\[
\ell_\eta:= \left(\left\lfloor\frac{n\underline{\Lambda}_{\tau^\star}}{\eta}\right\rfloor - 1\right)\eta +1,
\qquad
\ell_\eta^-:= \left\lfloor\frac{n\underline{\Lambda}_{\tau^\star}}{\eta}\right\rfloor \eta.
\]
Then, on $\Omega_{\alpha_0}$, the following hold:
\begin{enumerate}
\item The projected segments remain fully contained within the true oracle segments:
\[
S^{\star-}_{i_0,\eta} \subseteq S^{\star-}_{i_0},
\qquad 
S^{\star+}_{i_0,\eta} \subseteq S^{\star+}_{i_0}.
\]

\item The projected segments have a guaranteed minimal length: \\
if $i_0>1$, then 
\[
|S^{\star-}_{i_0,\eta}|,\ |S^{\star+}_{i_0,\eta}| \ \ge \ \ell_\eta,
\]
whereas for $i_0=1$,
\[
|S^{\star-}_{i_0,\eta}| \ge \ell_\eta^-,
\qquad
|S^{\star+}_{i_0,\eta}| \ge \ell_\eta.
\]

\item The pair $(S^{\star-}_{i_0,\eta},S^{\star+}_{i_0,\eta})$ belongs to the
restricted grid-set
\[
\begin{aligned}
\widetilde{\underline{\mathcal{E}}}_{i_0,\eta}
:=\Bigl\{
  & (\widetilde{S}^-_{i_0,\eta},\widetilde{S}^+_{i_0,\eta})\in \widetilde{\mathcal{E}}_{i_0,\eta}: \\
  & |\widetilde{S}^-_{i_0,\eta}|\ge \max\left(\ell_{\eta}^-\cdot\mathbf{1}_{\{i_0 - 1 = 0\}} \, , \, \ell_{\eta} \cdot \mathbf{1}_{\{i_0 - 1 \neq 0\}} \right) \ \mathrm{ and } \  |\widetilde{S}^+_{i_0,\eta}|\ge \ell_\eta
\Bigr\}.
\end{aligned}
\]
\end{enumerate}
\end{lemma}
The proof is provided in Appendix~\ref{app:proof_theorM1_bis}.

\begin{remark}[Lower bound for the parameter $\ell_\eta$]
The parameter $\ell_{\eta}$ may be replaced by any valid lower bound. For example, one could use:
\[
\ell_{\eta} \geq \widetilde{\ell}_{\eta} := \left(\left\lfloor\frac{\delta_n'}{\eta} \right\rfloor - 1\right)\eta + 1.
\]
\end{remark}

\vspace{0.3cm}
\noindent Since the projected segments $S^{\star-}_{i_0, \eta}$ and $S^{\star+}_{i_0, \eta}$ are deterministic and contained within the true oracle segments, the observations within each of these regions remain independent and identically distributed under the null hypothesis. This guarantees the validity of applying a local two-sample test. We define the global test statistic for GTST as follows:
\[
T_{\alpha_1}^{\mathcal{Z},\mathrm{GTST}}(\widetilde{\underline{\mathcal{E}}}_{i_0, \eta}):= \prod_{\substack{(\widetilde{S}_{i_0, \eta}^-, \widetilde{S}_{i_0, \eta}^+) \in \widetilde{\underline{\mathcal{E}}}_{i_0, \eta}}}^{} T_{\alpha_1}^{\mathcal{Z}} \left((Z_i)_{i \in \widetilde{S}^{-}_{i_0, \eta}}, (Z_i)_{i \in \widetilde{S}^{+}_{i_0, \eta}}\right),
\]
where $T_{\alpha_1}^{\mathcal{Z}}(\mathbf{Z}_A, \mathbf{Z}_B) \in \{0, 1\}$ denotes any generic level-$\alpha_1$ two-sample test taking two $\mathcal{Z}$-valued time series as input. Specifically, for any two mutually independent sequences $\mathbf{Z}_A$ and $\mathbf{Z}_B$ of i.i.d.\ random variables with marginal distributions $P_A$ and $P_B$ respectively, $T_{\alpha_1}^{\mathcal{Z}}$ tests the local null hypothesis
\[
H_0^{\mathrm{local}}: P_A = P_B \quad \text{against} \quad H_1^{\mathrm{local}}: P_A \neq P_B.
\]
The test returns $1$ to reject the null hypothesis, and the level $\alpha_1$ condition formally means that $\mathbb{P}_{H_0^{\mathrm{local}}}\left(T_{\alpha_1}^{\mathcal{Z}}(\mathbf{Z}_A, \mathbf{Z}_B) = 1\right) \leq \alpha_1.$ By convention, whenever one of the two input segments is empty, the test returns $0$, i.e.\ the null hypothesis is retained by default.

\begin{remark}[Influence of the grid step on the number of evaluations]
For any grid step $\eta \in [1, 2\delta_n]$, the GTST procedure requires performing at most $(m+1)^3$ local tests, where $m = \lceil 2\delta_n / \eta \rceil$. The choice of $\eta$ therefore controls a direct trade-off between computational cost and grid resolution. For instance, setting $\eta = 2\delta_n$ yields $m = 1$, requiring at most $2^3 = 8$ tests but relying on a very coarse discretization. Conversely, choosing the finest grid step $\eta = 1$ gives $m = 2\delta_n$, meaning no discretization is applied. In this case, the procedure exhaustively explores all valid configurations, resulting in up to $(2\delta_n + 1)^3$ tests to perform.
\end{remark}


To synthesize the preceding construction, Algorithm~\ref{algo:gtst} summarizes the complete practical workflow of the GTST procedure, from the initial data-driven change-point estimation to the final global decision. The theoretical validity of this methodology is then formally established in the following theorem.

\begin{algorithm}[tbp]
\caption{Grid-based two-sample test (GTST)}
\label{algo:gtst}
\textbf{Input:} Time series $(X_1,\ldots,X_n)$, number of segments $\kappa^\star$, index $i_0$, level $\alpha_1$, grid step $\eta$, generic two-sample test $T_{\alpha_1}^{\mathcal Z}$.\\
\textbf{Output:} Test decision.

\begin{enumerate}
\item Estimate the change-points in the time series $(X_1,\ldots,X_n)$ assuming $\kappa^\star$ segments, and denote the resulting segmentation by $\hat{\tau}$.

\item Compute the theoretical localization error bound $\delta_{n}$ and the estimated minimal spacing $\underline{\Lambda}_{\widehat{\tau}}$.

\item If the separability condition is not satisfied, i.e.\ $n\,\underline{\Lambda}_{\widehat{\tau}} \le \delta'_{n}$, return $T_{\alpha_1}^{\mathcal Z, \mathrm{GTST}} = 0$ (retain $H_0$) and stop.

\item Construct the collection $\widetilde{\underline{\mathcal{E}}}_{i_0,\eta}$ of candidate pairs of segments on the grid around $\hat{\tau}_{i_0}$.

\item For each $(S^-,S^+) \in \widetilde{\underline{\mathcal{E}}}_{i_0,\eta}$, compute the local two-sample test:
\[
T_{\alpha_1}^{\mathcal Z}\big((Z_i)_{i\in S^-},(Z_i)_{i\in S^+}\big).
\]

\item Reject the global null hypothesis if all local tests in the collection reject, that is, define:
\[
T_{\alpha_1}^{\mathcal{Z},\mathrm{GTST}}(\widetilde{\underline{\mathcal{E}}}_{i_0, \eta})
=
\prod_{(S^-,S^+) \in \widetilde{\underline{\mathcal{E}}}_{i_0,\eta}}
T_{\alpha_1}^{\mathcal Z}\big((Z_i)_{i\in S^-},(Z_i)_{i\in S^+}\big).
\]

\item Return $T_{\alpha_1}^{\mathcal Z, \mathrm{GTST}}$.
\end{enumerate}
\end{algorithm}

\begin{theorem}
\label{thm:GTST}
Suppose that Assumptions~\ref{assum:detection_event} and~\ref{assum:minimal_spacing} hold. Then $T_{\alpha_1}^{\mathcal{Z},\mathrm{GTST}}$ is a test of
\[
H_0^{\phi}: P_{(i_0)}^{\phi} = P_{(i_0+1)}^{\phi} \quad \text{against} \quad H_1^{\phi}: P_{(i_0)}^{\phi} \neq P_{(i_0+1)}^{\phi},
\]
with level $ \alpha_0 + \alpha_1 $.
\end{theorem}
\begin{proof}[\textbf{Proof.}]
\begin{align*}
\mathbb{P}_{H_0^{\phi}}\left(T_{\alpha_1}^{\mathcal{Z},\mathrm{GTST}}\left(\widetilde{\underline{\mathcal{E}}}_{i_0, \eta}\right) = 1\right) &\leq \mathbb{P}\left(\Omega_{\alpha_0}^c\right) + \mathbb{P}_{H_0^{\phi}}\left(\Omega_{\alpha_0} \cap\left\{T_{\alpha_1}^{\mathcal{Z},\mathrm{GTST}}\left(\widetilde{\underline{\mathcal{E}}}_{i_0, \eta}\right) = 1\right\}\right)\\
&\leq \alpha_0 + \mathbb{P}_{H_0^{\phi}}\left(T_{\alpha_1}^{\mathcal{Z}}\left((Z_i)_{i \in S^{\star-}_{i_0, \eta}}, (Z_i)_{i \in S^{\star+}_{i_0, \eta}}\right) = 1\right) \quad  \text{(Lemma \ref{lem:theorM1_bis})}\\
&\leq \alpha_0 + \alpha_1.
\end{align*}
For the last inequality, we use the fact that under $H_0^{\phi}$, the $Z_i$ with $i \in S^{\star-}_{i_0, \eta} \cup S^{\star+ }_{i_0, \eta}$ are i.i.d.
\end{proof}

\begin{remark}[On the Type I error and minimal spacing]
\label{rem:minimal_spacing}
The Type~I error of the \emph{post hoc} GTST procedure decomposes into two components: $\alpha_0$ for the initial mislocalization risk, and $\alpha_1$ for the conditional two-sample false alarm rate, given a correct localization.

If Condition~\eqref{eq:delta_n'} is not satisfied, the GTST procedure conservatively retains the null hypothesis $H_0$ by default. This behavior highlights the critical importance of this assumption: without sufficient spacing, retaining $H_0$ might simply result from the violated condition rather than a true absence of distributional change, making empirical performance impossible to evaluate fairly.

In practice, the separability condition~\eqref{eq:delta_n'} cannot be 
checked directly, since $\underline{\Lambda}_{\tau^\star}$ is unknown; 
Step~3 of Algorithm~\ref{algo:gtst} therefore uses the plug-in estimate 
$\Lambda_{\hat{\tau}}$.
\end{remark}

Having demonstrated that valid inference can be achieved using the entire sequence of observations, we now present an alternative and equally valid approach based on data splitting. As we will show, this strategy provides analogous theoretical guarantees.

\subsection{Hold-out Method}
\label{subsec:Holdout}
\subsubsection{Method description}
While data splitting is a standard statistical tool \cite{Cox_1975}, its properties within the specific context of change-point attribution require careful analysis. The hold-out approach we consider consists of partitioning the indices of the time series $Z_1, \dots, Z_n$ into two disjoint subsets. The change-point localization phase is performed exclusively on the sub-series with indices in the set $I_{\mathrm{det}}$, which we refer to as the ``detection set'', with its elements called ``detection indices''. The complement, $I_{\mathrm{inf}} = \llbracket 1, n \rrbracket \setminus I_{\mathrm{det}}$, which we refer to as the ``inference set'', is reserved for the statistical testing of attribution. By decoupling the detection and testing phases, this strategy provides valid inference guarantees, which we establish in this section. By construction, $I_{\mathrm{det}} \cap I_{\mathrm{inf}} = \varnothing$ and $I_{\mathrm{det}} \cup I_{\mathrm{inf}} = \llbracket 1, n \rrbracket$. We denote their respective sizes by $n_{\mathrm{det}} = |I_{\mathrm{det}}|$ and $n_{\mathrm{inf}} = |I_{\mathrm{inf}}|$, such that $n_{\mathrm{det}} + n_{\mathrm{inf}} = n$.
For example, a simple even/odd partition corresponds to $n_{\mathrm{det}} \approx n/2$ and $n_{\mathrm{inf}} \approx n/2$, with $I_{\mathrm{det}} = \{1,3,5,\dots\} $ and $I_{\mathrm{inf}} =  \{2,4,6,\dots\} $. More generally, the procedure imposes no specific regularity assumption on $I_{\mathrm{det}}$ at this stage: the partition may be randomly generated or constructed according to a fixed scheme, provided it satisfies the structural conditions detailed below.

Let $j_1 < j_2 < \cdots < j_{n_{\mathrm{det}}}$ denote the elements of $I_{\mathrm{det}}$, sorted in increasing order. For the hold-out procedure to be valid, it is necessary that $I_{\mathrm{det}}$ contains at least one point in each segment $\llbracket \tau^\star_r + 1,\; \tau^\star_{r+1} \rrbracket$, for $0 \le r \le \kappa^\star - 1$. We denote
\[
I_{\mathrm{det}}^{(r)} \;:=\; I_{\mathrm{det}} \,\cap\, \llbracket \tau^\star_r + 1,\; \tau^\star_{r+1}\rrbracket
\]
the portion of segment $r$ included in the detection set.

\begin{assumption}[Minimal coverage condition.]
\label{assum:Holdoutseg}
The success of the hold-out approach relies heavily on how the observations are partitioned. Specifically, to prevent missing true change-points, every true segment must be represented in the data used for localization. The minimal requirement on the detection set $I_{\mathrm{det}}$ is thus expressed as:
\begin{equation}\label{eq:Holdoutseg}
\min_{0\leq r \leq \kappa^\star-1} \, \mathrm{Card}\left(I_{\mathrm{det}}^{(r)} \right) =: n_{\mathrm{det}} \, \underline{\Lambda}_{\tau^{\star}}^{\mathrm{det}} > 0.
\end{equation}
\end{assumption}
This condition ensures that there is at least one observation in the detection set between any two consecutive true change-points. As illustrated in Figure~\ref{fig:HoldOut}, a valid partition successfully meets this requirement, whereas a bad partition leaves an entire segment empty, making the detection of the surrounding change-points impossible.



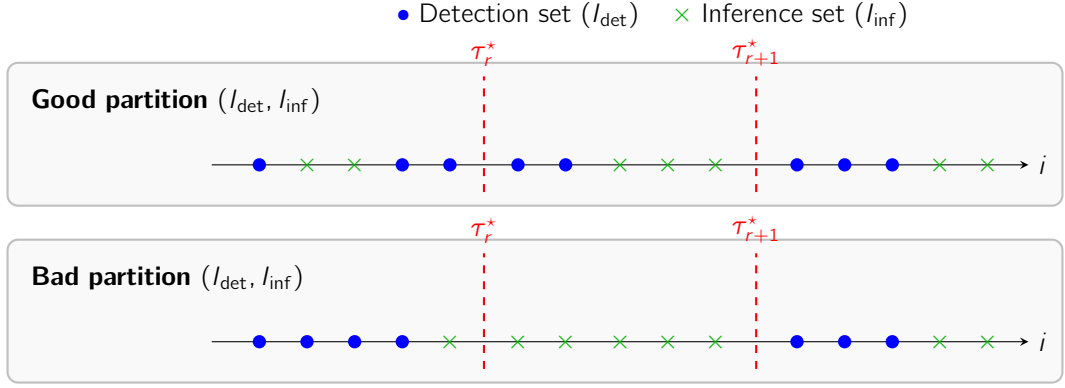
\begin{figure}[tbp]
\centering
\begin{tikzpicture}[x=0.9cm,y=0.9cm,>=stealth]

\node at (4.5, 3.2) {\textcolor{blue}{$\bullet$} Detection set ($I_{\mathrm{det}}$)};
\node at (8.5, 3.2) {\textcolor{green!70!black}{$\times$} Inference set ($I_{\mathrm{inf}}$)};

\draw[draw=gray!50, thick, rounded corners, fill=gray!5] (-3, 0.4) rectangle (12.5, 2.5);

\draw[->] (0,1) -- (12,1) node[right] {$i$};

\node[above right, font=\bfseries] at (-2.8, 1.6) {Good partition $(I_{\mathrm{det}}, I_{\mathrm{inf}})$};

\def\tone{4}
\def\ttwo{8}
\draw[red,dashed, thick] (\tone,0.6) -- (\tone,2.3) node[above] {$\tau^\star_r$};
\draw[red,dashed, thick] (\ttwo,0.6) -- (\ttwo,2.3) node[above] {$\tau^\star_{r+1}$};

\foreach \x in {0.7,2.8,3.5,4.5,5.2,8.6,9.3,10.0} {
  \fill[blue] (\x,1) circle (2.5pt);
}
\foreach \x in {1.4,2.1,6.0,6.7,7.4,10.7,11.4} {
  \draw[green!70!black,thick] (\x,1) node {$\times$};
}

\draw[draw=gray!50, thick, rounded corners, fill=gray!5] (-3, -2.2) rectangle (12.5, -0.1);

\draw[->] (0,-1.6) -- (12,-1.6) node[right] {$i$};

\node[above right, font=\bfseries] at (-2.8, -1.0) {Bad partition $(I_{\mathrm{det}}, I_{\mathrm{inf}})$};

\draw[red,dashed, thick] (\tone,-2.0) -- (\tone,-0.3) node[above] {$\tau^\star_r$};
\draw[red,dashed, thick] (\ttwo,-2.0) -- (\ttwo,-0.3) node[above] {$\tau^\star_{r+1}$};

\foreach \x in {0.7,1.4,2.1,2.8,8.6,9.3,10.0} {
  \fill[blue] (\x,-1.6) circle (2.5pt);
}
\foreach \x in {3.5,4.5,5.2,6.0,6.7,7.4,10.7,11.4} {
  \draw[green!70!black,thick] (\x,-1.6) node {$\times$};
}

\end{tikzpicture}
\caption{Illustration of the minimal coverage condition: in a valid partition (top), the condition is met: there is at least one detection point (blue dots) between the two consecutive true change-points $\tau^\star_r$ and $\tau^\star_{r+1}$. In a bad partition (bottom), this condition is violated: the entire segment between $\tau^\star_r$ and $\tau^\star_{r+1}$ only contains inference points (green crosses). The absence of detection points in this region makes it impossible to localize the change-point $\tau^\star_r$.}
\label{fig:HoldOut}
\end{figure}


\noindent For theoretical results, we define $\pi^+(i)$ as the nearest detection index to the right of~$i$ used for localization, and
$\pi^-(i)$ as the nearest one to the left:  
\[
\pi^-(i):=\max\{j\in I_{\mathrm{det}}\cup \{0\}\,:\,j\le i\},
\qquad
\pi^+(i):=\min\{j\in I_{\mathrm{det}} \cup \{n\}\,:\,j\ge i\},
\]
We have 
\[
I_{\mathrm{det}}^{(r)} = I_{\mathrm{det}} \cap \llbracket \tau_r^{\star}+1, \tau_{r+1}^{\star} \rrbracket \neq \emptyset,
\]
by the minimal coverage condition \eqref{eq:Holdoutseg}. Its first element is $\pi^+(\tau_r^\star+1)$, that is,
$\pi^+(\tau_r^\star+1)=\min(I_{\mathrm{det}}^{(r)})$ and its last element is $\pi^-(\tau_{r+1}^\star)$, that is,
$\pi^-(\tau_{r+1}^\star)=\max(I_{\mathrm{det}}^{(r)})$. Moreover, since $I_{\mathrm{det}}^{(r)} \subseteq \llbracket \tau_r^{\star}+1, \tau_{r+1}^{\star} \rrbracket$ and the oracle segments $S^\star_r = \llbracket \tau_r^{\star}+1, \tau_{r+1}^{\star} \rrbracket$ are disjoint and ordered, it follows that the projections $\pi^+(\tau_r^\star+1)$ are monotone. This monotonicity is proved in Lemma~\ref{lem:pi+} below. As illustrated in Figure~\ref{fig:It}, these projections partition the detection set into disjoint subsets $I_{\mathrm{det}}^{(r)}$, contained within their corresponding true segments.

\begin{figure}[tbp]
  \centering
  \begin{tikzpicture}[x=0.9cm,y=0.9cm,>=stealth]

    \draw[draw=gray!50, thick, rounded corners, fill=gray!5] (-0.5, -2.5) rectangle (14.5, 4.0);

    \node[above right, font=\bfseries] at (0, 3.2) {\textcolor{blue}{$\bullet$} Detection indices ($I_{\mathrm{det}}$)};

    \draw[->] (0,0) -- (14,0) node[right] {$i$};

    \def\trP{1}      
    \def\trp{6}      
    
    \def\trpP{7}     
    \def\trpp{13}    

    \draw[thick] (\trP, 0) node[below=2pt] {$\tau_r^\star+1$} -- ++(0,0.2);
    \draw[thick] (\trp, 0) node[below=2pt] {$\tau_{r+1}^\star$} -- ++(0,0.2);
    \draw[thick] (\trpP,0) node[below=2pt] {$\tau_{r+1}^\star+1$} -- ++(0,0.2);
    \draw[thick] (\trpp,0) node[below=2pt] {$\tau_{r+2}^\star$} -- ++(0,0.2);

    \draw[decorate,decoration={brace,mirror,raise=0.7cm}, thick]
      (\trP,0) -- (\trp,0) node[midway,below=1.0cm] {True segment $S^\star_r$};
    \draw[decorate,decoration={brace,mirror,raise=0.7cm}, thick]
      (\trpP,0) -- (\trpp,0) node[midway,below=1.0cm] {True segment $S^\star_{r+1}$};

    \def\ar{2}
    \def\itrA{4}
    \def\itrB{5}
    
    \def\arp{8}
    \def\itrpA{10}
    \def\itrpB{12}

    \foreach \x in {\ar, \itrA, \itrB, \arp, \itrpA, \itrpB}{
      \fill[blue] (\x, 0) circle (2.5pt);
    }

    \draw[dashed, blue, thick] (\ar, 0) -- (\ar, 1.2);
    \draw[dashed, blue, thick] (\arp, 0) -- (\arp, 1.2);

    \node[above, text=blue] at (\ar, 1.2) {$\pi^+(\tau_r^\star+1)$};
    \node[above, text=blue] at (\arp, 1.2) {$\pi^+(\tau_{r+1}^\star+1)$};

    \draw[dashed, blue, thick] (\itrB, 0) -- (\itrB, 1.2);
    \draw[dashed, blue, thick] (\itrpB, 0) -- (\itrpB, 1.2);

    \node[above, text=blue] at (\itrB, 1.2) {$\pi^-(\tau_{r+1}^\star)$};
    \node[above, text=blue] at (\itrpB, 1.2) {$\pi^-(\tau_{r+2}^\star)$};

    \draw[decorate,decoration={brace,raise=0.4cm}, thick]
      (\ar, 1.6) -- (\itrB, 1.6) node[midway,above=0.6cm] {$I_{\mathrm{det}}^{(r)}$};

    \draw[decorate,decoration={brace,raise=0.4cm}, thick]
      (\arp, 1.6) -- (\itrpB, 1.6) node[midway,above=0.6cm] {$I_{\mathrm{det}}^{(r+1)}$};

  \end{tikzpicture}

  \caption{Partition of the detection set $I_{\mathrm{det}}$: due to the minimal coverage condition, each true segment (e.g., $S^\star_r$) contains at least one detection index (blue dots). The projections $\pi^+(\tau_r^\star+1)$ and $\pi^-(\tau_{r+1}^\star)$ respectively identify the first and last available detection indices inside the true segment $S^\star_r$. These projections effectively partition $I_{\mathrm{det}}$ into disjoint sub-collections $I_{\mathrm{det}}^{(r)}$ exactly aligned with the interior of the true underlying change-points.}
  \label{fig:It}
\end{figure}
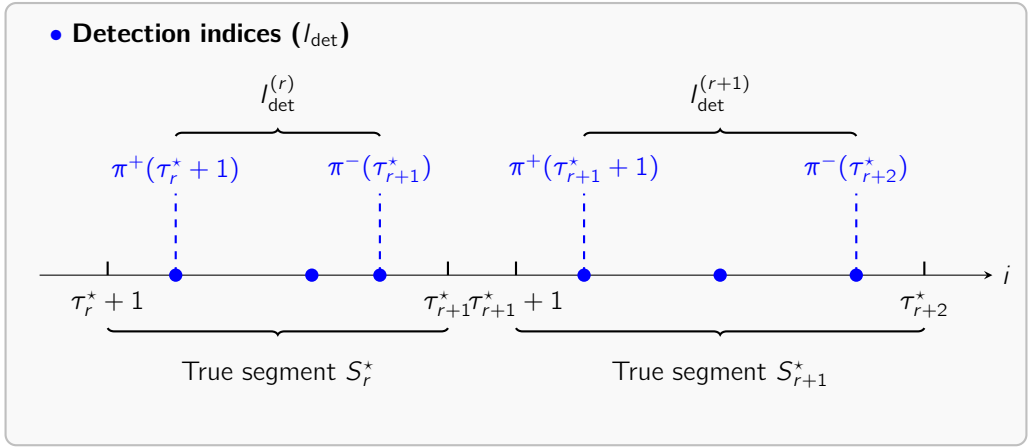

\subsubsection{Theoretical analysis for localization phase}
Using the projections $\pi^+(\tau_r^\star + 1)$ (i.e., the closest point in $I_{\mathrm{det}}$ located to the right of the breakpoint $\tau_r^\star$), the detection set $I_{\mathrm{det}}$ can be partitioned into segments that are aligned with the true change-points:
\[
I_{\mathrm{det}} = I_{\mathrm{det}}^{(0)} \cup I_{\mathrm{det}}^{(1)} \cup \cdots \cup I_{\mathrm{det}}^{(\kappa^\star-1)}.
\]
Each segment is defined as
\[
I_{\mathrm{det}}^{(r)} = I_{\mathrm{det}} \cap \llbracket \tau_r^\star + 1,\; \tau_{r+1}^\star \rrbracket,
\]
and therefore contains only observations generated from the same distribution. This construction ensures that the detection subset $I_{\mathrm{det}}$ is coherently partitioned into homogeneous sub-collections, which is critical for accurate change-point localization.
\begin{lemma}\label{lem:pi+}
Under the minimal coverage condition (Assumption \ref{assum:Holdoutseg}), the projections $\pi^+(\tau_r^\star+1)$ are well-defined and satisfy
\[
\pi^+(\tau_0^\star+1) < \pi^+(\tau_1^\star+1) < \cdots < \pi^+(\tau_{\kappa^\star-1}^\star+1).
\]
Moreover, $\pi^-(n)$, defined as the last element of $I_{\mathrm{det}}$, is also well-defined.
\end{lemma}
The proof is provided in Appendix~\ref{app:proof_pi+}.

The true change-points $(\tau_r^\star)$ partition the full series $(Z_i)_{i\in \llbracket 1,n \rrbracket}$. However, since the detection set $I_{\mathrm{det}}=\{j_1<\dots<j_{n_{\mathrm{det}}}\}$ is only a subset of  $\llbracket 1,n \rrbracket$, the true boundaries do not directly induce a segmentation of
the reindexed detection series $(\widetilde Z_i)$, where
$\widetilde Z_i:= Z_{j_i}$. To recover the analogue of the true segmentation on the detection set, we introduce the projected boundaries
$\tau^{\star,\mathrm{det}}=[\tau^{\star,\mathrm{det}}_0,\dots,\tau^{\star,\mathrm{det}}_{\kappa^\star}]$ defined by
\[
[j_{\tau^{\star,\mathrm{det}}_1+1},\; j_{\tau^{\star,\mathrm{det}}_2+1},\;\dots,\;
 j_{\tau^{\star,\mathrm{det}}_{\kappa^\star-1}+1}]
=
[\pi^+(\tau_1^\star+1),\; \pi^+(\tau_2^\star+1),\;\dots,\;
 \pi^+(\tau_{\kappa^\star-1}^\star+1)].
\]
We set $\tau^{\star,\mathrm{det}}_0:= 0$, so that $j_{\tau^{\star,\mathrm{det}}_0+1}=\pi^+(1)$, and
$\tau^{\star,\mathrm{det}}_{\kappa^\star}:=n_{\mathrm{det}}$, so that $j_{\tau^{\star,\mathrm{det}}_{\kappa^\star}}=\pi^-(n)$.

\begin{lemma}\label{lem:pi-}
Under the minimal coverage condition (Assumption \ref{assum:Holdoutseg}), for every integer $r$ such that $0 \le r \le \kappa^\star-1$, we have:
\[
j_{\tau^{\star,\mathrm{det}}_{r}} = \pi^-(\tau^\star_r).
\]
\end{lemma}
The proof is provided in Appendix~\ref{app:proof_pi-}.

\begin{lemma}\label{lem:pi_det}
Under the minimal coverage condition (Assumption \ref{assum:Holdoutseg}), for every $0 \le r \le \kappa^\star-1$, the interval
\[
\llbracket \tau^{\star,\mathrm{det}}_r + 1,\; \tau^{\star,\mathrm{det}}_{r+1} \rrbracket
\]
is nonempty, and the observations
$(\widetilde{Z}_i)_{i \in \llbracket \tau^{\star,\mathrm{det}}_r + 1,\, \tau^{\star,\mathrm{det}}_{r+1} \rrbracket}$
are i.i.d.\ with distribution $P_{(r+1)}^\phi$.
\end{lemma}
The proof is provided in Appendix~\ref{app:proof_pi_det}.

We therefore conclude that $\tau^{\star,\mathrm{det}}$ is the segmentation of the series $(\widetilde{Z}_i)_{i\in \llbracket1, n_{\mathrm{det}}\rrbracket}$.

We seek to estimate this segmentation. To do so, we assume the existence of a high-probability event related to the change-point detection algorithm on the hold-out sample, mirroring our initial global assumptions:

\begin{assumption}[Change-point detection and localization on the detection set]
\label{assum:detection_event_holdout}
There exists an event $\widetilde{\Omega}_{\alpha_0}$ with $\mathbb{P}(\widetilde{\Omega}_{\alpha_0}) \geq 1 - \alpha_0$ upon which the detection procedure provides an estimated segmentation $\hat{\tau}^{\mathrm{det}}$ (with $|\hat{\tau}^{\mathrm{det}}| = \widehat{\kappa}^{\mathrm{det}} + 1$) satisfying the following two properties:
\begin{enumerate}
    \item \textbf{Correct number of change-points:} 
    \begin{equation}\label{eq:kappa_starHoldout}
         \widehat{\kappa}^{\mathrm{det}} = \kappa^{\star}.
    \end{equation}
    \item \textbf{Localization accuracy:} There exists $\delta_{n_{\mathrm{det}}} = \delta_{n_{\mathrm{det}}}(\alpha_0) \in \mathbb{N}$ such that:
    \begin{equation}\label{eq:delta_nHoldout}
        d_{\infty}^{(1)}(\tau^{\star,\mathrm{det}}, \hat{\tau}^{\mathrm{det}}) \leq \delta_{n_{\mathrm{det}}}.
    \end{equation}
\end{enumerate}
\end{assumption}

\begin{assumption}[Minimal spacing on the detection set]
\label{assum:minimal_spacing_holdout}
We assume that the minimal segment length in the detection set is bounded below, satisfying:
\begin{equation}\label{eq:delta_n'Holdout}
    n_{\mathrm{det}} \underline{\Lambda}^{\mathrm{det}}_{\tau^{\star}} > \delta'_{n_{\mathrm{det}}},
\end{equation}
where $\delta'_{n_{\mathrm{det}}} := \delta_{n_{\mathrm{det}}}\left( 1 + \mathbf{1}_{\kappa^\star > 2} \right)$.
\end{assumption}

\subsubsection{Theoretical analysis for attribution}
Let 
\[
q_{\mathrm{det}}:= \max_{0 \le i \le n_{\mathrm{det}}-1}\left(j_{i+1}-j_{i} -1\right),
\]
with $j_0 := 0$. The parameter $q_{\mathrm{det}}$ represents the maximum number of elements lying between $j_i$ and $j_{i+1}$ that do not belong to $I_{\mathrm{det}}$ but instead belong to $I_{\mathrm{inf}}$. According to the theoretical results presented below, the larger this parameter is, the fewer observations remain available for the attribution step, making the inference problem more difficult.
Under Assumptions~\ref{assum:detection_event_holdout} and~\ref{assum:minimal_spacing_holdout}, the localization properties established in Lemma~\ref{lem:IC} extend to the detection subseries $(\widetilde{Z}_{i})_{1 \le i \le n_{\mathrm{det}}}$. We can thus define the corresponding localization error bound in the total series scale as $\delta_{n}^{\mathrm{tot}} := (q_{\mathrm{det}}+1)\delta_{n_{\mathrm{det}}}$. We have the following lemma.
\begin{lemma}\label{lem:tau_det}
Suppose that Assumptions~\ref{assum:Holdoutseg}, \ref{assum:detection_event_holdout} and \ref{assum:minimal_spacing_holdout} hold. Then, for every $1 \le r \le \kappa^\star-1$, 
\[ 
-\delta_{n}^{\text{tot}} \leq \tau^{\star}_{r} - j_{\hat{\tau}^{\mathrm{det}}_{r}} \leq \delta_{n}^{\text{tot}} + q_{\mathrm{det}}.
\]
\end{lemma}
The proof is provided in Appendix~\ref{app:proof_tau_det}.

Using this result, the inference segments are precisely defined to avoid data contamination from adjacent true change-points:
\[
S_{i_0}^{-,\mathrm{inf}}:= \llbracket j_{\hat{\tau}^{\mathrm{det}}_{i_0-1}} + 1 + (\delta_{n}^{\text{tot}}+q_{\mathrm{det}})\cdot\mathbf{1}_{\{i_0 - 1 \neq 0\}}\, , \, j_{\hat{\tau}^{\mathrm{det}}_{i_0}+1}-1  \rrbracket \cap I_{\mathrm{inf}},
\]
\[
S_{i_0}^{+,\mathrm{inf}}:= \llbracket j_{\hat{\tau}^{\mathrm{det}}_{i_0}+1}\, , \, j_{\hat{\tau}^{\mathrm{det}}_{i_0+1}} - \delta_{n}^{\text{tot}}\cdot\mathbf{1}_{\{i_0 + 1 \neq \kappa^\star\}} \rrbracket \cap I_{\mathrm{inf}}.
\]
We have the following result.
\begin{lemma}\label{lem:S_S+}
Suppose that Assumptions~\ref{assum:Holdoutseg}, \ref{assum:detection_event_holdout} and \ref{assum:minimal_spacing_holdout} hold. Then:
\[
\forall j \in S_{i_0}^{-,\mathrm{inf}} \cup S_{i_0}^{+,\mathrm{inf}}, \quad j \in \llbracket \tau^{\star}_{i_0-1}+1\, , \, \tau^{\star}_{i_0+1} \rrbracket.
\]
\end{lemma}
The proof is provided in Appendix~\ref{app:proof_S_S+}.

To make the hold-out procedure more concrete, Figure~\ref{fig:odd_even_split} illustrates an example of the odd/even splitting strategy.
\begin{figure}[tbp]
        \centering
        \resizebox{\textwidth}{!}{
        \begin{tikzpicture}[x=0.65cm, y=1.2cm, >=stealth]

        \node[font=\small] at (6, 8.6) {\textcolor{blue}{$\bullet$} Detection set ($I_{\mathrm{det}}$, Odd)};
        \node[font=\small] at (18, 8.6) {\textcolor{green!70!black}{$\times$} Inference set ($I_{\mathrm{inf}}$, Even)};

        \def\tstarm{2}   
        \def\tstar{12}   
        \def\tstarp{20}  

        \def\thatm{3}    
        \def\that{11}    
        \def\thatp{21}   
        \def\delt{2}     

        \draw[draw=gray!50, thick, rounded corners, fill=gray!5] (-0.5, 6.2) rectangle (25.5, 8.0);
        \node[below right, font=\bfseries] at (-0.3, 8.0) {1. Odd/even Split};
        \draw[->, thin] (0.0, 6.5) -- (25, 6.5) node[right] {$i$};
        
        \draw[red, dashed, thick] (\tstarm, 6.1) -- (\tstarm, 7.0) node[above, font=\bfseries] {$\tau^\star_{i_0-1}$};
        \draw[red, dashed, thick] (\tstar, 6.1) -- (\tstar, 7.0) node[above, font=\bfseries] {$\tau^\star_{i_0}$};
        \draw[red, dashed, thick] (\tstarp, 6.1) -- (\tstarp, 7.0) node[above, font=\bfseries] {$\tau^\star_{i_0+1}$};
        
        \foreach \x in {1,3,...,23} { \fill[blue] (\x, 6.5) circle (2.5pt); }
        \foreach \x in {2,4,...,24} { \node[green!70!black, thick] at (\x, 6.5) {$\times$}; }

        \draw[draw=gray!50, thick, rounded corners, fill=gray!5] (-0.5, 3.5) rectangle (25.5, 6.0);
        \node[below right, font=\bfseries] at (-0.3, 6.0) {2. Detection (on odds only)};
        \draw[->, thin] (0.0, 4.0) -- (25, 4.0) node[right] {$i$};
        
        \draw[red, dashed, thick, opacity=0.3] (\tstarm, 3.6) -- (\tstarm, 5.3);
        \draw[red, dashed, thick, opacity=0.3] (\tstar, 3.6) -- (\tstar, 5.3);
        \draw[red, dashed, thick, opacity=0.3] (\tstarp, 3.6) -- (\tstarp, 5.3);
        
        \foreach \x in {1,3,...,23} { \fill[blue] (\x, 4.0) circle (2.5pt); }

        \draw[thick, blue] (\thatm, 3.8) -- (\thatm, 4.2) node[above=2pt, font=\bfseries] {$\widehat{\tau}^{\mathrm{det}}_{i_0-1}$};
        \draw[thick, blue] (\that, 3.8) -- (\that, 4.2) node[above=2pt, font=\bfseries] {$\widehat{\tau}^{\mathrm{det}}_{i_0}$};
        \draw[thick, blue] (\thatp, 3.8) -- (\thatp, 4.2) node[above=2pt, font=\bfseries] {$\widehat{\tau}^{\mathrm{det}}_{i_0+1}$};
        
        \draw[blue, thick, |-|] (\thatm-\delt, 5.0) -- (\thatm+\delt, 5.0) node[midway, above, scale=0.85] {$\pm \delta_{n_\mathrm{det}}$};
        \fill[blue, opacity=0.1] (\thatm-\delt, 3.6) rectangle (\thatm+\delt, 5.1);

        \draw[blue, thick, |-|] (\thatp-\delt, 5.0) -- (\thatp+\delt, 5.0) node[midway, above, scale=0.85] {$\pm \delta_{n_\mathrm{det}}$};
        \fill[blue, opacity=0.1] (\thatp-\delt, 3.6) rectangle (\thatp+\delt, 5.1);

        \draw[draw=gray!50, thick, rounded corners, fill=gray!5] (-0.5, 0.5) rectangle (25.5, 3.2);
        \node[below right, font=\bfseries] at (-0.3, 3.2) {3. Exclusion boundaries};
        \draw[->, thin] (0.0, 1.5) -- (25, 1.5) node[right] {$i$};

        \draw[blue, dashed, thin] (\thatm+\delt+1, 3.6) -- (\thatm+\delt+1, 0.5);
        \fill[blue, opacity=0.05] (\thatm-\delt, 0.5) rectangle (\thatm+\delt+1, 3.2);

        \draw[blue, dashed, thin] (\thatp-\delt, 3.6) -- (\thatp-\delt, 0.5);
        \fill[blue, opacity=0.05] (\thatp-\delt, 0.5) rectangle (\thatp+\delt+1, 3.2);

        \draw[gray, dotted, thick] (\that, 3.6) -- (\that, 0.5);

        \def\sminstart{7}   
        \def\sminend{12}     
        \def\splusstart{13}  
        \def\splusend{19}    

        \fill[violet] (\sminstart, 1.5) circle (3pt);
        \fill[violet] (\sminend, 1.5) circle (3pt);
        \fill[violet] (\splusstart, 1.5) circle (3pt);
        \fill[violet] (\splusend, 1.5) circle (3pt);

        \node[below, text=violet, scale=0.85, yshift=-3pt] at (\sminstart, 1.5) {$j_{\hat{\tau}^{\mathrm{det}}_{i_0-1}} + 1 + (2\delta_{n_{\mathrm{det}}}+1)$};
        
        \node[below left, text=violet, scale=0.85, yshift=-3pt, xshift=4pt] at (\sminend, 1.5) { $j_{\hat{\tau}^{\mathrm{det}}_{i_0}+1}-1$};
        
        \node[below right, text=violet, scale=0.85, yshift=-3pt] at (\splusstart, 1.5) {$j_{\hat{\tau}^{\mathrm{det}}_{i_0}+1}$};
        
        \node[below, text=violet, scale=0.85, yshift=-3pt] at (\splusend, 1.5) {$j_{\hat{\tau}^{\mathrm{det}}_{i_0+1}} - 2\delta_{n_{\mathrm{det}}}$};

        \draw[draw=gray!50, thick, rounded corners, fill=gray!5] (-0.5, -2.2) rectangle (25.5, 0.2);
        \node[below right, font=\bfseries] at (-0.3, 0.2) {4. Inference (Intersection with evens: $\cap \, I_{\mathrm{inf}}$)};
        \draw[->, thin] (0.0, -1.0) -- (25, -1.0) node[right] {$i$};
        
        \foreach \x in {2,4,...,24} { \node[green!70!black, thick] at (\x, -1.0) {$\times$}; }
        
        \draw[blue, dashed, thin] (\thatm+\delt+1, 0.5) -- (\thatm+\delt+1, -1.5);
        \fill[blue, opacity=0.05] (\thatm-\delt, -1.5) rectangle (\thatm+\delt+1, 0.2);
        
        \draw[blue, dashed, thin] (\thatp-\delt, 0.5) -- (\thatp-\delt, -1.5);
        \fill[blue, opacity=0.05] (\thatp-\delt, -1.5) rectangle (\thatp+\delt+1, 0.2);

        \draw[gray, dotted, thick] (\that, 0.5) -- (\that, -1.5);
        
        \def\smininfstart{8}
        \def\smininfend{12}
        \draw[green!70!black, very thick, |-|] (\smininfstart, -1.5) -- (\smininfend, -1.5) node[midway, below, font=\bfseries] {$S_{i_0}^{-,\mathrm{inf}}$};

        \def\splusinfstart{14}
        \def\splusinfend{18}
        \draw[green!70!black, very thick, |-|] (\splusinfstart, -1.5) -- (\splusinfend, -1.5) node[midway, below, font=\bfseries] {$S_{i_0}^{+,\mathrm{inf}}$};

        \end{tikzpicture}
        }
        \caption{Complete illustration of the odd/even hold-out splitting strategy defining the test segments. \textbf{Step 1 (Split):} Partitioning into an odd detection set and an even inference set. \textbf{Step 2 (Detection):} Change-points are estimated on the detection set. \textbf{Step 3 (Exclusion Boundaries):} By applying Lemma~\ref{lem:tau_det}, we establish exclusion neighborhoods (shaded blue zones) around the estimated change-points. Discarding these regions ensures that, under the null hypothesis $H_0^{\phi}$, the remaining observations are strictly i.i.d. \textbf{Step 4 (Inference):} The final discrete test segments $S_{i_0}^{-,\mathrm{inf}}$ and $S_{i_0}^{+,\mathrm{inf}}$ are obtained by intersecting these boundary rules with the inference set ($I_{\mathrm{inf}}$).}
        \label{fig:odd_even_split}
\end{figure}
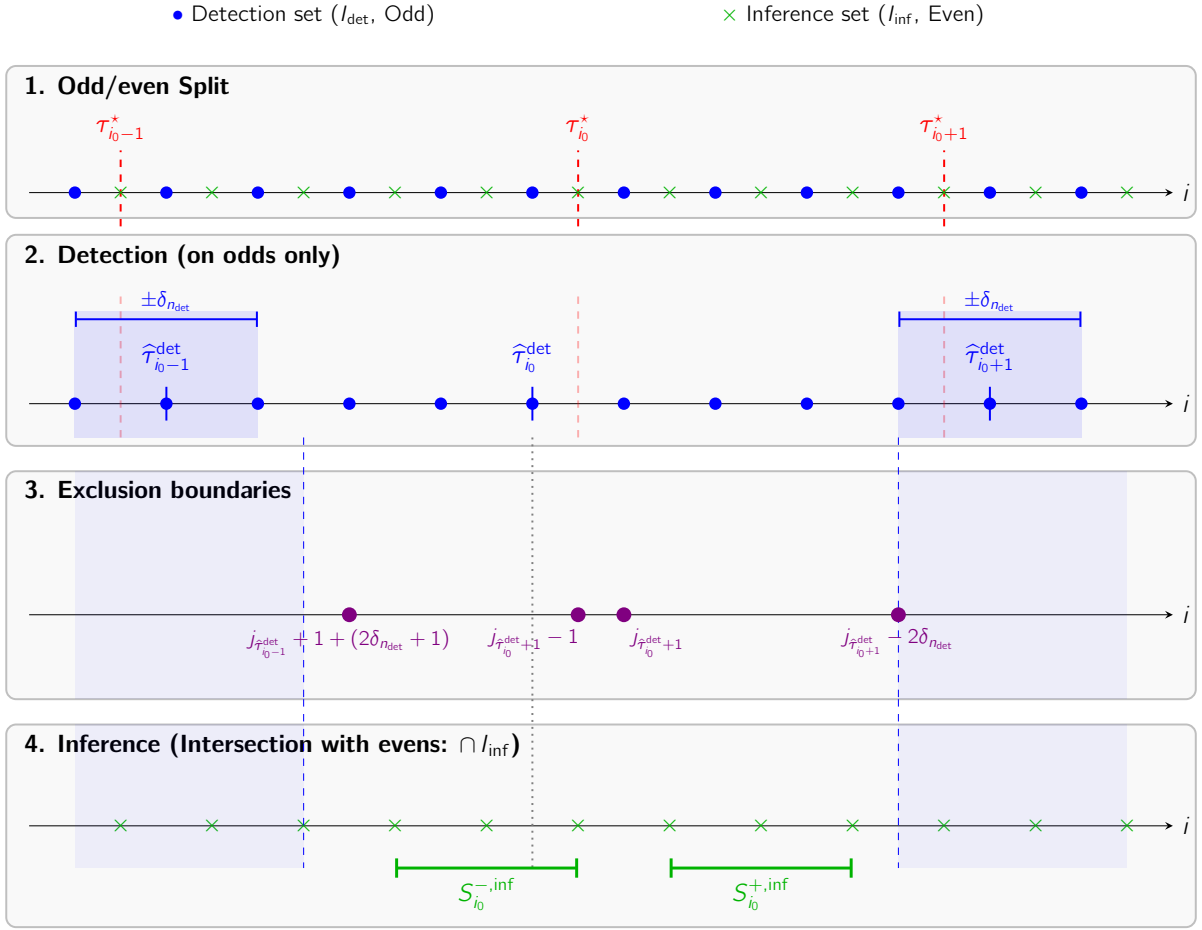

Before introducing the associated theoretical guarantees, we formally define the global test for the hold-out procedure:
\[
T_{\alpha_1}^{\mathcal{Z},\mathrm{holdout}}(S_{i_0}^{-,\mathrm{inf}}, S_{i_0}^{+,\mathrm{inf}}):= T_{\alpha_1}^{\mathcal{Z}}\left( (Z_{j})_{j \in S_{i_0}^{-,\mathrm{inf}}}, (Z_{j})_{j \in S_{i_0}^{+,\mathrm{inf}}} \right),
\]
where $T_{\alpha_1}^{\mathcal{Z}}$ is the same generic level-$\alpha_1$ local two-sample test introduced in Section~\ref{subsec:GTST}. 

Similar to the GTST framework, Algorithm~\ref{algo:holdout} summarizes the complete practical workflow of this data-splitting approach. The theoretical validity of the method is then formally established in the following theorem.

\begin{algorithm}[tbp]
\caption{Hold-out \emph{post hoc} attribution test}
\label{algo:holdout}
\textbf{Input:} Time series $(X_1,\ldots,X_n)$, number of segments $\kappa^\star$, index $i_0$, localization sample $I_{\mathrm{det}}$, inference sample $I_{\mathrm{inf}}$, level $\alpha_1$, generic two-sample test $T_{\alpha_1}^{\mathcal Z}$.\\
\textbf{Output:} Test decision.

\begin{enumerate}
\item Estimate the change-points on the localization sample $(X_t)_{t\in I_{\mathrm{det}}})$, assuming $\kappa^\star$ segments, and denote the resulting segmentation by $\hat{\tau}^{\,\mathrm{det}}$.

\item Compute the theoretical localization error bound $\delta_{n_{\mathrm{det}}}$.

\item Determine the corresponding indices in the total series:
\[
j_{\hat{\tau}^{\,\mathrm{det}}_{i_0-1}+1},\qquad
j_{\hat{\tau}^{\,\mathrm{det}}_{i_0}},\qquad
j_{\hat{\tau}^{\,\mathrm{det}}_{i_0}+1},\qquad
j_{\hat{\tau}^{\,\mathrm{det}}_{i_0+1}}.
\]

\item Construct the two inference test segments:
\[
S_{i_0}^{-,\mathrm{inf}}:= \llbracket j_{\hat{\tau}^{\mathrm{det}}_{i_0-1}} + 1 + (\delta_{n}^{\text{tot}}+q_{\mathrm{det}})\cdot\mathbf{1}_{\{i_0 - 1 \neq 0\}}\, , \, j_{\hat{\tau}^{\mathrm{det}}_{i_0}+1}-1  \rrbracket \cap I_{\mathrm{inf}},
\]
\[
S_{i_0}^{+,\mathrm{inf}}:= \llbracket j_{\hat{\tau}^{\mathrm{det}}_{i_0}+1}\, , \, j_{\hat{\tau}^{\mathrm{det}}_{i_0+1}} - \delta_{n}^{\text{tot}}\cdot\mathbf{1}_{\{i_0 + 1 \neq \kappa^\star\}} \rrbracket \cap I_{\mathrm{inf}}.
\]
\item Compute the test statistic using the observations within the inference segments:
\[
T_{\alpha_1}^{\mathcal Z,\mathrm{holdout}}
=
T_{\alpha_1}^{\mathcal Z}
\bigl(
(Z_i)_{i\in S_{i_0}^{-,\mathrm{inf}}},
(Z_i)_{i\in S_{i_0}^{+,\mathrm{inf}}}
\bigr).
\]

\item Return $T_{\alpha_1}^{\mathcal Z, \mathrm{holdout}}$.
\end{enumerate}
\end{algorithm}

\begin{theorem} 
\label{thm:Hold-out}
Suppose that Assumptions~\ref{assum:Holdoutseg}, \ref{assum:detection_event_holdout} and \ref{assum:minimal_spacing_holdout} hold. Then $T_{\alpha_1}^{\mathcal{Z}, \mathrm{holdout}}$ is a test of
\[
H_0^{\phi}: P_{(i_0)}^{\phi} = P_{(i_0+1)}^{\phi} \quad \text{against} \quad H_1^{\phi}: P_{(i_0)}^{\phi} \neq P_{(i_0+1)}^{\phi},
\]
with level $ \alpha_0 + \alpha_1 $.
\end{theorem}
\begin{proof}[\textbf{Proof.}]
We have:
\begin{align*}
\mathbb{P}_{H_0^{\phi}}\left(T_{\alpha_1}^{\mathcal{Z}, \mathrm{holdout}}(S_{i_0}^{-,\mathrm{inf}}, S_{i_0}^{+,\mathrm{inf}}) = 1\right) &\leq \mathbb{P}(\widetilde{\Omega}_{\alpha_0}^c) + \mathbb{P}_{H_0^{\phi}}\left(\widetilde{\Omega}_{\alpha_0} \cap \left\{T_{\alpha_1}^{\mathcal{Z}, \mathrm{holdout}}(S_{i_0}^{-,\mathrm{inf}}, S_{i_0}^{+,\mathrm{inf}}) = 1\right\}\right)\\
&\leq \alpha_0 + \mathbb{P}_{H_0^{\phi}}\left(\widetilde{\Omega}_{\alpha_0} \cap \left\{T_{\alpha_1}^{\mathcal{Z}}\left( (Z_{j})_{j \in S_{i_0}^{-,\mathrm{inf}}}, (Z_{j})_{j \in S_{i_0}^{+,\mathrm{inf}}} \right) = 1\right\}\right)\\
&=  \alpha_0 + \mathbb{E}\left[ \mathbb{1}_{\widetilde{\Omega}_{\alpha_0}} \mathbb{P}_{H_0^{\phi}}\left(T_{\alpha_1}^{\mathcal{Z}}\left( (Z_{j})_{j \in S_{i_0}^{-,\mathrm{inf}}}, (Z_{j})_{j \in S_{i_0}^{+,\mathrm{inf}}} \right) = 1 \big| \hat{\tau}^{\mathrm{det}} \right) \right]\\
&\leq \alpha_0 + \alpha_1.
\end{align*}
For the last inequality, we use the fact that under the null hypothesis $H_0^{\phi}$, the observations $(Z_{j})_{j \in S_{i_0}^{-,\mathrm{inf}} \cup S_{i_0}^{+,\mathrm{inf}}}$ are i.i.d.\ conditionally on $\hat{\tau}^{\mathrm{det}}$, and that $T_{\alpha_1}^{\mathcal{Z}}$ is a level-$\alpha_1$ test. Indeed, independence follows from the fact that $\hat{\tau}^{\mathrm{det}}$ is a function of the data in $I_{\mathrm{det}}$ only, and is therefore independent of the data originating from $I_{\mathrm{inf}}$.
\end{proof}
\begin{remark}[Single change-point setting]
\label{rem:single_cp}
For the hold-out procedure when $\kappa^\star = 2$, i.e., in the single change-point setting, the assumptions \eqref{eq:delta_nHoldout} and \eqref{eq:delta_n'Holdout} are no longer required. In this case ($i_0 = 1$), it suffices that conditions \eqref{eq:Holdoutseg} and \eqref{eq:kappa_starHoldout} hold. Indeed, the data used for the \emph{post hoc} test are independent of the data used for change-point localization. As a consequence, regardless of the localization accuracy, the resulting two-sample test is performed on observations that are independent of the estimated change-point. Therefore, the test remains controlled at level $\alpha_1$, rather than at the level $\alpha_0 + \alpha_1$. We have:
\begin{align*}
\mathbb{P}_{H_0^{\phi}}\left(T_{\alpha_1}^{\mathcal{Z}, \mathrm{holdout}}(S_{i_0}^{-,\mathrm{inf}}, S_{i_0}^{+,\mathrm{inf}}) = 1\right) &\leq  \mathbb{P}_{H_0^{\phi}}\left(T_{\alpha_1}^{\mathcal{Z}}\left( (Z_{j})_{j \in S_{i_0}^{-,\mathrm{inf}}}, (Z_{j})_{j \in S_{i_0}^{+,\mathrm{inf}}} \right) = 1\right)\\
&= \mathbb{E}\left[ \mathbb{P}_{H_0^{\phi}}\left(T_{\alpha_1}^{\mathcal{Z}}\left( (Z_{j})_{j \in S_{i_0}^{-,\mathrm{inf}}}, (Z_{j})_{j \in S_{i_0}^{+,\mathrm{inf}}} \right) = 1 \big| \hat{\tau}^{\mathrm{det}} \right) \right]\\
&\leq \alpha_1,
\end{align*}
using the same argument as before for the last inequality.
\end{remark}

\begin{remark}[Practical behavior with multiple change-points]
\label{rem:holdout_multiple}
In the presence of multiple change-points ($\kappa^\star > 2$), it is important to clarify the method's behavior when the minimal spacing condition is not strictly met. Theoretically, this condition is necessary to guarantee the localization bounds established in Lemma~\ref{lem:IC} and, by extension, the formal control of the Type~I error. If the condition is violated, these formal guarantees are lost. In practice, however, the procedure is applied without explicitly verifying this condition. If a severe spacing violation occurs (i.e., $n_{\mathrm{det}} \cdot \underline{\Lambda}^{\mathrm{det}}_{\tau^{\star}} \leq \delta_{n_{\mathrm{det}}}$), one of the resulting test segments is mechanically left empty. As a natural consequence, the two-sample test cannot be executed, and the procedure defaults to retaining $H_0^{\phi}$. Thus, while intermediate violations cause the test to execute without theoretical backing, severe violations simply result in a default conservative decision.
\end{remark}

\subsection{Block-wise post hoc attribution}
Once a \emph{post hoc} attribution procedure has been selected, we apply it independently to the two predefined blocks of coordinates. Recall that the goal is to determine whether a distributional change occurs in each specific block around the detected change-point. 

For a given block, the procedure consists of performing a two-sample test comparing the observations before and after the predicted change-point. As established in Theorems~\ref{thm:GTST} and~\ref{thm:Hold-out}, each marginal test is controlled at level $\alpha_0+\alpha_1$, where $\alpha_0$ corresponds to the probability of localization error and $\alpha_1$ is the nominal level of the two-sample test. The subset of blocks driving the global change is then estimated as the collection of blocks for which the local null hypothesis is rejected.

The primary statistical guarantee provided by this block-wise approach is the control of false positives. Since the initial change-point localization is performed only once jointly on the full multivariate time series, its associated error probability $\alpha_0$ is shared. Subsequently, because two \emph{post hoc} tests are performed (one for each block) at a conditional level $\alpha_1$, the family-wise error rate (FWER)---defined as the probability of falsely attributing a change to at least one invariant block---is controlled via the union bound:
$$ \mathrm{FWER} = \mathbb{P}\big(\text{false positive in Block~1 or Block~2}\big) \le \alpha_0+2\alpha_1. $$

\begin{remark}[Joint distributional changes]
In our framework, the initial detection algorithm may identify a global change-point even when the \emph{post hoc} procedure detects no significant change in any individual block. Statistically, this scenario admits two explanations. On the one hand, the shift might be purely driven by a change in the dependence structure between the blocks (e.g., an inter-block covariance shift), leaving their respective marginal distributions invariant. On the other hand, it may simply reflect insufficient statistical power to detect a weak marginal change. Distinguishing a genuine joint dependence shift from a weak marginal signal would require a dedicated test of the dependence structure, which lies beyond the scope of the present marginal procedures.
\end{remark}

\begin{remark}[Extension to multiple blocks]
Although this exposition focuses on two predefined blocks of coordinates, the proposed framework extends to settings involving more than two blocks. In such cases, the attribution procedure can be applied independently to each block, resulting in multiple hypothesis tests. To control the probability of false positives when several blocks are tested simultaneously, standard multiple testing procedures may be used, such as Bonferroni or Holm corrections \cite{Holm1979, giraud2021introduction}. These approaches ensure that the family-wise error rate (FWER) remains controlled while allowing the identification of the subset of blocks responsible for the detected change. In particular, if $L$ blocks are tested and each individual \emph{post hoc} test is conditionally controlled at level $\alpha_1$, a simple union bound yields
$$\mathrm{FWER} \le \alpha_0+L\alpha_1,$$
where, as previously noted, the initial localization error $\alpha_0$ is shared across all tests.
\end{remark}
\section{Experimental evaluation}
\label{sec:experiments}
In this section, we evaluate the empirical performance of the proposed \emph{post hoc} procedures. The Python implementation of the proposed methods, along with a reproducible example of the simulation study, are publicly available on GitHub at \url{https://github.com/Dhia-Elhaq/Post-hoc-change-point-detection}.

The experiments aim to assess the behavior of the
methods when change-points are estimated with localization uncertainty. In the numerical experiments, the underlying two-sample tests are instantiated as kernel-based MMD tests, chosen for their theoretical guarantees in nonparametric settings (see Section~\ref{subsec:MMD}). We first report results on synthetic data under
various simulation settings, and then illustrate the behavior of the
methods on a real-world dataset.

\subsection{Benchmark procedures}
To assess the effect of localization uncertainty, we introduce two
benchmark procedures that serve as reference baselines in our numerical
experiments.

\paragraph{Oracle test.}
The first procedure corresponds to an ideal scenario in which the true location of the change-point is exactly known. Consequently, no uncertainty or error is introduced by the change-point estimation step. This oracle procedure thus serves as a theoretical upper bound on the achievable statistical power for any valid \emph{post hoc} test.
We recall the definition of the two oracle segments adjacent to $\tau^\star_{i_0}$:
\[
S^{\star-}_{i_0} = \llbracket \tau^\star_{i_0-1}+1 ,\, \tau^\star_{i_0} \rrbracket
\quad \mathrm{and} \quad
S^{\star+}_{i_0} = \llbracket \tau^\star_{i_0}+1 ,\, \tau^\star_{i_0+1} \rrbracket.
\]
The oracle test simply applies the two-sample test directly to these deterministic segments:
\[
T_{\alpha_1}^{\mathcal Z, \mathrm{oracle}}
:=
T_{\alpha_1}^{\mathcal Z}
\big( (Z_i)_{i\in S^{\star-}_{i_0}}, (Z_i)_{i\in S^{\star+}_{i_0}} \big).
\]
By construction, under $H_0^{\phi}$, the observations in the two segments are i.i.d., and therefore this procedure controls the false alarm rate at level $\alpha_1$.

\paragraph{Naive test.}
The second baseline ignores the uncertainty induced by change-point estimation. Let $\widehat{\tau}_{i_0}$ be the estimated change-point returned by the detection procedure, and define the adjacent estimated segments by
\[
\widehat{S}^{-}_{i_0} := \llbracket \widehat{\tau}_{i_0-1}+1 ,\, \widehat{\tau}_{i_0} \rrbracket
\quad \mathrm{and} \quad
\widehat{S}^{+}_{i_0} := \llbracket \widehat{\tau}_{i_0}+1 ,\, \widehat{\tau}_{i_0+1} \rrbracket.
\]
The naive approach consists of applying the same two-sample test to these data:
\[
T_{\alpha_1}^{\mathcal Z, \mathrm{naive}}
:=
T_{\alpha_1}^{\mathcal Z}
\big( (Z_i)_{i\in \widehat{S}^{-}_{i_0}}, (Z_i)_{i\in \widehat{S}^{+}_{i_0}} \big).
\]
This procedure can be viewed as a natural baseline that a practitioner might implement. However, since the segments depend on the data through $\widehat{\tau}_{i_0}$, the independence assumptions underlying the test are no longer guaranteed, and Type I error control may fail.

\subsection{Two-sample tests based on maximum mean discrepancy}\label{subsec:MMD}
Previous \emph{post hoc} methods are based on tests that compare distributions at different time points. In this study, we use kernel-based two-sample tests based on maximum mean discrepancy (MMD) \cite{gretton2012}, which offer a nonparametric approach to comparing distributions in potentially high-dimensional spaces.

\medskip

We first introduce some notions and properties.  
Let $ k : \mathcal{Z} \times \mathcal{Z} \to \mathbb{R} $ be a positive semi-definite kernel, and let $\varphi : \mathcal{Z} \to \mathcal{H}$ denote the feature map associated with $k$, where $\mathcal{H}$ is the reproducing kernel Hilbert space (RKHS) induced by $k$ \cite{Berlinet2004}.  
For all $z \in \mathcal{Z}$, we have $\varphi(z) := k(z, \cdot)$.  
We assume that $\mathcal{H}$ is a separable Hilbert space.  
Assuming that the kernel $k$ is bounded, the kernel mean embeddings of two probability distributions $P$ and $Q$ on $\mathcal{Z}$ are well defined as
\[
\mu_P = \mathbb{E}_{X \sim P}[\varphi(X)] 
\quad \mathrm{and} \quad 
\mu_Q = \mathbb{E}_{X \sim Q}[\varphi(X)].
\]
The \emph{Maximum Mean Discrepancy} (MMD) between $P$ and $Q$ is then defined by:
\[
\mathrm{MMD}(P, Q) := \left\| \mu_P - \mu_Q \right\|_{\mathcal{H}}.
\]
Furthermore, $\mathrm{MMD}$ defines a true metric on probability measures---meaning that $\mathrm{MMD}(P, Q) = 0$ if and only if $P = Q$---provided that the underlying kernel $k$ is \emph{characteristic} \cite{Sriperumbudur_2010}. 

In practice, for hypothesis testing, it is standard to work with the squared distance $\mathrm{MMD}^2(P, Q)$. Given two finite multisets $S_1, S_2 \subseteq \mathcal{Z}$ corresponding respectively to samples drawn from $P$ and $Q$, we estimate the theoretical squared distance $\mathrm{MMD}^2(P, Q)$ by replacing the true probability measures with their empirical counterparts. This yields the biased empirical estimator $\mathrm{MMD}_b^2$:
\begin{align*}
    \mathrm{MMD}_b^2(S_1, S_2) 
    &:= \| \mathcal{B}(S_1) - \mathcal{B}(S_2) \|_{\mathcal{H}}^2 \\
    &= \| \mathcal{B}(S_1) \|_{\mathcal{H}}^2 
     + \| \mathcal{B}(S_2) \|_{\mathcal{H}}^2 
     - 2 \langle \mathcal{B}(S_1), \mathcal{B}(S_2) \rangle_{\mathcal{H}},
\end{align*}
where
\[
    \mathcal{B}(S) := \frac{1}{|S|} \sum_{x \in S} \varphi(x)
\]
denotes the empirical kernel mean embedding of the sample $S$.
\subsubsection{Test based on a theoretical threshold}
\label{subsec:thresholdd}
We assume that the kernel $k$ is positive and bounded by a constant $M^2 > 0$. We consider two disjoint segments $S_1$ and $S_2$, consisting respectively of the indices of data points coming from two subparts of the projected time series $(Z_i)_{1 \le i \le n}$.
The test is defined by:
\[
T((Z_i)_{i \in S_1}, (Z_i)_{i \in S_2}) := \mathbf{1}_{ \left\{\mathrm{MMD}_b((Z_i)_{i \in S_1}, (Z_i)_{i \in S_2}) > \mathrm{threshold}_{\alpha_1}\right\} },
\]
where 
\[\mathrm{threshold}_{\alpha_1} := \sqrt{\frac{M^2}{\min(|S_1|,|S_2|)}}\left(4+2\sqrt{\log\tfrac{2}{\alpha_1}}\right),
\]
is determined from the concentration bound of Theorem~7 in \cite{gretton2012}. Under the null hypothesis $H_0^{\phi}$, where $(Z_i)_{i \in S_1 \cup S_2}$ are i.i.d., the test has level at most $\alpha_1$; see Appendix~\ref{app:threshold} for the proof.

\subsubsection{Permutation test (MMD)}
In this approach, the null distribution of the MMD statistic is approximated by recomputing the statistic under many random permutations of the observation indices. Let $B \in \mathbb{N}^*$ denote the number of performed permutations. The procedure is then as follows:
\begin{enumerate}
    \item Let $\mathrm{MMD}_{b}^{\mathrm{obs}}$ denote the statistic computed on the original (non-permuted) data;
    \item Generate $B$ random permutations of $S_1 \cup S_2$;
    \item For each permutation $\sigma_b, \, b \in \{1, \dots, B\}$, compute \\ $\mathrm{MMD}_b^{(b)}:= \mathrm{MMD}_b\left( (Z_{\sigma_b(i)})_{i \in S_1}, (Z_{\sigma_b(i)})_{i \in S_2} \right)$, i.e., the statistic obtained from the permuted data.
\end{enumerate}
The p-value is then estimated as:
\[
\widehat{p} := \frac{1}{B + 1} \left(1 + \sum_{b = 1}^B \mathbf{1}_{\mathrm{MMD}_b^{(b)} \ge \mathrm{MMD}_{b}^{\mathrm{obs}}}\right).
\]
The test is defined by:
\[
T_{\mathrm{perm}}((Z_i)_{i \in S_1}, (Z_i)_{i \in S_2}) := \mathbf{1}_{\{\widehat{p} < \alpha_1\}}.
\]
\begin{remark}[Level of the permutation test]
If under $H_0$ the observations $(Z_i)_{i \in S_1 \cup S_2}$ are i.i.d., then the permutation test $T_{\mathrm{perm}}$ provides an exact non-asymptotic level $\alpha_1$ (see \cite{Romano2005}).
\end{remark}

\subsection{Simulation study}
\subsubsection{General experimental setting}
\label{subsubsec:general_setting}
\paragraph{Datasets.}
We consider multivariate time series $X_1,\ldots,X_n \in \mathbb{R}^d$ generated from Gaussian distributions whose parameters are piecewise constant over time. Unless explicitly signaled in a specific subsection, the default parameters are: $d = 5, \, n = 500$. A single change-point is placed at the center of the sequence, i.e., $\tau^\star = n/2$, corresponding to $\kappa^\star=2$ segments. Before the change-point, $X_t \sim \mathcal{N}(0, I_d)$, while after the change-point the distribution depends on the experiment (mean, variance, or covariance change). 
As a first step, we evaluate the ability of the procedures to detect changes at the level of blocks of coordinates. The coordinates are divided into two disjoint blocks: $\text{Block 1} = \{1,2,3\}$ and $\text{Block 2} = \{4,5\}$. 

Distributional changes are introduced only in Block~1, while Block~2 remains unchanged and is used to assess Type I error control. This simplified block structure allows us to clearly evaluate the ability of the procedures to attribute a detected change to the correct subset of coordinates while controlling false discoveries.

\paragraph{Evaluation metrics.}
\emph{Post hoc} tests are applied separately to both blocks. Since Block~1 contains the true distributional change, the empirical power is measured on Block~1 as the proportion of simulations in which the null hypothesis is correctly rejected. Conversely, Block~2 remains unchanged, and the empirical Type~I error is estimated on this block as the proportion of simulations in which the null hypothesis is incorrectly rejected. Finally, we report the exact block recovery rate. This empirical metric translates the theoretical concept of exact block-wise attribution into practice: it measures the joint probability of successfully rejecting the null hypothesis on the affected block (power) while simultaneously retaining it on the unchanged block (Type~I error control).

\paragraph{Change-point detection and kernel choice.}
We first perform global change-point detection using the Kernel Change-Point (KCP) procedure, assuming that the true number of segments $\kappa^\star$ is known \cite{Arlot2019} and provided that Assumption \ref{assum:detection_event} is satisfied. 
To maintain a consistent experimental framework, we rely on the Gaussian kernel for all simulations. For any two observations $x$ and $y$, this kernel is defined as
\[ 
k(x,y) = \exp\!\left(-\frac{\|x-y\|^2}{2\sigma^2}\right), 
\]
where $\sigma > 0$ represents the bandwidth parameter. The Gaussian kernel is a standard choice in the literature as it is characteristic \cite{Sriperumbudur_2010}. The kernel bandwidth $\sigma$ is systematically set using the median heuristic \cite{garreau2017large}, computed globally from the pairwise distances between all available observations in the entire time series $(X_1, \dots, X_n)$. While this heuristic provides a
robust and well-established baseline that performs well across the simulation settings considered
in this work, it is worth noting that the optimal choice of the kernel and its parameters inherently
depends on the specific characteristics of the target application. Crucially, this single fixed bandwidth is used for both the initial KCP detection procedure and all subsequent \emph{post hoc} tests.

\paragraph{Choice of the \emph{post hoc} testing procedure.}
Given the estimated change-points, we focus on a specific change-point and conduct \emph{post hoc} two-sample testing. In this study, we rely on a permutation-based MMD test. The main reason is that our framework is fully nonparametric; no explicit assumptions are made regarding the underlying distributions or the type of distributional change. In such settings, tests based on theoretical concentration thresholds can be overly conservative, leading to a significant loss of power \cite{gretton2012}. 
Permutation tests provide an alternative, automatically adapting to the empirical distribution of the statistic under the null hypothesis to achieve higher power while preserving valid level control. The main drawback of this approach is its computational cost, as the MMD squared statistic must be recomputed for a large number of permutations. In our experiments, inference relies on $B=1\,000$ permutations, with each experimental configuration repeated 10\,000 times to obtain stable performance estimates. Recent work has proposed permutation-free kernel two-sample tests aiming to reduce this computational burden. For instance, \cite{shekhar2022permutation} introduce a kernel two-sample test based on a cross-MMD statistic that admits an explicit asymptotic threshold, while maintaining a power close to that of permutation-based methods. This could reduce the computational cost of permutation testing without significantly sacrificing statistical performance.

\paragraph{Threshold parameter and plug-in estimates.}
We set the nominal probability levels to $\alpha_0 = \alpha_1 = 0.05$. According to Theorem~3.1 of \cite{Garreau2018}, successful localization requires the signal-to-noise ratio to be sufficiently large. Specifically, the theoretical condition ensuring valid detection can be expressed as:
\[
\frac{\underline{\Delta}^2}{M^2} > 148 \kappa^\star (\kappa^\star +1)  \frac{\log\big(1/\alpha_0\big) + \log n + 1}{n\underline{\Lambda}_{\tau^\star}},
\]
where $M^2$ denotes the kernel bound (here $M^2 = 1$ since we use the Gaussian kernel). When this condition is met, the theoretical localization error bound $\delta_n^{\mathrm{theo}}$ is directly given by:
\[
\delta_n^{\mathrm{theo}} := \frac{148 \kappa^\star  M^2}{\underline{\Delta}^2} \cdot (\log\big(1/\alpha_0\big) + \log n + 1).
\]
The leading constant $148$ stems from distribution-free concentration inequalities. As is often the case with such theoretical bounds, this constant is known to be overly conservative in practice. 

To mitigate this effect and obtain an operational procedure, we replace the highly conservative constant $148$ with a much smaller empirical tuning parameter $C_0$ that depends on the family of distributions being tested. In most of our simulations, this value was set to $C_0 = 0.15$, unless specified otherwise. Furthermore, we substitute the unknown terms $\kappa^\star $ and $\underline{\Delta}$ with plug-in estimates based on the detected change-points $\hat{\tau}$. We estimate the squared jump size by the empirical squared MMD between adjacent detected segments:
\[
\widehat{\underline{\Delta}}^2 := \min_{1 \le i \le \kappa^\star-1} \mathrm{MMD}_b^2\!\left( \{X_{\hat{\tau}_{i-1}+1},\ldots,X_{\hat{\tau}_i}\}, \{X_{\hat{\tau}_i+1},\ldots,X_{\hat{\tau}_{i+1}}\} \right).
\]
Similarly, we estimate $\underline{\Lambda}_{\tau^\star}$ by $\underline{\Lambda}_{\hat{\tau}} = \frac{1}{n} \min_{1 \le i \le \kappa^\star} (\hat{\tau}_{i} - \hat{\tau}_{i-1})$. Finally, our \emph{post hoc} procedure evaluates the following practical condition:
\[
\frac{\widehat{\underline{\Delta}}^2}{M^2} > C_0 \kappa^\star (\kappa^\star +1)  \frac{\log\big(1/\alpha_0\big) + \log n + 1}{n\underline{\Lambda}_{\hat{\tau}}}.
\]
If this empirical condition is satisfied, the operational localization radius $\delta_n$ used to construct our test segments is computed as:
\[
\delta_n := C_0 \frac{\kappa^\star M^2}{\widehat{\underline{\Delta}}^2} \cdot (\log\big(1/\alpha_0\big) + \log n + 1).
\]
In practice, the tuning parameter $C_0$ is properly calibrated to ensure valid localization with high probability. Therefore, our simulations systematically compute $\delta_n$ and apply the \emph{post hoc} tests directly, without explicitly enforcing the preliminary theoretical condition. 

The calibration of the parameter $C_0$ is studied and discussed in Section \ref{sec:C_0}.

\begin{remark}[Interpretation of the performance metrics]
We emphasize that the evaluated \emph{post hoc} procedures are applied after the KCP algorithm has identified candidate change-points. Therefore, the performance metrics reported in this section should be interpreted as the combined outcome of two successive steps: (1) the initial detection step, which yields the estimated locations, and (2) the subsequent \emph{post hoc} inference step.
\end{remark}

To summarize the complete methodology and clarify the differences between the evaluated approaches, Figure~\ref{fig:pipeline} illustrates the comprehensive simulation pipeline. Specifically, we compare our method against baseline approaches: (i) an \emph{oracle} test, used for benchmarking, which applies a two-sample test directly to the segments defined by the true change-points $ \tau^\star $; (ii) a \emph{naive} test, which applies a two-sample test to the segments adjacent to the estimated change-point $\widehat{\tau}_{i_0}$; (iii) a \emph{hold-out} test, which splits the data into a detection subsample and an independent inference subsample to ensure valid level control; and (iv) the proposed \emph{GTST} procedure, for which we set the grid step to $\eta = \max(\lfloor \delta_n / 2 \rfloor, 1)$. Several values of $\eta$ were tested in preliminary simulations, and this choice provided a reasonable trade-off between statistical power and computational cost.

\begin{figure}[htbp]
        \centering
        \resizebox{\textwidth}{!}{
        \begin{tikzpicture}[
          basebox/.style={draw, rounded corners, align=center, inner sep=6pt, thick, font=\small},
          databox/.style={basebox, fill=gray!10, draw=gray!80},
          oraclebox/.style={basebox, fill=blue!20, draw=blue},
          naivebox/.style={basebox, fill=orange!20, draw=orange!70},
          mybox/.style={basebox, fill=green!10, draw=green!60!black},
          mybox_2/.style={basebox, fill=red!20, draw=red!60!black},
          arr/.style={-{Stealth}, thick, draw=gray!80}
        ]

        \node[databox] (sim) at (0,0)
        {\textbf{Simulation}\\
        Multivariate time series\\
        True change-points $\tau^\star_{i_0}$};

        \node[oraclebox] (oracle) at (-6.2, -2.5)
        {\textbf{Oracle}\\
        Two-sample test (MMD)\\
        on segments around $\tau^\star_{i_0}$};
        
        \draw[arr] (sim.south) -- (oracle.north);

        \node[databox] (kcp) at (0, -2.5)
        {\textbf{KCP on the full series} $\llbracket 1,n\rrbracket$\\
        Estimated change-points $\widehat{\tau}_{i_0}$};
        
        \draw[arr] (sim.south) -- (kcp.north);

        \node[naivebox] (naive) at (-3.2, -5.5)
        {\textbf{Naive}\\
        Two-sample test (MMD)\\
        on segments around $\widehat{\tau}_{i_0}$};
        
        \draw[arr] (kcp.south) -- (naive.north);

        \node[mybox_2] (gtst) at (3.2, -5.5)
        {\textbf{GTST}\\
        Multiple two-sample tests (MMD)\\
        on grid-based segments\\
        around $\widehat{\tau}_{i_0}$};
        
        \draw[arr] (kcp.south) -- (gtst.north);

        \node[mybox] (hold) at (6.2, -2.5)
        {\textbf{Hold-out}\\
        1. KCP on subsample ($I_{\mathrm{det}}$): $\widehat{\tau}^{\mathrm{det}}$\\
        2. Inference (MMD) on subsample \\
        ($I_{\mathrm{inf}}$) around $\widehat{\tau}^{\mathrm{det}}$};
        
        \draw[arr] (sim.south) -- (hold.north);

        \end{tikzpicture}
        }
        \caption{Simulation pipeline and \emph{post hoc} methods.}
\label{fig:pipeline}
    \end{figure}
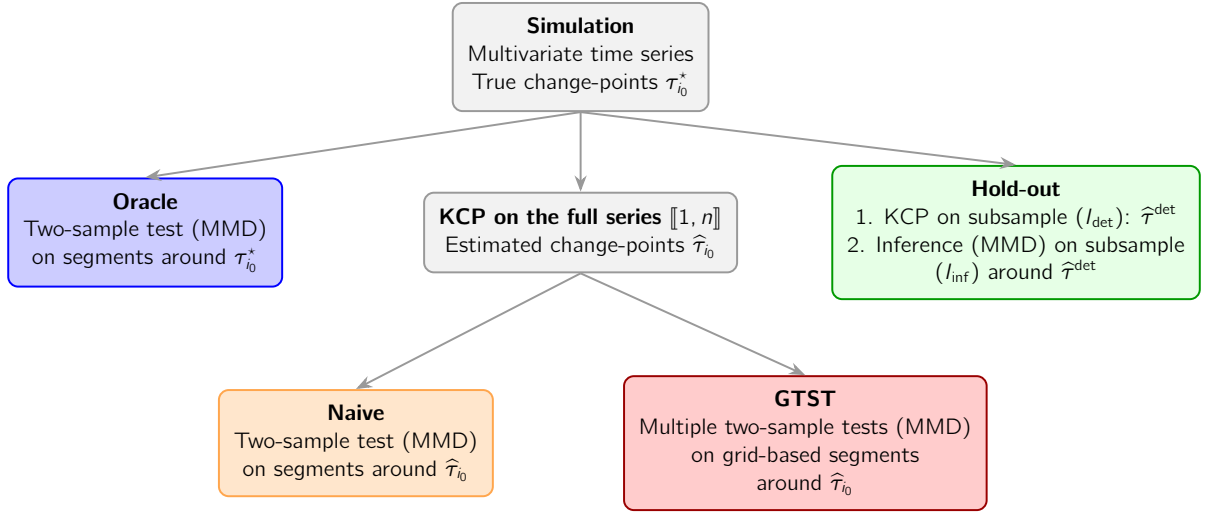
    
\paragraph{General performance trends and effective signal-to-noise ratio.}
Before outlining specific scenarios, it is important to emphasize the consistent qualitative pattern observed in all our experiments. Varying the mean jump magnitude, the location of the change-point, the sample size $n$, or the dimensionality fundamentally modulates the effective signal-to-noise ratio (SNR) of the initial Kernel Change-Point (KCP) detection step. Because the accuracy of our \emph{post hoc} tests is inherently constrained by the initial localization error, a low SNR naturally leads to poor change-point estimates, which in turn degrades downstream statistical power and attribution accuracy.

To maintain clarity, the main text focuses on two representative scenarios illustrating this dynamic: varying the mean jump (to reflect signal strength) and varying the minimal segment length (to reflect estimation accuracy and noise). A comprehensive set of supplementary simulations is provided in Appendix~\ref{sec:complementary_simulations}. These additional experiments explore the underlying SNR effect from broader perspectives, categorizing challenges as either signal-related (e.g., changes in variance, covariance, dimensionality, or sparsity) or estimation-related (e.g., kernel selection or the presence of multiple change-points).

\subsubsection{Signal strength: The effect of mean jump magnitude} 

\paragraph{Setting.} 
In this experiment, we investigate the influence of the mean-shift amplitude, denoted by \texttt{jump}, on the performance of the \emph{post hoc} procedures. We consider the general simulation framework defined in Section~\ref{subsubsec:general_setting}, where the change-point occurs at $\tau^\star = n/2$. Following this change-point, a mean shift affects exclusively the coordinates in Block~1 = $\{1,2,3\}$. Specifically, the distribution shifts to:
\[ X_t \sim \mathcal{N}(\mu, I_d), \] 
with 
\[ \mu = (\texttt{jump}, \texttt{jump}, \texttt{jump}, 0, 0)^\top. \] 
The magnitude \texttt{jump} varies over a predefined grid, while all other parameters are kept fixed, allowing us to evaluate the impact of the signal strength on the detection and attribution accuracy.

\paragraph{Results.} 
The performance metrics across the different procedures are displayed in Figure~\ref{fig:mean_jump_results} (see Table~\ref{tab:results_mean_jump} for specific numerical values at selected jump magnitudes).

\begin{figure}[tbp]
\centering

\begin{subfigure}{0.48\textwidth}
    \centering
    \includegraphics[width=\textwidth]{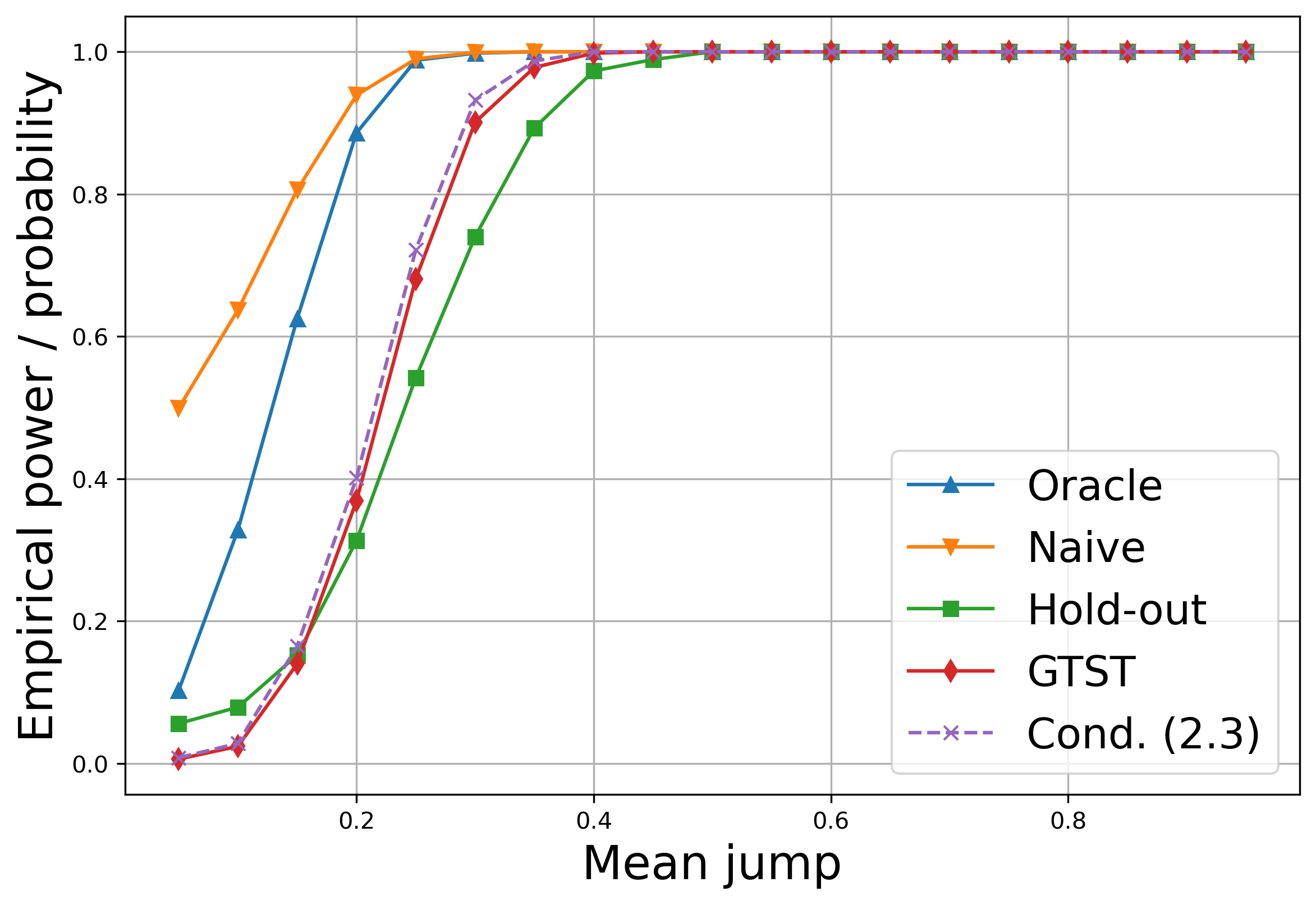}
    \caption{Power}
    \label{fig:mean_power}
\end{subfigure}
\hfill
\begin{subfigure}{0.48\textwidth}
    \centering
    \includegraphics[width=\textwidth]{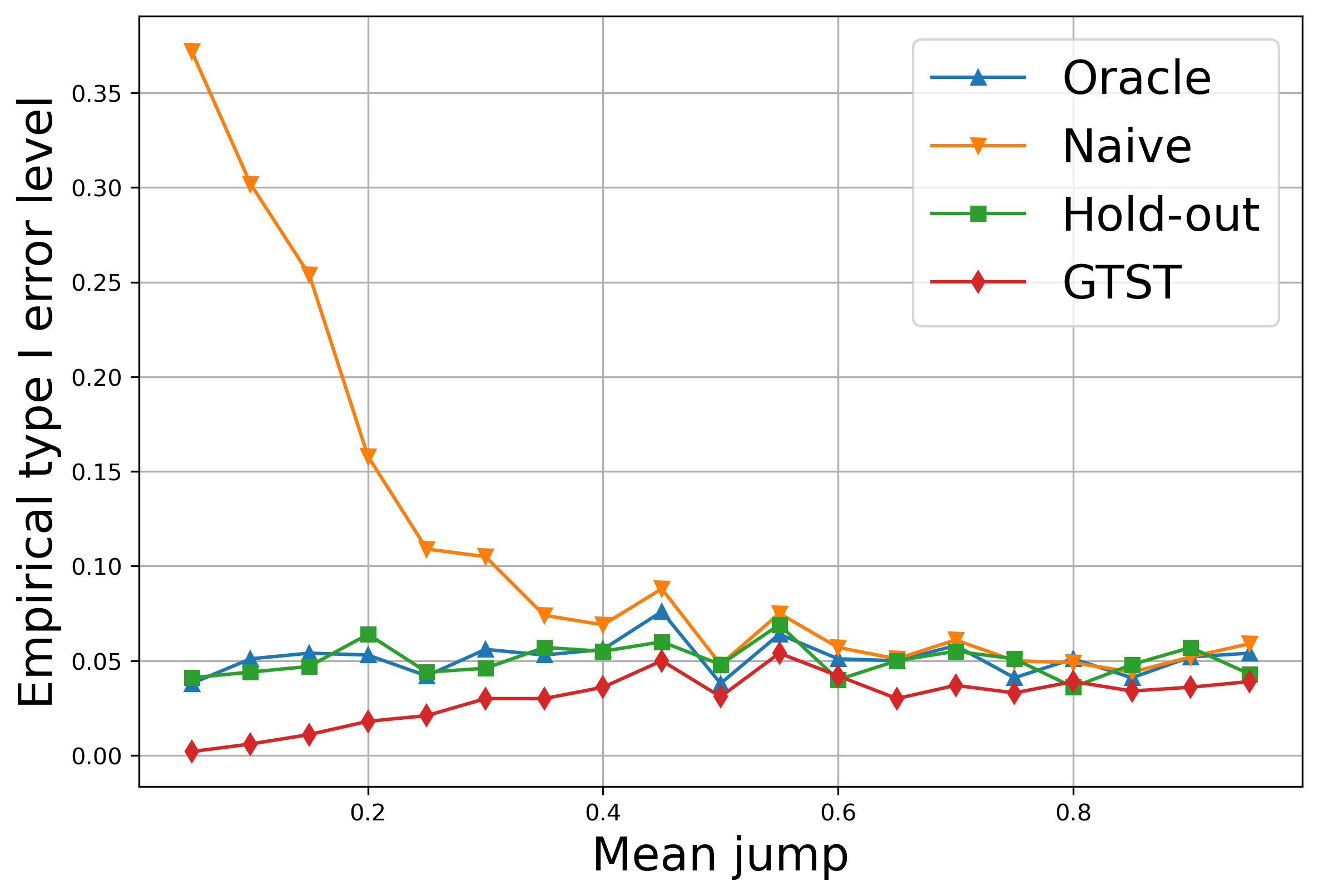}
    \caption{Type I error level}
    \label{fig:mean_level}
\end{subfigure}

\vspace{0.5cm}

\begin{subfigure}{0.5\textwidth}
    \centering
    \includegraphics[width=\textwidth]{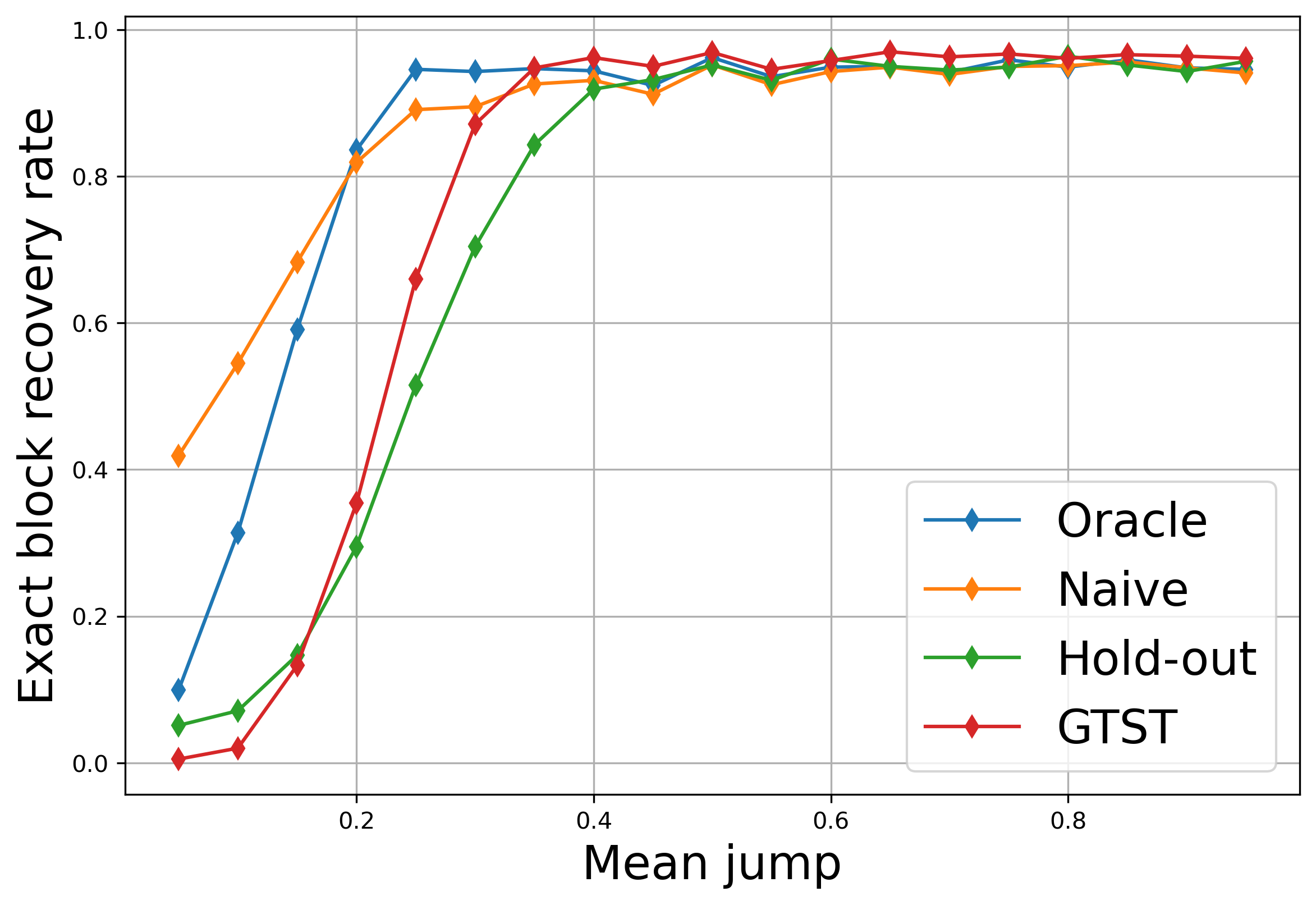}
    \caption{Exact block recovery rate}
    \label{fig:mean_support}
\end{subfigure}

\caption{
Performance of the \emph{post hoc} procedures as a function of the mean jump magnitude. (a) Empirical power with the probability that Condition~\eqref{eq:delta_n'} holds (dashed line); (b) empirical Type~I error rate; (c) exact block recovery rate. The dashed line in (a) indicates the frequency of Condition~\eqref{eq:delta_n'} and provides an upper bound on the power of the GTST procedure.
}
\label{fig:mean_jump_results}
\end{figure}
%
%
%
\begin{table}[tbp]
\centering
\caption{Numerical summary of the performance metrics for selected mean jump magnitudes. These values correspond to specific points from the curves presented in Figure~\ref{fig:mean_jump_results}.}
\label{tab:results_mean_jump}
\resizebox{\textwidth}{!}{
\begin{tabular}{c|ccc|ccc|ccc|ccc}
\toprule
 & \multicolumn{12}{c}{Method} \\
\cmidrule(lr){2-13}
Mean change 
& \multicolumn{3}{c}{GTST} 
& \multicolumn{3}{c}{Hold-out} 
& \multicolumn{3}{c}{Naive} 
& \multicolumn{3}{c}{Oracle} \\
\cmidrule(lr){2-4} \cmidrule(lr){5-7} \cmidrule(lr){8-10} \cmidrule(lr){11-13}
& Power & Type I & Recovery
& Power & Type I & Recovery
& Power & Type I & Recovery
& Power & Type I & Recovery \\
\midrule
$0.15$
& 0.14 & 0.01 & 0.13
& 0.15 & 0.05 & 0.15
& 0.80 & 0.25 & 0.68
& 0.63 & 0.05 & 0.59 \\
$0.3$
& 0.90 & 0.03 & 0.87
& 0.74 & 0.05 & 0.70
& 1.00 & 0.10 & 0.89
& 1.00 & 0.06 & 0.94 \\
$0.4$
& 1.00 & 0.04 & 0.96
& 0.97 & 0.05 & 0.92
& 1.00 & 0.07 & 0.93
& 1.00 & 0.06 & 0.94 \\
\bottomrule
\end{tabular}
}
\end{table}

In these evaluations, the oracle procedure serves as a benchmark, as inference is performed at the true change-point location. As expected, its power increases rapidly with signal strength, reaching values close to one while maintaining stable Type~I error control. Even in this idealized setting, however, small mean shifts remain difficult to assess. For example, when the \texttt{jump} is $0.15$, the oracle's power is limited to $0.63$, reflecting the inherent difficulty of the two-sample test under weak signal conditions. For larger shifts (\texttt{jump} $\ge 0.3$), the oracle attains a high power approaching $1$, with a Type~I error controlled near the nominal level.

The naive procedure, which applies the test directly to the segments estimated by the KCP algorithm, achieves high power on the affected block, even for small mean shifts (e.g., $0.80$ at \texttt{jump} $= 0.15$). This occurs because the change-point detection algorithm explicitly searches for the split that maximizes the difference between adjacent segments. Performing the statistical test directly at this purely data-driven location amplifies the perceived signal. However, this apparent gain is misleading, as it results in severe inflation of the Type~I error, particularly in low-signal regimes (reaching $0.25$ for \texttt{jump} $= 0.15$). This clearly illustrates the danger of selection bias, often referred to as ``double dipping'' in the literature. Using the exact same data to both locate the change-point and assess its significance inevitably leads to false positive detections.

The hold-out strategy successfully restores valid Type~I error control by decoupling detection from inference. However, this safety measure results in a significant loss of power for small mean shifts, dropping to $0.3$ for \texttt{jump} $= 0.2$. This loss is driven by two compounding factors: first, sample splitting halves the data available for the change-point detection step, which degrades the localization accuracy. Second, the subsequent two-sample test suffers from both this increased localization error and the reduced sample size, drastically reducing its statistical power. As the jump amplitude increases, the change-point is estimated more accurately, and the power of the hold-out procedure gradually approaches that of the oracle method.

To properly evaluate the proposed GTST procedure, we must account for the minimal spacing requirement discussed in Remark~\ref{rem:minimal_spacing}. Because GTST inherently retains the null hypothesis when Condition~\eqref{eq:delta_n'} is violated, its empirical power is limited by the frequency at which this condition is satisfied. Consequently, in the power plot (Figure~\ref{fig:mean_power}), the dashed line represents this empirical frequency, thereby acting as a structural upper bound for the achievable power of the GTST.

Remarkably, the empirical power of the GTST procedure closely tracks this upper bound, confirming that whenever the condition is satisfied, the grid-based test reliably detects the distributional change. The GTST method provides an effective compromise: it safely maintains the Type~I error below the nominal level, while its power increases rapidly as the signal strength grows, overtaking the hold-out procedure for \texttt{jump} $\ge 0.2$ (e.g., $0.90$ vs $0.74$ at \texttt{jump} $= 0.3$).

As a direct consequence of these individual power and Type~I error profiles, the exact block recovery rate (Figure~\ref{fig:mean_support}) improves as the mean jump increases. By successfully combining signal detection with strict false alarm control, both the hold-out and GTST procedures achieve increasingly high recovery rates as the signal strength grows, successfully identifying the exact affected block.

These results reveal the distinct power profiles of the hold-out and GTST procedures in the single change-point setting. In practical applications, the magnitude of the mean jump is inherently unknown. Therefore, the choice of a \emph{post hoc} procedure must be guided by the standard statistical testing paradigm: first, ensure strict control of the Type~I error, and second, maximize statistical power. As established in Remark~\ref{rem:single_cp}, the hold-out strategy does not rely on Condition~\eqref{eq:delta_n'} in this specific single change-point setting. Consequently, it maintains a slight power advantage in configurations with an extremely small mean shift. Conversely, as soon as the jump reaches a moderate amplitude, allowing Condition~\eqref{eq:delta_n'} to be satisfied with high probability, GTST rapidly overtakes the hold-out approach. By exploiting the full sample rather than splitting it, GTST achieves significantly higher efficiency. Ultimately, while the hold-out method offers slight robustness for barely detectable jumps, GTST emerges as the globally preferable procedure due to its superior power capacity as soon as the signal becomes practically discernible.

\begin{remark}[Type~I error control]
An additional observation concerns the empirical Type~I error levels of the hold-out and GTST procedures. Although the theoretical analysis provides an upper bound of $\alpha_0+\alpha_1$ (with $\alpha_0=\alpha_1=0.05$), the empirical Type~I error observed in the simulations remains close to $0.05$ for the hold-out method and is bounded below $0.05$ for GTST. 
This behavior is expected for the hold-out procedure in the single change-point setting. Since the data used for the \emph{post hoc} test are independent of the data used for change-point localization, the resulting two-sample test remains exactly controlled at level $\alpha_1$.

For GTST, the bound $\alpha_0+\alpha_1$ is a worst-case guarantee. In practice, the procedure aggregates multiple local tests, which makes it somewhat conservative. Moreover, GTST only performs the test when the separability condition~\eqref{eq:delta_n'} holds; by construction, this regime corresponds to a sufficiently large minimal segment relative to the localization error, where the change-points are well separated and accurately localized. The localization error term $\alpha_0$ therefore rarely contributes, and the empirical Type~I error of GTST remains close to, or slightly below, $\alpha_1$.

\end{remark}
\subsubsection{Estimation accuracy: The effect of the smallest segment size}
\label{subsec:estimation_accuracy}

Having examined the impact of signal strength, we now shift our focus to the challenges associated with estimation accuracy and noise. In this experiment, we investigate how the true change-point location, governed by the parameter $\underline{\Lambda}_{\tau^\star}$, influences the results. This parameter represents the relative size of the smallest segment as a proportion of the total time series length $n$. To isolate this effect, we fix the mean jump amplitude at a moderate level (\texttt{jump} $=0.3$) and vary $\underline{\Lambda}_{\tau^\star}$ by shifting the change-point location towards the boundaries of the sequence.

This manipulation directly controls the number of observations available in the smallest segment, thereby significantly impacting the effective noise level of the problem. As $\underline{\Lambda}_{\tau^\star}$ decreases, fewer data points are available to estimate the pre- and post-change distributions. Consequently, the empirical $\text{MMD}_b^2$ estimator becomes highly unstable: its variance increases drastically as the sample size shrinks. This heightened variability inherently degrades the precision of the initial change-point localization, which in turn diminishes the statistical power of the subsequent \emph{post hoc} tests. All other simulation parameters remain identical to those of the general experimental framework.

\begin{figure}[tbp]
\centering

\begin{subfigure}{0.48\textwidth}
    \centering
    \includegraphics[width=\textwidth]{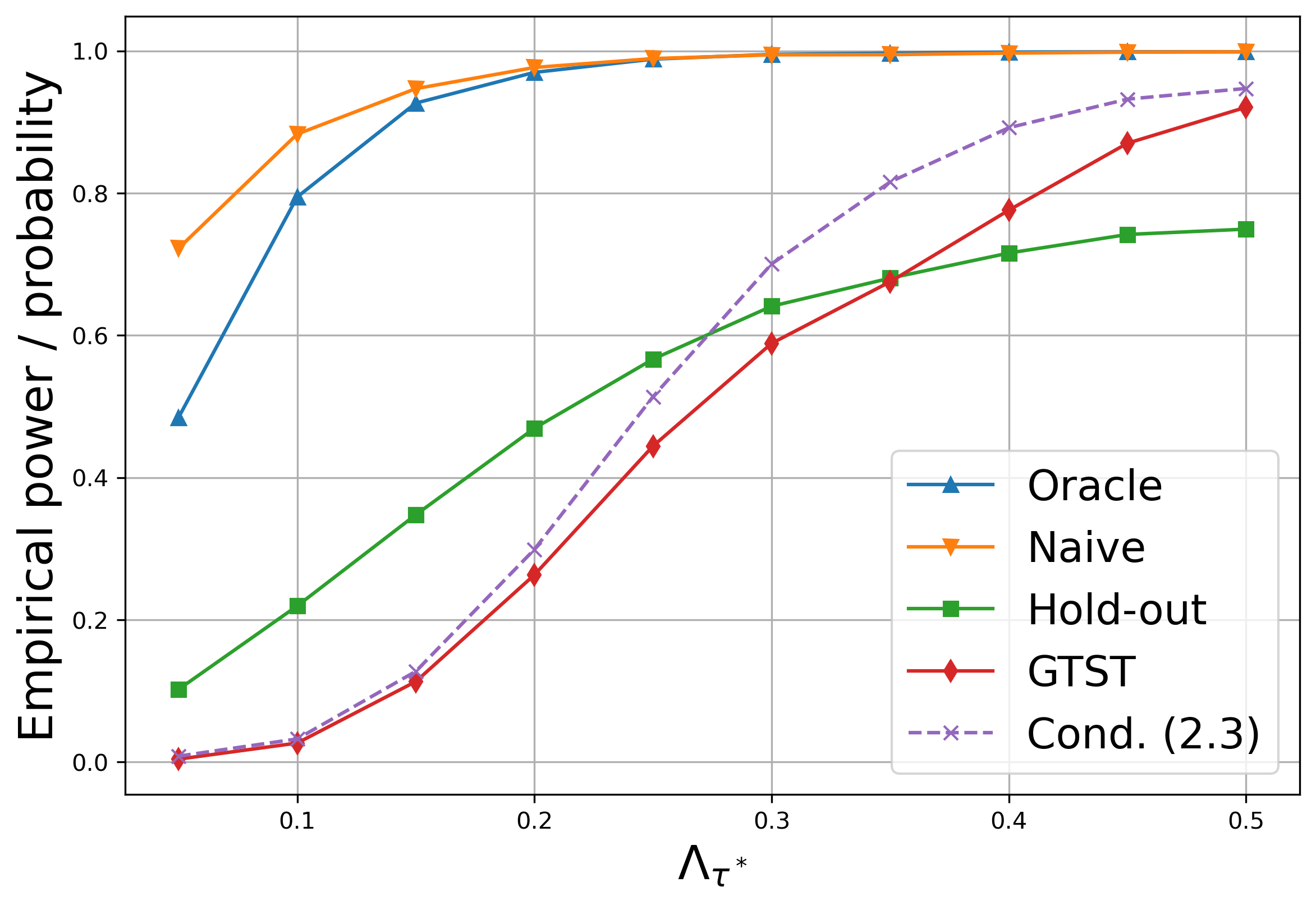}
    \caption{Power}
    \label{fig:pos_power}
\end{subfigure}
\hfill
\begin{subfigure}{0.48\textwidth}
    \centering
    \includegraphics[width=\textwidth]{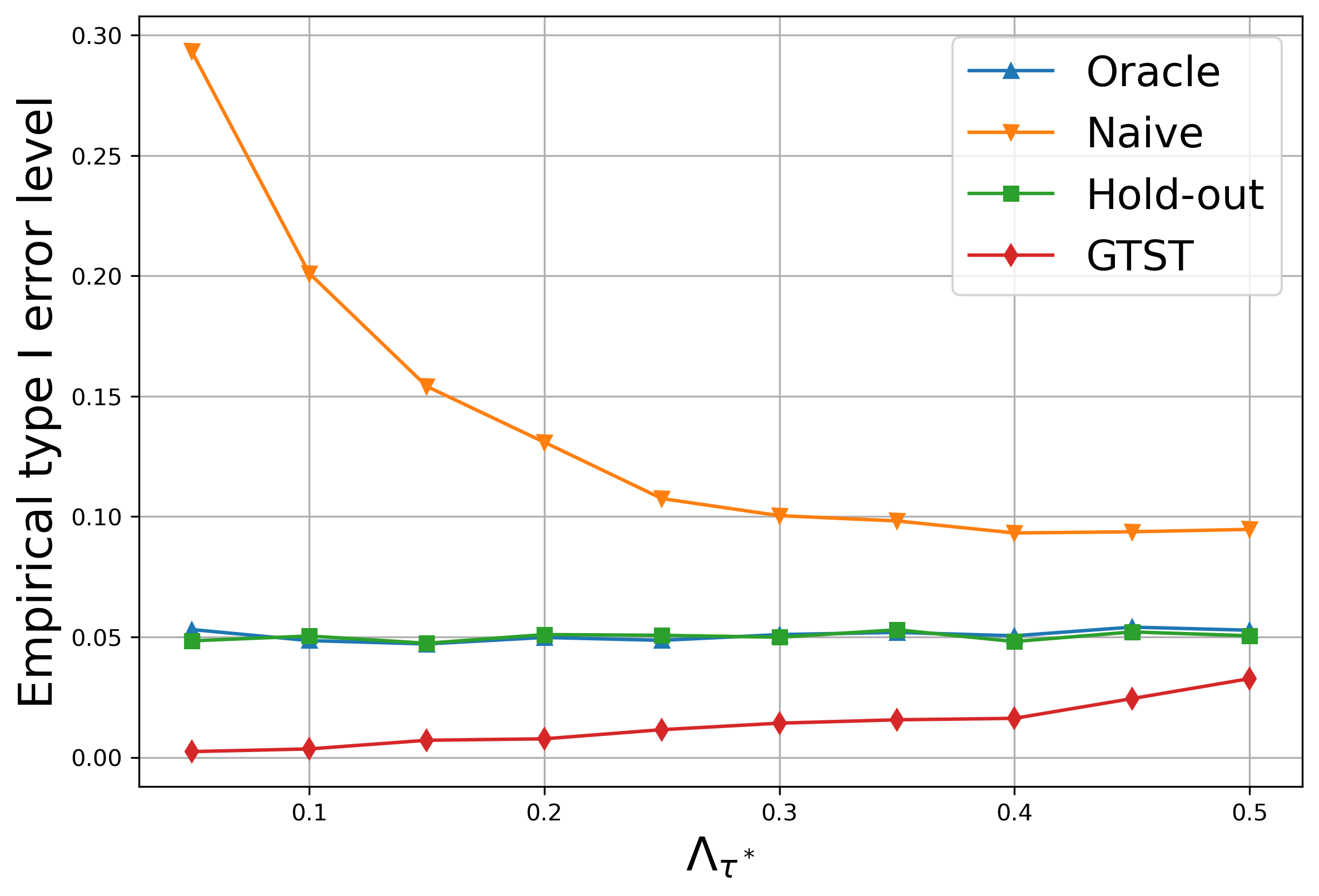}
    \caption{Type I error level}
    \label{fig:pos_level}
\end{subfigure}

\vspace{0.5cm}

\begin{subfigure}{0.5\textwidth}
    \centering
    \includegraphics[width=\textwidth]{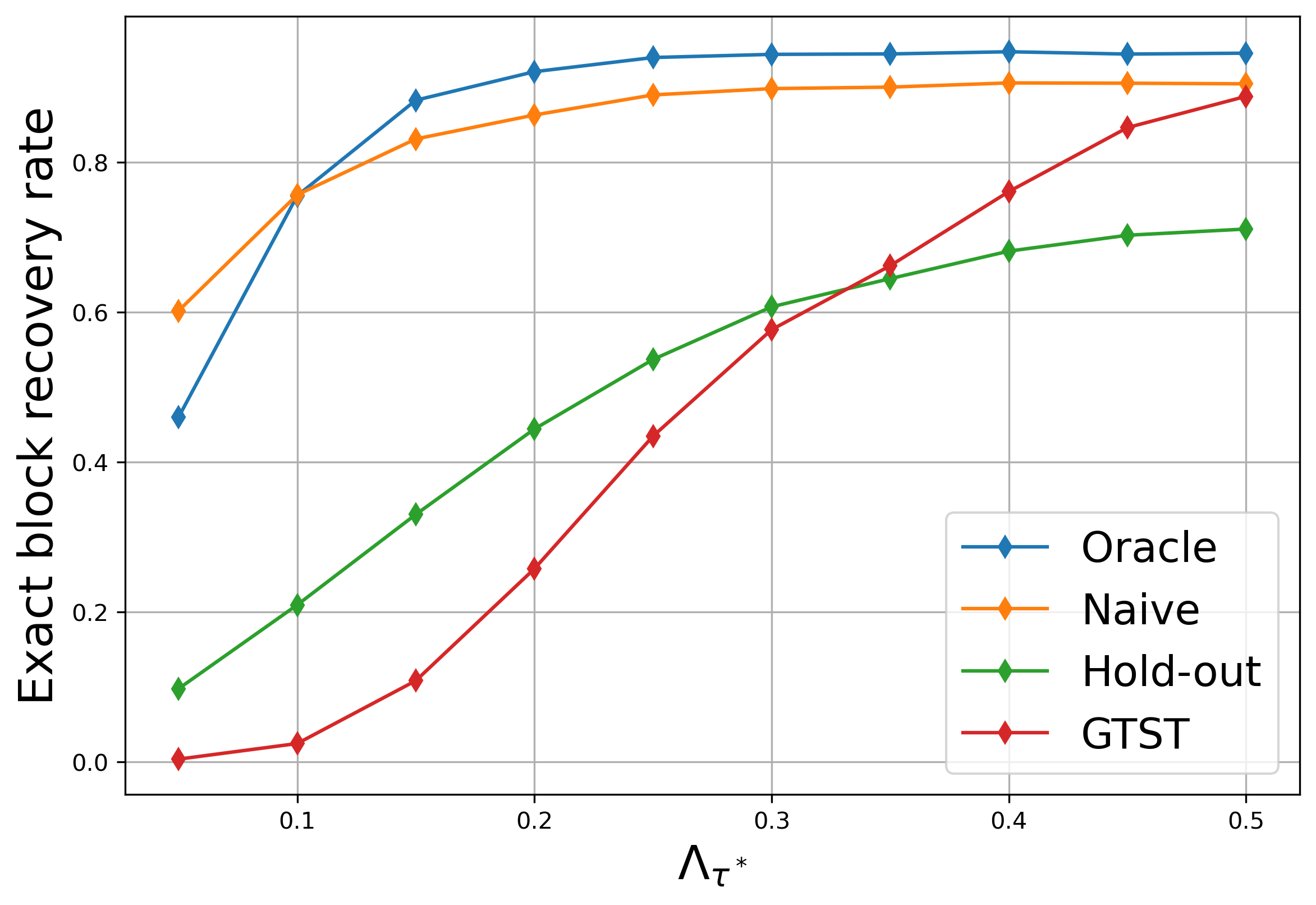}
    \caption{Exact block recovery rate}
    \label{fig:pos_support}
\end{subfigure}

\caption{
Performance of the \emph{post hoc} procedures as a function of the relative size of the smallest segment $\underline{\Lambda}_{\tau^\star}$, which represents the location of the change-point. The mean jump amplitude is fixed at $0.3$. (a) Empirical power, along with the empirical probability of satisfying Condition~\eqref{eq:delta_n'} (dashed line); (b) Empirical Type~I error level; (c) Exact block recovery rate. As in the previous experiment, the dashed line in panel (a) illustrates the empirical frequency with which Condition~\eqref{eq:delta_n'} is satisfied, acting as a structural upper bound for the power of the GTST approach.
}
\label{fig:pos_results}
\end{figure}

\begin{table}[tbp]
\centering
\caption{Numerical summary of the performance metrics for selected sizes of the smallest segment. These values correspond to specific points from the curves presented in Figure~\ref{fig:pos_results}.}
\label{tab:results_segment_size}
\resizebox{\textwidth}{!}{
\begin{tabular}{c|ccc|ccc|ccc|ccc}
\toprule
 & \multicolumn{12}{c}{Method} \\
\cmidrule(lr){2-13}
Segment size ($\underline{\Lambda}_{\tau^\star}$) 
& \multicolumn{3}{c}{GTST} 
& \multicolumn{3}{c}{Hold-out} 
& \multicolumn{3}{c}{Naive} 
& \multicolumn{3}{c}{Oracle} \\
\cmidrule(lr){2-4} \cmidrule(lr){5-7} \cmidrule(lr){8-10} \cmidrule(lr){11-13}
& Power & Type I & Recovery
& Power & Type I & Recovery
& Power & Type I & Recovery
& Power & Type I & Recovery \\
\midrule
$0.1$
& 0.03 & 0.01 & 0.02
& 0.22 & 0.05 & 0.21
& 0.88 & 0.20 & 0.76
& 0.79 & 0.05 & 0.76 \\
$0.3$
& 0.59 & 0.01 & 0.58
& 0.64 & 0.05 & 0.61
& 0.99 & 0.10 & 0.90
& 1.00 & 0.05 & 0.94 \\
$0.5$
& 0.92 & 0.03 & 0.89
& 0.75 & 0.05 & 0.71
& 1.00 & 0.09 & 0.91
& 1.00 & 0.05 & 0.95 \\
\bottomrule
\end{tabular}
}
\end{table}

\paragraph{Results.}
The performance metrics displayed in Figure~\ref{fig:pos_results} and summarized in Table~\ref{tab:results_segment_size} show that the location of the change-point strongly affects the performance of all \emph{post hoc} procedures. When the change-point lies near the center of the sequence ($\underline{\Lambda}_{\tau^\star} = 0.5$), both segments contain a sufficient number of observations. This leads to high power and accurate exact block recovery. In contrast, as the change-point approaches the boundaries, the size of the smallest segment decreases, causing the performance of all methods to deteriorate. This pattern is consistent across procedures and resembles the behavior observed in the mean-jump experiment. However, the limiting factor here is not the magnitude of the distributional change, but rather the amount of data available in the smallest segment. Small segments lead to increased variability in change-point localization and \emph{post hoc} testing, ultimately reducing detection power.

As seen earlier, GTST's power closely tracks the empirical probability that Condition~\eqref{eq:delta_n'} is satisfied. This condition implicitly requires the smallest segment to be sufficiently large relative to the localization error of the estimated change-point. When the change-point is close to a boundary, the number of observations on one side becomes too small, making accurate localization difficult and causing the condition to fail more frequently. When the condition fails, GTST by definition retains the null hypothesis $H_0$, which explains the observed loss of power. Conversely, once the smallest segment becomes sufficiently large, the condition is satisfied with high probability and GTST rapidly achieves high power while maintaining valid Type~I error control.

To conclude, these observations complement the mean-jump experiment by highlighting the two fundamental components of the effective signal-to-noise ratio in statistical detection. On one hand, the jump amplitude controls the true magnitude of the distributional shift, governing theoretical distances such as $\mathrm{MMD}(P,Q)=\|\mu_P-\mu_Q\|_{\mathcal H}$. On the other hand, the change-point location parameter $\underline{\Lambda}_{\tau^\star}$ governs the size of the smallest segment, dictating the effective sample size $n\,\underline{\Lambda}_{\tau^\star}$. Because the variance of the empirical estimator $\mathrm{MMD}_b^2$ strictly depends on these available sample sizes, successful detection requires both a sufficiently large theoretical signal and an adequate number of observations to estimate it reliably.

\subsection{Real data application}

To illustrate the practical application of the GTST and hold-out methods on real-world data, we use the \textit{Household Electric Power Consumption} dataset (UCI Machine Learning Repository) \cite{hebrail2012household}. This dataset contains electrical consumption measurements of a single household, originally recorded at a one-minute resolution.

Each observation at a given time \(t\) describes the energy state of the dwelling through several electrical variables. Among these, we focus on three sub-metering variables that measure the electricity consumption of specific circuits within the house. For notational convenience, we denote these variables as follows:

\begin{itemize}
    \item \(Sub_{kit}\) (originally \textit{Sub\_metering\_1}): consumption associated with the kitchen (oven, dishwasher, microwave),
    \item \(Sub_{ldy}\) (originally \textit{Sub\_metering\_2}): consumption associated with the laundry room and specific appliances (washing machine, tumble-dryer, refrigerator),
    \item \(Sub_{thm}\) (originally \textit{Sub\_metering\_3}): consumption related to the thermal equipment of the dwelling (electric water heater, heating or air conditioning).
\end{itemize}
The observed multivariate vector at time \(t\) is thus defined by:
\[
X_t = \left(Sub_{kit, t},\, Sub_{ldy, t},\, Sub_{thm, t}\right).
\]
These variables provide a \textit{functional decomposition} of the household's electricity consumption, making it easier to attribute detected structural breaks to specific usage patterns. To exploit this inherent structure within our framework, we partition the variables into two distinct blocks:
\begin{itemize}
    \item \textbf{Block 1}: \((Sub_{kit},\, Sub_{ldy})\), representing domestic uses.
    \item \textbf{Block 2}: \((Sub_{thm})\), representing persistent thermal uses.
\end{itemize}
This partition contrasts two fundamental types of energy behavior. The first block groups \textit{domestic uses}, characterized by occasional activations directly linked to human activity and appliance usage. The second block corresponds to \textit{thermal uses}, which are typically more regular and prone to exhibit persistent operating regimes (thermostatic cycles) independent of immediate human presence. With this structure, we can examine whether the change-points detected in the multivariate series are primarily driven by shifts in domestic behaviors or by changes in the thermal consumption of the dwelling.

\subsubsection{Experimental setup and visual expectations}
To precisely evaluate the detection capabilities of our approach, we isolate a specific continuous time window from Sunday at 20:00 to Monday at 07:00. At a 1-minute resolution, this period yields 661 observations. This specific timeframe provides an ideal testing ground for structural break detection. Human behavior during the day is inherently highly stochastic, generating constant distributional shifts. By focusing on the transition into the night period, we isolate automated, machine-driven behaviors from random human interactions.

Visually, as shown in Figure~\ref{fig:power_consumption_break}, this period reveals interesting deterministic patterns. In Block 1, while occasional human-driven activities like cooking ($Sub_{kit}$) remain completely flat at zero, the $Sub_{ldy}$ variable exhibits continuous, periodic spikes throughout the entire night. Given the physical appliances monitored by this sub-meter, this highly regular pulse is the classic signature of a refrigerator's compressor cycling on and off to maintain its internal temperature. Conversely, the thermal equipment ($Sub_{thm}$, Block 2), which operates at a high and persistent regime during the evening (e.g., heating the house or charging the water tank), abruptly drops to zero at exactly 02:11 AM once its automated cycle is complete. Consequently, we expect the algorithms to detect a structural break late in the night, and correctly attribute this break exclusively to the thermal regime shift (Block 2). Simultaneously, the algorithms should recognize that the intermittent background behavior of Block 1 (driven by the refrigerator) remains structurally unchanged across the estimated change-point.

\subsubsection{Detection results and comparison}
To evaluate the performance of our framework in this real-world scenario, we take two steps: first, we estimate change-points, and then we perform structural testing. We maintain the functional variables in their original physical scale (watt-hour) to preserve their relative energy magnitudes during distance computations.

For the estimation step, we apply the Kernel Change-Point (KCP) detection algorithm equipped with a Gaussian kernel. The kernel bandwidth is data-driven, selected via the standard median heuristic. This procedure successfully detects a single major structural break estimated at index $\widehat{\tau}$ (corresponding to 02:11 AM), as visually confirmed in Figure~\ref{fig:power_consumption_break}.

Following the estimation, we apply three different structural testing procedures to assign the break to the correct functional blocks: the naive method (testing directly on the estimated break point), the hold-out splitting method, and the proposed GTST procedure. For the localized tests (Hold-out and GTST), the uncertainty margin $\delta_n$ is calibrated using the constant $C_0 = 0.15$. Furthermore, the GTST procedure explores the search space using a dynamic grid step defined as $\eta = \max(\lfloor \delta_n / 2 \rfloor, 1)$. Finally, all structural tests are conducted at the nominal significance levels $\alpha_0 = \alpha_1 = 0.05$.

\begin{figure}[tbp]
    \centering
    \includegraphics[width=\textwidth]{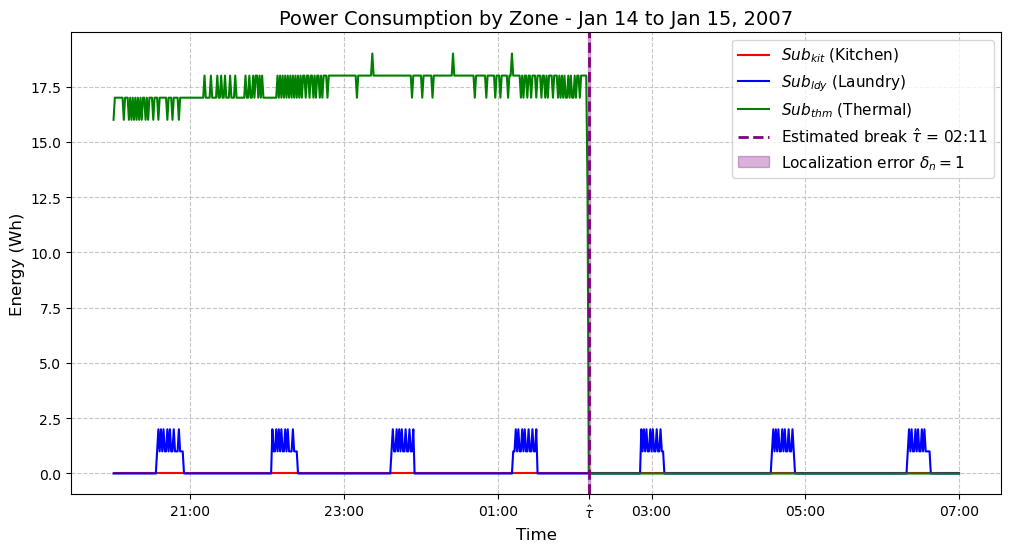}
    \caption{Multivariate time series of household electric power consumption from January 14 (20:00) to January 15 (07:00). The signal is decomposed into three functional variables: $Sub_{kit}$ (kitchen), $Sub_{ldy}$ (laundry), and $Sub_{thm}$ (thermal uses). The vertical dashed line indicates the estimated structural break $\hat{\tau}$ at 02:11, corresponding to the abrupt shutdown of the thermal equipment (Block 2). The shaded region represents the localization error margin ($\delta_n = 1$).}
    \label{fig:power_consumption_break}
\end{figure}
The core objective of this application is to evaluate whether the methods can correctly identify which functional blocks are responsible for the detected change, i.e., to assess the \textit{exact structural recovery} of the break. Based on visual inspection of the signal (Figure~\ref{fig:power_consumption_break}), the structural shift is entirely driven by the thermal usage ($Sub_{thm}$, Block 2) shutting down abruptly. Meanwhile, the behavior of the domestic block ($Sub_{kit}, Sub_{ldy}$, Block 1) remains stationary across the estimated change-point, maintaining its periodic cyclic pattern.

The results of the \emph{post hoc} attribution tests are presented in Table~\ref{tab:posthoc_power}.

\begin{table}[h!]
\centering
\caption{Post hoc tests around the detected structural break (Household Power Consumption)}
\label{tab:posthoc_power}
\begin{tabular}{lllll}
\toprule
\textbf{Block} & \textbf{Method} & \textbf{MMD Score} & \textbf{p-value} & \textbf{Conclusion} \\
\midrule
Block 1 ($Sub_{kit},Sub_{ldy}$) & Naive   & 0.017 & 0.516 & No change \\
Block 1 ($Sub_{kit},Sub_{ldy}$) & GTST    & 0.019 & 0.452 & No change \\
Block 1 ($Sub_{kit},Sub_{ldy}$) & Hold-out& 0.012 & 0.781 & No change \\
\midrule
Block 2 ($Sub_{thm}$) & Naive   & 1.371 & 0.001 & Change detected \\
Block 2 ($Sub_{thm}$) & GTST    & 1.365 & 0.001 & Change detected \\
Block 2 ($Sub_{thm}$) & Hold-out& 1.367 & 0.001 & Change detected \\
\bottomrule
\end{tabular}
\end{table}

As detailed in Table~\ref{tab:posthoc_power}, all three methods (Naive, Hold-out, and GTST) successfully achieve the exact recovery of the underlying structural change:
\begin{itemize}
    \item They correctly identify a significant distributional shift in \textbf{Block 2} (thermal usage), decisively rejecting the null hypothesis ($p = 0.001$).
    \item They correctly assign zero variation to \textbf{Block 1} (domestic uses), properly recognizing that the periodic pattern of the laundry cycles remains statistically unaffected by the structural break.
\end{itemize}

This application illustrates the practical use and interpretability of the framework on real data: in a clear scenario, all methods correctly attribute the break to the thermal block while leaving the domestic block unaffected. Since the break here is strong and well localized, the example does not discriminate between the methods — including the naive one — and the validity guarantees that distinguish them are established in the simulation section.

\section{Conclusion}
\label{sec:conclusion}

In this paper, we studied the problem of \emph{post hoc} explainability in multivariate change-point detection. While many modern detection procedures can reliably estimate the locations of structural changes in a time series, they typically fail to provide insight into which specific variables are responsible for these shifts. Addressing this attribution problem is crucial in numerous real-world applications, where understanding the physical or statistical origin of a detected change is just as important as pinpointing its occurrence.

To bridge this gap, we introduced a general framework for the \emph{post hoc} attribution of distributional changes. Our approach operates in a structured setting where the coordinates of the multivariate series are partitioned into predefined blocks, and the objective is to statistically determine whether the detected change manifests within each block. The fundamental challenge of this task lies in selection bias: performing inference at a purely data-driven location creates a strict dependency between the localization step and the subsequent statistical test, a phenomenon that systematically invalidates classical statistical guarantees.

To overcome this issue, we proposed two distinct procedures---the hold-out strategy and the Grid-based Two-Sample Test (GTST)---that enable valid \emph{post hoc} inference while accounting for the uncertainty induced by change-point estimation. Both procedures provide strong theoretical guarantees for Type~I error control under assumptions regarding localization accuracy. Through extensive simulation studies, we demonstrated that while a naive testing approach systematically fails to control false discoveries, both of our proposed methods achieve valid error control. The results further highlighted the trade-offs between statistical power and sample efficiency: while the hold-out procedure maintains slight robustness for barely detectable signals, GTST rapidly emerges as the most powerful strategy as soon as the separability condition~\eqref{eq:delta_n'} is satisfied.

A key strength of the proposed framework is its broad flexibility. Developed within a nonparametric setting, the methodology can successfully identify a wide range of structural shifts, including mean, variance, and covariance changes. Furthermore, the framework is largely agnostic to the specific algorithms employed: it can be seamlessly combined with any two-sample test that provides valid level control, alongside any change-point detection procedure satisfying basic localization guarantees. If a practitioner aims to detect a specific type of distributional change, our method becomes particularly powerful when combined with a tailor-made test statistic or a domain-specific kernel.

Overall, this work formalizes a statistically principled approach to \emph{post hoc} explainability in multivariate change-point detection. The proposed framework contributes to more interpretable change-point analysis and opens new directions at the intersection of change-point detection and interpretable statistical inference.

Several directions for future work may be considered. First, establish theoretical guarantees for the power of the GTST and hold-out procedures. Second, extend the framework to automatically identify the coordinates that contribute to the detected changes in higher-dimensional settings. Finally, study the case where the number of change-points $\kappa^\star$ is unknown and must be estimated from the data.

\section*{Acknowledgments}
The authors would like to thank Gilles Blanchard for insightful discussions.

\begin{appendix}
\section{Proofs of the Technical Lemmas}\label{app:all_proofs}
\subsection{Proof of Lemma \ref{lem:IC}}\label{app:proof_IC}
\begin{proof}[Proof of Lemma \ref{lem:IC}]
Points 1., 2., and 3. directly follow from Lemma 3.1 in \cite{Garreau2018}.
For 4., for all $i$ in $\llbracket 0,\kappa^\star-1 \rrbracket$ 
we have
\[
n\, \cdot \, \underline{\Lambda}_{\tau^\star} \leq |\tau^\star_i - \tau^\star_{i+1}|.
\]
By adding $\hat{\tau}_i$ and $\hat{\tau}_{i+1}$ and applying two triangle inequalities, with  Assumption \ref{assum:detection_event} we obtain
\[
n\, \cdot \, \underline{\Lambda}_{\tau^\star} \leq |\tau^\star_i - \tau^\star_{i+1}| \leq |\tau^\star_i - \hat{\tau}_i| + |\hat{\tau}_i - \hat{\tau}_{i+1}|+ |\hat{\tau}_{i+1} - \tau^\star_{i+1}|  \leq |\hat{\tau}_i - \hat{\tau}_{i+1} | + 2 \delta_n.
\]
If $\kappa^\star>2$, for all $i$ in $\llbracket 0,\kappa^\star-1 \rrbracket$ we have:
\[
|\hat{\tau}_{i+1} - \hat{\tau}_i| \geq n\, \cdot \, \underline{\Lambda}_{\tau^\star}  - 2\delta_n.
\]
If there is only one change-point ($\kappa^\star = 2$), under assumption \eqref{eq:delta_n} we directly have:
\[
|\hat\tau_1-\hat\tau_0|
=|\hat\tau_1-0|
\ge |\tau^\star_1-0| - |\hat\tau_1-\tau^\star_1|
\ge (\tau^\star_1-\tau^\star_0)-\delta_n
\ge n\,\underline{\Lambda}_{\tau^\star}-\delta_n
\]
and
\[
|\hat\tau_2-\hat\tau_1|
=|n-\hat\tau_1|
\ge |n-\tau^\star_1| - |\hat\tau_1-\tau^\star_1|
\ge (\tau^\star_2-\tau^\star_1)-\delta_n
\ge n\,\underline{\Lambda}_{\tau^\star}-\delta_n.
\]
Consequently, for all $i \in \llbracket 0 , \kappa^\star - 1\rrbracket$, we have $|\hat{\tau}_{i+1} - \hat{\tau}_i| \geq n\, \cdot \, \underline{\Lambda}_{\tau^\star}  - \delta_n'$.

\end{proof}

\subsection{Proof of Lemma \ref{lem:theorM1_bis}}\label{app:proof_theorM1_bis}
\begin{proof}[Proof of Lemma \ref{lem:theorM1_bis}]
\noindent\textbf{(1) Minimal lengths.}
By definition of $\underline{\Lambda}_{\tau^\star}$, for every $1\le i\le \kappa^\star$ we have:
\[
\tau_i^\star - \tau_{i-1}^\star \;\ge\; n \,\underline{\Lambda}_{\tau^\star}.
\]
Using the inequality $\lfloor a \rfloor - \lfloor b \rfloor \ge \lfloor a - b \rfloor$,
established in Lemma~\ref{lem:partentiere} (Appendix), and taking
$a = \tau_i^\star/\eta$ and $b = \tau_{i-1}^\star/\eta$, we obtain:
\[
\Bigl\lfloor \tfrac{\tau_i^\star}{\eta} \Bigr\rfloor 
-
\Bigl\lfloor \tfrac{\tau_{i-1}^\star}{\eta} \Bigr\rfloor
\;\ge\;
\Bigl\lfloor \tfrac{\tau_i^\star - \tau_{i-1}^\star}{\eta} \Bigr\rfloor
\;\ge\;
\Bigl\lfloor \tfrac{n\,\underline{\Lambda}_{\tau^\star}}{\eta} \Bigr\rfloor .
\tag{*}
\]
\noindent\emph{Case $i_0>1$.}
The left and right projected oracle segments are
\[
 S^{\star-}_{i_0,\eta}
 = \llbracket (\lfloor\tau_{i_0-1}^\star/\eta\rfloor+1)\eta,\;
            \lfloor\tau_{i_0}^\star/\eta\rfloor \eta \rrbracket
\quad \textrm{and} \quad
 S^{\star+}_{i_0,\eta}
 = \llbracket (\lfloor\tau_{i_0}^\star/\eta\rfloor+1)\eta,\;
            \lfloor\tau_{i_0+1}^\star/\eta\rfloor \eta \rrbracket,
\]
with lengths 
\[
|S^{\star-}_{i_0,\eta}|
 = \Bigl(
     \Bigl\lfloor\tfrac{\tau_{i_0}^\star}{\eta}\Bigr\rfloor
     -\bigl(\Bigl\lfloor\tfrac{\tau_{i_0-1}^\star}{\eta}\Bigr\rfloor+1\bigr)
   \Bigr)\eta + 1
\quad \textrm{and} \quad
|S^{\star+}_{i_0,\eta}|
 = \Bigl(
     \Bigl\lfloor\tfrac{\tau_{i_0+1}^\star}{\eta}\Bigr\rfloor
     -\bigl(\Bigl\lfloor\tfrac{\tau_{i_0}^\star}{\eta}\Bigr\rfloor+1\bigr)
   \Bigr)\eta + 1.
\]
Using $(*)$ with $i_0$ for $S^{\star-}_{i_0,\eta}$ and
$i_0+1$ for $S^{\star+}_{i_0,\eta}$, we obtain
\[
| S^{\star-}_{i_0,\eta}|,
\;| S^{\star+}_{i_0,\eta}|
\;\ge\;
\Bigl(\Bigl\lfloor\tfrac{n\,\underline{\Lambda}_{\tau^\star}}{\eta}\Bigr\rfloor -1\Bigr)\eta + 1
 \;=\; \ell_\eta.
\]
\noindent\emph{Case $i_0=1$.}
For the right segment we still have, by the argument above with $i_0 +1$, $|S^{\star+}_{1,\eta}| \;\ge\; \ell_\eta.$ For the left segment, the projected interval is $ S^{\star-}_{1,\eta}
 = \llbracket 1,\lfloor\tau_1^\star/\eta\rfloor \eta \rrbracket,$
so that $ | S^{\star-}_{1,\eta}|
 = \lfloor\tfrac{\tau_1^\star}{\eta}\rfloor \eta.$ Since $\tau_1^\star = \tau_1^\star-\tau_0^\star \ge n\,\underline{\Lambda}_{\tau^\star}$, we get $\lfloor\tfrac{\tau_1^\star}{\eta}\rfloor 
 \;\ge\; \Bigl\lfloor\tfrac{n\,\underline{\Lambda}_{\tau^\star}}{\eta}\Bigr\rfloor,$ hence $ | S^{\star-}_{1,\eta}|
 \;\ge\; \Bigl\lfloor\tfrac{n\,\underline{\Lambda}_{\tau^\star}}{\eta}\Bigr\rfloor \eta
 \;=\; \ell_\eta^-.$
\medskip

\noindent \textbf{(2) Inclusion inside the oracle segments.}
By definition of the projected endpoints,
\[
\Bigl\lfloor\tfrac{\tau^\star_{i_0}}{\eta}\Bigr\rfloor \eta \le \tau^\star_{i_0},
\qquad 
\left(\Bigl\lfloor\tfrac{\tau^\star_{i_0-1}}{\eta}\Bigr\rfloor 
      + \mathbf{1}_{\{i_0>1\}}\right)\eta
      + \mathbf{1}_{\{i_0=1\}}
\;\ge\;
\tau^\star_{i_0-1}+1,
\]
which shows that
$
S^{\star-}_{i_0,\eta} \subseteq S^{\star-}_{i_0}.
$
The proof for 
$S^{\star+}_{i_0,\eta} \subseteq S^{\star+}_{i_0}$
is identical.
\medskip

\noindent \textbf{(3) Inclusion in the restricted grid.} 
We have
\[
\left(S^{\star-}_{i_0, \eta}, S^{\star+}_{i_0, \eta}\right) 
=  
\left(\widetilde{S}^{-}_{i_0, \beta, \gamma, \zeta}, \widetilde{S}^{+}_{i_0,  \beta, \gamma, \zeta}\right),
\]
with
\begin{align*}
    \beta =& \left\lfloor \frac{\tau^\star_{i_0-1}}{\eta} \right\rfloor - \Biggl\lfloor \frac{\widehat\tau_{i_0-1}}{\eta}
        - \frac{\delta_n}{\eta}\,\mathbf 1_{\{i_0-1\neq 0\}}
      \Biggr\rfloor,\\
    \gamma =& \left\lfloor \frac{\tau^\star_{i_0}}{\eta} \right\rfloor  
    - \left\lfloor \frac{\hat{\tau}_{i_0}}{\eta} - \frac{\delta_n}{\eta} \right\rfloor ,\\
    \zeta =&  \left\lfloor \frac{\tau^\star_{i_0 + 1}}{\eta} \right\rfloor  
    - \Biggl\lfloor
            \frac{\widehat\tau_{i_0+1}}{\eta}
            - \frac{\delta_n}{\eta}\,\mathbf 1_{\{i_0+1\neq \kappa^\star\}}
          \Biggr\rfloor,
\end{align*}
and using $\hat{\tau}_0 = \tau^\star_0$, $\hat{\tau}_{\kappa^\star} = \tau^\star_{\kappa^\star}$.
Moreover, for all $i \in \{i_0 - 1, i_0, i_0 +1\}$, we have (using Lemma \ref{lem:IC}):
\[
\hat{\tau}_{i} - \delta_n\mathbf{1}_{\{i \notin \{0,\kappa^\star\}\}} 
\leq 
\tau_{i}^\star 
\leq 
\hat{\tau}_{i} + \delta_n\mathbf{1}_{\{i \notin \{0,\kappa^\star\}\}},
\]
which implies (using Lemma~\ref{lem:partentiere}) that
\[
\left\lfloor 
\frac{\hat{\tau}_{i}}{\eta} - \frac{\delta_n}{\eta}\mathbf{1}_{\{i \notin \{0,\kappa^\star\}\}} 
\right\rfloor 
\leq 
\left\lfloor \frac{\tau_{i}^\star}{\eta} \right\rfloor 
\leq 
\left\lfloor 
\frac{\hat{\tau}_{i}}{\eta} + \frac{\delta_n}{\eta}\mathbf{1}_{\{i \notin \{0,\kappa^\star\}\}} 
\right\rfloor ,
\]
and therefore
\[
\left\lfloor \frac{\tau_{i}^\star}{\eta} \right\rfloor 
\leq 
\left\lfloor 
\frac{\hat{\tau}_{i}}{\eta} - \frac{\delta_n}{\eta}\mathbf{1}_{\{i \notin \{0,\kappa^\star\}\}} 
\right\rfloor
+
\left\lceil \frac{2\delta_n}{\eta}\mathbf{1}_{\{i \notin \{0,\kappa^\star\}\}} \right\rceil .
\]
The localization assumption $|\widehat\tau_i - \tau^\star_i| \le \delta_n$ gives, for every $i\in\{i_0-1,i_0,i_0+1\}$,
\[
0 
\leq 
\left\lfloor \frac{\tau_{i}^\star}{\eta} \right\rfloor 
-
\left\lfloor 
\frac{\hat{\tau}_{i}}{\eta} - \frac{\delta_n}{\eta}\mathbf{1}_{\{i \notin \{0,\kappa^\star\}\}} 
\right\rfloor
\leq
\left\lceil \frac{2\delta_n}{\eta}\mathbf{1}_{\{i \notin \{0,\kappa^\star\}\}} \right\rceil
\leq
\left\lceil \frac{2\delta_n}{\eta} \right\rceil = m.
\]
Hence the indices $\beta,\gamma,\zeta$ used to define 
$(S^{\star,-}_{i_0,\eta},S^{\star+}_{i_0,\eta})$
all belong to $\{0,\dots,m\}$, so the pair $(S^{\star-}_{i_0,\eta}, S^{\star+}_{i_0,\eta})$ belongs to 
$\widetilde{\mathcal{E}}_{i_0, \eta}$.
Combined with point~(1), this implies
$
(S^{\star-}_{i_0,\eta}, S^{\star+}_{i_0,\eta})
\in \widetilde{\underline{\mathcal{E}}}_{i_0, \eta}.
$
\end{proof}

\subsection{Proof of Lemma \ref{lem:pi+}}\label{app:proof_pi+}
\begin{proof}[Proof of Lemma \ref{lem:pi+}]
For each $0 \le r \le \kappa^\star-1$, the minimal coverage condition ensures that $I_{\mathrm{det}}^{(r)} = I_{\mathrm{det}} \cap \llbracket \tau_r^\star + 1, \tau_{r+1}^\star \rrbracket$ is nonempty. Therefore, the projection $\pi^+(\tau_r^\star + 1)$ exists and lies in $I_{\mathrm{det}}^{(r)}$.
Since the oracle segments $S_r^\star = \llbracket \tau_r^\star + 1, \tau_{r+1}^\star \rrbracket$ are disjoint and ordered, and since each projection $\pi^+(\tau_r^\star + 1)$ lies in its respective segment, the sequence of projections is strictly increasing.
Finally, $I_{\mathrm{det}}$ is nonempty, so $\pi^-(n)$, defined as the last element of $I_{\mathrm{det}}$, is well-defined.
\end{proof}

\subsection{Proof of Lemma \ref{lem:pi-}}\label{app:proof_pi-}
\begin{proof}[Proof of Lemma \ref{lem:pi-}]
By definition, the projections $\pi^-(\tau^\star_r)$ and $\pi^+(\tau^\star_r+1)$ both belong to the detection set $I_{\mathrm{det}}$. Suppose, for the sake of contradiction, that there exists an index $i \in I_{\mathrm{det}}$ strictly between them, such that
\[
\pi^-(\tau^\star_r) < i < \pi^+(\tau^\star_r+1).
\]
By definition of the projection operators, $\pi^-(\tau^\star_r)$ is the largest element in $I_{\mathrm{det}}$ that is less than or equal to $\tau^\star_r$. Thus, since $i \in I_{\mathrm{det}}$, the inequality $i > \pi^-(\tau^\star_r)$ implies $i > \tau^\star_r$. Similarly, $\pi^+(\tau^\star_r+1)$ is the smallest element in $I_{\mathrm{det}}$ that is greater than or equal to $\tau^\star_r+1$, meaning that $i < \pi^+(\tau^\star_r+1)$ implies $i < \tau^\star_r+1$. Combining these deductions yields:
\[
\tau^\star_r < i < \tau^\star_r+1.
\]
Since $i$ and $\tau^\star_r$ are both integers, such a strict framing is impossible. This contradiction proves that no element of $I_{\mathrm{det}}$ can lie strictly between $\pi^-(\tau^\star_r)$ and $\pi^+(\tau^\star_r+1)$. Consequently, these two elements are consecutive in $I_{\mathrm{det}}$. 

By definition, we have $j_{\tau^{\star,\mathrm{det}}_{r}+1} = \pi^+(\tau^\star_r+1)$. Because $\pi^-(\tau^\star_r)$ is its immediate predecessor in the ordered set $I_{\mathrm{det}}$, it must correspond to the index $j_{\tau^{\star,\mathrm{det}}_{r}}$. We thus conclude that $j_{\tau^{\star,\mathrm{det}}_{r}} = \pi^-(\tau^\star_r)$.
\end{proof}

\subsection{Proof of Lemma \ref{lem:pi_det}}\label{app:proof_pi_det}
\begin{proof}[Proof of Lemma \ref{lem:pi_det}]
From the previous lemma, the projections $\pi^+(\tau_r^\star+1)$ are well-defined and strictly increasing for all $0 \le r \le \kappa^\star - 1$. By definition of these projected change-points, this implies
\[
\tau^{\star,\mathrm{det}}_0 < \tau^{\star,\mathrm{det}}_1 < \cdots < \tau^{\star,\mathrm{det}}_{\kappa^\star},
\]
so that each interval $\llbracket \tau^{\star,\mathrm{det}}_r + 1,\, \tau^{\star,\mathrm{det}}_{r+1} \rrbracket$ is nonempty.
Moreover, the reindexed series $\widetilde{Z}_i = Z_{j_i}$ contains only indices from $I_{\mathrm{det}}$. By construction, the indices $j_{\tau^{\star,\mathrm{det}}_r + 1}, \dots, j_{\tau^{\star,\mathrm{det}}_{r+1}}$ all belong to $I_{\mathrm{det}}^{(r)} \subset \llbracket \tau_r^\star + 1, \tau_{r+1}^\star \rrbracket$. Thus, the corresponding observations $(\widetilde{Z}_i)_{i \in \llbracket \tau^{\star,\mathrm{det}}_r + 1,\, \tau^{\star,\mathrm{det}}_{r+1} \rrbracket}$ are i.i.d.\ with distribution $P_{(r+1)}^\phi, \, \forall r$ such that $0 \le r \le \kappa^\star - 1$.
\end{proof}

\subsection{Proof of Lemma \ref{lem:tau_det}}\label{app:proof_tau_det}
\begin{proof}[Proof of Lemma \ref{lem:tau_det}]
First, let us establish a general bound on the index spacing. For any $1 \leq a\leq b \leq n_{\mathrm{det}}$, we have:
\[
j_b - j_a = \sum_{i = a}^{b-1} (j_{i+1} - j_i) \le \sum_{i = a}^{b-1} (q_{\mathrm{det}}+ 1) = (q_{\mathrm{det}}+ 1) (b - a).
\]
By symmetry, this implies that for all $1 \le a, b \le n_{\mathrm{det}}$, the following inequality holds:
\begin{equation}
\label{eq:spacing}
|j_b - j_a| \le (q_{\mathrm{det}}+ 1)\, |b - a|.
\end{equation}

Second, we bound the distance between the true change-point $\tau^\star_r$ and its exact detection projection $j_{\tau^{\star,\mathrm{det}}_{r}}$. By definition, we know that $j_{\tau^{\star,\mathrm{det}}_{r}} = \pi^-(\tau^\star_{r})$, which directly implies:
\[
j_{\tau^{\star,\mathrm{det}}_{r}} \leq \tau^\star_{r}.
\]
Furthermore, the next detection index $j_{\tau^{\star,\mathrm{det}}_{r}+1}$ is strictly greater than $\tau^\star_{r}$. Using the maximum gap property $j_{k+1} - j_k \le q_{\mathrm{det}} + 1$, we obtain:
\[
\tau^\star_{r} \leq j_{\tau^{\star,\mathrm{det}}_{r}+1} - 1 \leq j_{\tau^{\star,\mathrm{det}}_{r}} + (q_{\mathrm{det}} + 1) - 1 = j_{\tau^{\star,\mathrm{det}}_{r}} + q_{\mathrm{det}}.
\]
This provides the strict framing:
\begin{equation}
\label{eq:framing_tau}
0 \leq \tau^\star_{r} - j_{\tau^{\star,\mathrm{det}}_{r}} \leq q_{\mathrm{det}}.
\end{equation}

Finally, we introduce the estimated detection break $\hat{\tau}^{\mathrm{det}}_{r}$. We can decompose the total error as:
\[
\tau^{\star}_{r} - j_{\hat{\tau}^{\mathrm{det}}_{r}} = (\tau^{\star}_{r} - j_{\tau^{\star,\mathrm{det}}_{r}}) + (j_{\tau^{\star,\mathrm{det}}_{r}} - j_{\hat{\tau}^{\mathrm{det}}_{r}}).
\]
From inequality \eqref{eq:spacing} and assumption \eqref{eq:delta_nHoldout}, we can bound the second term:
\[
|j_{\tau^{\star,\mathrm{det}}_{r}} - j_{\hat{\tau}^{\mathrm{det}}_{r}}| \le (q_{\mathrm{det}}+1) \left| \tau^{\star,\mathrm{det}}_{r} - \hat{\tau}^{\mathrm{det}}_{r} \right| \le (q_{\mathrm{det}}+1)\delta_{n_{\mathrm{det}}} = \delta_{n}^{\mathrm{tot}}.
\]
Thus, $-\delta_{n}^{\mathrm{tot}} \le j_{\tau^{\star,\mathrm{det}}_{r}} - j_{\hat{\tau}^{\mathrm{det}}_{r}} \le \delta_{n}^{\mathrm{tot}}$. Combining this with the framing \eqref{eq:framing_tau} yields the desired lower and upper bounds:
\[
\tau^{\star}_{r} - j_{\hat{\tau}^{\mathrm{det}}_{r}} \ge 0 - \delta_{n}^{\mathrm{tot}} = -\delta_{n}^{\mathrm{tot}},
\]
\[
\tau^{\star}_{r} - j_{\hat{\tau}^{\mathrm{det}}_{r}} \le q_{\mathrm{det}} + \delta_{n}^{\mathrm{tot}}.
\]
This completes the proof.
\end{proof}

\subsection{Proof of Lemma \ref{lem:S_S+}}\label{app:proof_S_S+}
\begin{proof}[Proof of Lemma \ref{lem:S_S+}]
We proceed by establishing the left and right bounds for the segments $S_{i_0}^{-,\mathrm{inf}}$ and $S_{i_0}^{+,\mathrm{inf}}$ separately.

\medskip\noindent
\textbf{Left and right bounds for $S_{i_0}^{-,\mathrm{inf}}$:} \\
By definition, we can write $S_{i_0}^{-,\mathrm{inf}} = \llbracket L^-_{i_0}\, , \, R^-_{i_0} \rrbracket \cap I_{\mathrm{inf}}$, where the boundaries are given by:
\[
L^-_{i_0} := j_{\hat{\tau}^{\mathrm{det}}_{i_0-1}} + 1 + (\delta_{n}^{\mathrm{tot}}+q_{\mathrm{det}})\cdot\mathbf{1}_{\{i_0 - 1 \neq 0\}}, \quad \text{and} \quad R^-_{i_0} := j_{\hat{\tau}^{\mathrm{det}}_{i_0}+1}-1.
\]
First, we show that $L^-_{i_0} \ge \tau^\star_{i_0-1}+1$. We distinguish two cases:
\begin{itemize}
    \item \textbf{Case $i_0 = 1$:} The indicator function is zero, yielding $L^-_{1} = j_{\hat{\tau}^{\mathrm{det}}_0} + 1 = j_0 + 1 = 1$. Since $\tau^\star_0 = 0$, we have $L^-_{1} = \tau^\star_0 + 1$, which satisfies the condition.
    \item \textbf{Case $i_0 \ge 2$:} By Lemma~\ref{lem:tau_det}, we know that $\tau^\star_{i_0-1} - j_{\hat{\tau}^{\mathrm{det}}_{i_0-1}} \le \delta_{n}^{\mathrm{tot}} + q_{\mathrm{det}}$. Rearranging this inequality directly yields:
    \[
    \tau^\star_{i_0-1} + 1 \le j_{\hat{\tau}^{\mathrm{det}}_{i_0-1}} + 1 + \delta_{n}^{\mathrm{tot}} + q_{\mathrm{det}} = L^-_{i_0}.
    \]
\end{itemize}
Thus, in all cases, we conclude that $L^-_{i_0} \ge \tau^\star_{i_0-1}+1$.

Next, we prove that $R^-_{i_0} \le \tau^\star_{i_0+1}$. Using the strict monotonicity of the sequence $(j_k)$ along with assumption~\eqref{eq:delta_nHoldout} and \eqref{eq:delta_n'Holdout}, we distinguish two cases depending on the true number of segments $\kappa^\star$:
\begin{itemize}
    \item \textbf{Case $\kappa^\star = 2$:} We have $1 \leq i_0 \leq \kappa^\star-1$, so $i_0 = 1$. Since the sequence $(j_k)$ is bounded by the total sample size $n$, we obtain:
    \[
    R^-_{1} = j_{\hat{\tau}^{\mathrm{det}}_{1}+1} - 1 \le n = \tau^\star_2 = \tau^\star_{i_0+1}.
    \]
    
    \item \textbf{Case $\kappa^\star \ge 3$:} We use the estimation error bound $\hat{\tau}^{\mathrm{det}}_{i_0} \le \tau^{\star,\mathrm{det}}_{i_0} + \delta_{n_{\mathrm{det}}}$, which yields $R^-_{i_0} \le j_{\tau^{\star,\mathrm{det}}_{i_0} + \delta_{n_{\mathrm{det}}} + 1} - 1$. The analysis splits based on $\delta_{n_{\mathrm{det}}}$:
    \begin{itemize}
        \item If $\delta_{n_{\mathrm{det}}} \ge 1$, then by definition $\delta_{n_{\mathrm{det}}} + 1 \le \delta'_{n_{\mathrm{det}}}$. We have:
        \[
        R^-_{i_0} \le j_{\tau^{\star,\mathrm{det}}_{i_0} + \delta_{n_{\mathrm{det}}} + 1} - 1 \le j_{\tau^{\star,\mathrm{det}}_{i_0} + \delta'_{n_{\mathrm{det}}}} - 1 \le j_{\tau^{\star,\mathrm{det}}_{i_0+1}} - 1.
        \]
        
        \item If $\delta_{n_{\mathrm{det}}} = 0$, this yields:
        \[
        R^-_{i_0} \le j_{\tau^{\star,\mathrm{det}}_{i_0} + \delta_{n_{\mathrm{det}}} + 1} - 1 = j_{\tau^{\star,\mathrm{det}}_{i_0} + 1} - 1 \le j_{\tau^{\star,\mathrm{det}}_{i_0+1}} - 1.
        \]
    \end{itemize}
    This last inequality holds because $\tau^{\star,\mathrm{det}}_{i_0} + 1 \le \tau^{\star,\mathrm{det}}_{i_0+1}$ according to Lemma~\ref{lem:pi+}, and the sequence of indices $(j_k)$ is strictly increasing. In both sub-cases, since $j_{\tau^{\star,\mathrm{det}}_{i_0+1}} = \pi^+(\tau^\star_{i_0+1}+1)$, we deduce:
    \[
    R^-_{i_0} \le j_{\tau^{\star,\mathrm{det}}_{i_0+1}} - 1 < \pi^+(\tau^\star_{i_0+1}+1).
    \]
    Because $\pi^+(\tau^\star_{i_0+1}+1)$ is the smallest element in $I_{\mathrm{det}}$ greater than or equal to $\tau^\star_{i_0+1}+1$, any integer strictly less than it must be bounded above by $\tau^\star_{i_0+1}$. Therefore, $R^-_{i_0} \le \tau^\star_{i_0+1}$.
\end{itemize}
It follows that every $j \in S_{i_0}^{-,\mathrm{inf}}$ satisfies:
\[
\tau^\star_{i_0-1}+1 \le j \le \tau^\star_{i_0+1}.
\]

\medskip\noindent
\textbf{Left and right bounds for $S_{i_0}^{+,\mathrm{inf}}$:} \\
Similarly, we define $S_{i_0}^{+,\mathrm{inf}} = \llbracket L^+_{i_0}\, , \, R^+_{i_0} \rrbracket \cap I_{\mathrm{inf}}$, with:
\[
L^+_{i_0} := j_{\hat{\tau}^{\mathrm{det}}_{i_0}+1}, \quad \text{and} \quad R^+_{i_0} := j_{\hat{\tau}^{\mathrm{det}}_{i_0+1}} - \delta_{n}^{\mathrm{tot}}\cdot\mathbf{1}_{\{i_0 + 1 \neq \kappa^\star\}}.
\]
To show that $L^+_{i_0} \ge \tau^\star_{i_0-1} + 1$, we rely on assumption~\eqref{eq:delta_nHoldout}, \eqref{eq:delta_n'Holdout} and the definition of the projection $\pi^+$:
\[
L^+_{i_0} = j_{\hat{\tau}^{\mathrm{det}}_{i_0}+1} 
\ge j_{\tau^{\star,\mathrm{det}}_{i_0} - \delta_{n_{\mathrm{det}}} + 1} 
> j_{\tau^{\star,\mathrm{det}}_{i_0-1} + 1} 
= \pi^+(\tau^\star_{i_0-1}+1) 
\ge \tau^\star_{i_0-1} + 1.
\]
Finally, we prove that $R^+_{i_0} \le \tau^\star_{i_0+1}$. We again distinguish two cases:
\begin{itemize}
    \item \textbf{Case $i_0 = \kappa^\star - 1$:} The indicator function is zero, meaning $R^+_{\kappa^\star-1} = j_{\hat{\tau}^{\mathrm{det}}_{\kappa^\star}}$. By definition, this yields $R^+_{\kappa^\star-1} \le  \pi^-(n) = j_{\tau^{\star,\mathrm{det}}_{\kappa^\star}}  \le n = \tau^\star_{\kappa^\star}$, so the claim holds.
    \item \textbf{Case $i_0 \le \kappa^\star - 2$:} Using Lemma~\ref{lem:tau_det}, we have:
    \[
    \tau^\star_{i_0+1} - R^+_{i_0} \ge \tau^\star_{i_0+1} - \left( j_{\hat{\tau}^{\mathrm{det}}_{i_0+1}} - \delta_{n}^{\mathrm{tot}} \right) \ge 0.
    \]
\end{itemize}
We thus conclude that $R^+_{i_0} \le \tau^\star_{i_0+1}$ in all configurations. 

Therefore, combining these bounds with the earlier result for $S_{i_0}^{-,\mathrm{inf}}$, we obtain:
\[
\forall j \in S_{i_0}^{-,\mathrm{inf}} \cup S_{i_0}^{+,\mathrm{inf}}, \quad \tau^\star_{i_0-1}+1 \le j \le \tau^\star_{i_0+1},
\]
which completes the proof.
\end{proof}

\subsection{Proof of the level control (\ref{subsec:thresholdd})}
\label{app:threshold}
\begin{proof}[Proof of the level control]
Under $H_0^{\phi}$, the observations $(Z_i)_{i \in S_1 \cup S_2}$ are i.i.d., so
$p=q$ and $\mathrm{MMD}[p,q]=0$. With $0 \le k \le M^2$, Theorem~7 of
\cite{gretton2012} gives, for every $\varepsilon>0$,
\[
P_{H_0^{\phi}}\!\left\{
\mathrm{MMD}_b\big((Z_i)_{i \in S_1},(Z_i)_{i \in S_2}\big)
> 2\!\left( \Big(\tfrac{M^2}{|S_1|}\Big)^{1/2} + \Big(\tfrac{M^2}{|S_2|}\Big)^{1/2} \right) + \varepsilon
\right\}
\le 2\exp\!\left( \frac{-\varepsilon^2 |S_1||S_2|}{2M^2(|S_1|+|S_2|)} \right).
\]
Choosing $\varepsilon$ so that the right-hand side equals $\alpha_1$ gives
$\varepsilon = \big(2M^2(\tfrac{1}{|S_1|}+\tfrac{1}{|S_2|})\log\tfrac{2}{\alpha_1}\big)^{1/2}$.
Writing $s_\wedge := \min(|S_1|,|S_2|)$ and bounding both terms by the smaller
segment, via $(\tfrac{M^2}{|S_i|})^{1/2}\le(\tfrac{M^2}{s_\wedge})^{1/2}$ and
$\tfrac{1}{|S_1|}+\tfrac{1}{|S_2|}\le\tfrac{2}{s_\wedge}$, we define the threshold:
\[
\mathrm{threshold}_{\alpha_1} := \sqrt{\frac{M^2}{\min(|S_1|,|S_2|)}}\left(4+2\sqrt{\log\tfrac{2}{\alpha_1}}\right).
\]
Since this exceeds the exact threshold above,
$P_{H_0^{\phi}}\big\{\mathrm{MMD}_b\big((Z_i)_{i \in S_1},(Z_i)_{i \in S_2}\big)>\mathrm{threshold}_{\alpha_1}\big\}\le\alpha_1$.
\end{proof}
\section{Additional proofs}
\subsection{Complements to the GTST method}
\begin{lemma}\label{lem:partentiere}
Let $a, b \in \mathbb{R}$. Then we have:
\[
\lfloor a - b\rfloor \leq \lfloor a \rfloor - \lfloor b \rfloor \leq \lceil a - b \rceil.
\]
\end{lemma}

\begin{proof}[\textbf{Proof.}]
We distinguish two cases depending on whether $a - b$ is an integer or not.

\noindent\textbf{Case 1:} If $a - b \in \mathbb{Z}$, then:
\[
\lfloor a \rfloor - \lfloor b \rfloor 
= \lfloor b + (a - b) \rfloor - \lfloor b \rfloor = \lfloor b  \rfloor + (a - b) - \lfloor b \rfloor  = a - b = \lceil a - b \rceil,
\]
and therefore the inequality is trivial.

\noindent \textbf{Case 2:} Otherwise, we can write $a - b = n + \varepsilon$, with $n \in \mathbb{Z}$ and $\varepsilon \in (0, 1)$. Then:
\[
\lfloor a \rfloor - \lfloor b \rfloor 
= \lfloor b + n + \varepsilon \rfloor - \lfloor b \rfloor 
= n + \lfloor b + \varepsilon \rfloor - \lfloor b \rfloor,
\]
and:
\[
0 = \lfloor b + 0 \rfloor - \lfloor b \rfloor \leq \lfloor b + \varepsilon \rfloor - \lfloor b \rfloor \leq \lfloor b + 1 \rfloor - \lfloor b \rfloor \leq 1, \quad \text{since } \varepsilon \in (0, 1),
\]
hence:
\[
\lfloor a - b \rfloor = n \leq \lfloor a \rfloor - \lfloor b \rfloor \leq n + 1 = \lceil a - b \rceil.
\]
\end{proof}

\section{Empirical estimation of the constant in Theorem 3.1 of \cite{Garreau2018}}
\label{sec:C_0}
In this section, we estimate empirically the constant appearing in Theorem~3.1 of \cite{Garreau2018}.
This constant determines the minimal spacing condition required for the theoretical guarantees of the kernel change-point detection procedure discussed in Section~\ref{sec:background}: $\delta_n(\alpha_0)$.

\subsection{Definition of the parameters}

We introduce the parameters used throughout the simulations.

\begin{itemize}

\item \textbf{Observation space.}
We restrict ourselves here to observations taking values in the space
\[
\mathcal{X} = \mathbb{R}^d,
\]
where \(d\) denotes the dimension of the data.

\item \textbf{Collection of distributions.}
Let \(\mathcal{P}\) denote the set of probability distributions used to generate the data.
In our experiments, the observations \(X_1,\ldots,X_n \in \mathcal{X}\) are independent and are sampled from distributions belonging to \(\mathcal{P}\), which includes the following families:

\begin{itemize}
\item Gaussian distributions with varying means, variances, and covariance matrices;
\item exponential distributions with different rate parameters;
\item Laplace distributions with different scale parameters;
\item truncated versions of the above distributions;
\item uniform distributions on different intervals.
\end{itemize}

This collection of distributions allows us to evaluate the behavior of the method under a variety of distributional scenarios.

\item \textbf{Collection of kernels.}
Let \(\mathcal{K}\) denote the family of reproducing kernels used in the simulations.
The following kernels are considered:

\begin{itemize}
\item Gaussian kernel with bandwidth \(\sigma = 1\) or with \(\sigma\) chosen by the median heuristic:
\[
k(x,y) = \exp\!\left(-\frac{\|x-y\|^2}{2\sigma^2}\right);
\]

\item Laplace kernel with bandwidth \(\sigma = 1\):
\[
k(x,y) = \exp\!\left(-\frac{\|x-y\|}{\sigma}\right), \qquad \sigma = 1;
\]

\item polynomial kernel of degree \(2\) with constant \(1\):
\[
k(x,y) = (\langle x,y \rangle + 1)^2 ;
\]

\item linear kernel:
\[
k(x,y) = \langle x,y \rangle .
\]
\end{itemize}

\item \textbf{Minimal distributional separation.}
Let
\[
\underline{\Delta} = \min_{1 \le j \le \kappa_{\tau^\star} -1}
\left\| \mu_{P_j} - \mu_{P_{j+1}} \right\|_{\mathcal{H}},
\]
where \(P_j\) denotes the distribution of the observations on the \(j\)-th segment associated with the true segmentation \(\tau^\star\), and \(\mu_{P_j}\) denotes the kernel mean embedding of \(P_j\) in the reproducing kernel Hilbert space \(\mathcal{H}\) associated with a kernel \(k\).

\item \textbf{Penalty term.}
For a segmentation \(\tau\), the penalty used in the kernel change-point detection procedure is defined as
\[\mathrm{pen}(\tau) := \frac{C\,M^2 \kappa_\tau}{n},\]
where \(\kappa_\tau\) denotes the number of segments in the segmentation \(\tau\), \(n\) is the sample size, \(M\) bounds the kernel, and \(C\) is a positive constant.

\end{itemize}

\noindent
\textbf{Notation.}
To avoid ambiguity, we distinguish between kernels defined on the original observation space
\(\mathcal X = \mathbb{R}^d\) and kernels defined on the projected space \(\mathcal Z\).
We therefore denote by \(k_{\mathcal X}\) the kernels acting on \(\mathcal X\), and by \(k_{\mathcal Z}\) those used after projection onto a block of coordinates.
Both types of kernels are assumed to satisfy the boundedness condition below whenever required by the theoretical analysis.

\paragraph{Assumption.}
We assume that the kernel used in the change-point detection procedure is bounded.
More precisely, there exists a strictly positive constant \(M\) such that
\begin{equation}
\forall i \in \{1,\ldots,n\}, \qquad
k_{\mathcal X}(X_i,X_i) \le M^2 < +\infty \quad \text{a.s.}
\label{eq:kernelbounded}
\end{equation}
When the kernel is not bounded (for instance the linear or polynomial kernels),
we restrict the simulations to truncated versions of the distributions in order to ensure that Assumption~\eqref{eq:kernelbounded} holds.
In practice, the value of \(M\) depends on both the kernel and the data distribution.

\medskip

\begin{theorem}[Garreau and Arlot, 2018, Theorem 3.1]
\label{thm:thm_3_1}
Assume that Assumption~\eqref{eq:kernelbounded} holds. For any $y>0$, there exists an event $\Omega$ with probability at least $1-e^{-y}$ on which the following result holds.

For any $C>0$, let $\hat{\tau}$ be the segmentation estimated by KCP with the penalty defined above. Define
\[
C_{\min} := \frac{74}{3}(\kappa_{\tau^\star}+1)(y+\log n+1),
\qquad
C_{\max} := \frac{\underline{\Delta}^2}{M^2}\frac{\underline{\Lambda}_{\tau^\star}}{6\kappa_{\tau^\star}}\,n .
\]
If
\[
C_{\min} < C < C_{\max},
\]
then, on $\Omega$,
\[
\kappa_{\hat{\tau}} = \kappa_{\tau^\star},
\]
and
\[
\frac{1}{n} d_{\infty}^{(1)}(\tau^\star,\hat{\tau})
\le v_1(y)
:=
\frac{148\,\kappa_{\tau^\star} M^2}{\underline{\Delta}^2}\,
\frac{y+\log n+1}{n}.
\]
\end{theorem}

\subsection{Replacing the constant 148 when $\kappa_{\tau^\star}$ is known}

\subsubsection{Theoretical setting}
Assume that $\hat{\tau}$ is estimated using KCP with the number of change-points fixed to $\kappa_{\tau^\star}-1$. Theorem~\ref{thm:thm_3_1} is applicable if there exists $C> 0$ such that $C_{\min} < C < C_{\max}$, i.e., if $C_{\min} < C_{\max}$; in that case, we obtain
\[
\frac{1}{n} d_{\infty}^{(1)}(\tau^\star,\hat{\tau})
\le
148 \frac{\kappa_{\tau^\star} M^2}{\underline{\Delta}^2}
\frac{y+\log(n)+1}{n}.
\]
Equivalently,
\[
148 \ge
\frac{\frac{1}{n} d_{\infty}^{(1)}(\tau^\star,\hat{\tau})}
{\frac{\kappa_{\tau^\star} M^2}{\underline{\Delta}^2}\frac{y+\log(n)+1}{n}} .
\]
Moreover, the condition $C_{\max}>C_{\min}$ yields
\[
C_{\max}>C_{\min}
\quad \Longleftrightarrow \quad
\frac{\underline{\Delta}^2}{M^2}\frac{\underline{\Lambda}_{\tau^\star}}{\kappa_{\tau^\star}}
\frac{n}{(\kappa_{\tau^\star}+1)(y+\log(n)+1)}
>148 .
\]
By replacing the constant $148$ appearing in the two inequalities above by the constant $C_{0,\sup}$ defined below, the theorem remains valid when $\kappa_{\tau^\star}$ is known. In the following definition, a \emph{configuration} refers to an admissible choice of the sample size $n$, the true segmentation $\tau^\star$ (hence $\kappa_{\tau^\star}$ and $\underline{\Lambda}_{\tau^\star}$), and the distributions $P_1,\ldots,P_{\kappa_{\tau^\star}} \in \mathcal{P}$ of the segments (hence $\underline{\Delta}$ and $M$).
 
\begin{align*}
C_{0,\sup} :=\
& \inf \left\{ C_0 > 0\ \middle|\
\begin{array}{l}
\forall y > 0,\ \forall k_{\mathcal{X}} \in \mathcal{K},\ \text{and all configurations}, \\
\displaystyle
\frac{\frac{\underline{\Delta}^2}{M^2}\frac{\underline{\Lambda}_{\tau^\star}}{\kappa_{\tau^\star}}}
{(\kappa_{\tau^\star} + 1)\frac{(y + \log(n) + 1)}{n}}
> C_0
\, \Rightarrow \\
Q_{1 - e^{-y}} \left(
\frac{\frac{1}{n} d_{\infty}^{(1)}(\tau^\star, \hat{\tau})}
{\displaystyle \frac{\kappa_{\tau^\star} M^2}{\underline{\Delta}^2} \cdot \frac{y + \log(n) + 1}{n}}
\right) \leq C_0
\end{array}
\right\} \\[1em]
=\
& \sup \left\{ \min \left(
Q_{1 - e^{-y}} \left(
\frac{\frac{1}{n} d_{\infty}^{(1)}(\tau^\star, \hat{\tau})}
{\displaystyle \frac{\kappa_{\tau^\star} M^2}{\underline{\Delta}^2} \cdot \frac{y + \log(n) + 1}{n}}
\right), \frac{\frac{\underline{\Delta}^2}{M^2}\frac{\underline{\Lambda}_{\tau^\star}}{\kappa_{\tau^\star}}}
{(\kappa_{\tau^\star} + 1)\frac{(y + \log(n) + 1)}{n}} \right)  \right. \\
&
\left.
\middle|\
\begin{array}{l}
y > 0,\ k_{\mathcal{X}} \in \mathcal{K},\ \text{and all configurations}
\end{array}
\right\},
\end{align*}
where \(Q_{\alpha}(X)\) denotes the quantile of order \(\alpha\) of the random variable \(X\), and where we used Lemma~\ref{lem:C_0_sup} in Section~\ref{sec:complements}.

\subsubsection{Simulations}
From Theorem~\ref{thm:thm_3_1}, we know that $C_{0,\sup} \leq 148$. The goal of this section is to estimate empirically the constant $C_{0,\sup}$ over the class of scenarios explored in our simulations. Exhaustively evaluating all possible configurations is computationally infeasible. Instead, we explore a large collection of representative scenarios in order to obtain an estimate close to the true value of $C_{0,\sup}$.

We proceed as follows. First, we fix a data-generating distribution $P \in \mathcal{P}$ on $\mathbb{R}^d$ and a kernel $k \in \mathcal{K}$ defined on $\mathcal{X}$. In practice, the distribution $P$ is chosen randomly among several families of distributions (e.g., Gaussian, Laplace, exponential, uniform, possibly truncated when required), and the kernel is selected among a predefined collection of kernels. We also fix the confidence parameter $y$, chosen such that the considered quantile is of order $0.975$.

The remaining parameters of the model include the dimension $d$, the number of segments $\kappa_{\tau^\star}$, the minimal spacing $\underline{\Lambda}_{\tau^\star}$, and the minimal jump size $\underline{\Delta}$. In the simulations, we vary these parameters across different configurations.

For each configuration, we estimate the corresponding value of $C_0$; the constant $C_{0,\sup}$ is then estimated by aggregating these values across configurations.

\subsubsection{Description of the experimental procedure}
The objective is to empirically estimate the constant $C_{0,\sup}$ defined in the theoretical guarantees of change-point detection. The experimental procedure is as follows.

\begin{enumerate}
\item \textbf{Data generation.}
We generate multiple random time series under different configurations. The configurations vary in terms of distribution families, marginal distributions, and distribution parameters (such as mean, variance, or covariance). Each time series contains a prescribed number of change-points.

\item \textbf{Change-point detection.}
For each simulated time series, we apply the Kernel Change-Point (KCP) detection algorithm using the chosen kernel $k$.

\item \textbf{Evaluation of the localization error.}
For each simulation, we compute the localization error using the distance
\[
d_\infty^{(1)}(\tau^\star,\hat{\tau}),
\]
which corresponds to the maximum alignment error between the true and estimated change-points.

\item \textbf{Construction of the curve.}
For each simulated configuration, we evaluate the empirical quantity
\[
\min \left(
Q_{1-e^{-y}}
\left(
\frac{\frac{1}{n}d_\infty^{(1)}(\tau^\star,\hat{\tau})}
{\frac{\kappa_{\tau^\star} M^2}{\underline{\Delta}^2}\cdot \frac{y+\log(n)+1}{n}}
\right),
\;
\frac{\frac{\underline{\Delta}^2}{M^2}\frac{\underline{\Lambda}_{\tau^\star}}{\kappa_{\tau^\star}}}
{(\kappa_{\tau^\star}+1)\frac{(y+\log(n)+1)}{n}}
\right).
\]
This defines the curve of $C_0$, evaluated as a function of the minimal discrepancy between distributions, measured through the minimum MMD between consecutive segments.

\item \textbf{Estimation of $C_{0,\sup}$.}
For each configuration, we retain the largest value attained by this curve. This provides an empirical estimate of $C_0$ for the considered configuration (where $n$, $\kappa_{\tau^\star}$, the change-point locations, and the kernel $k$ are fixed, and the distributional change varies through $\underline{\Delta}$), in accordance with the sup--min formulation of $C_{0,\sup}$ given above.

\item \textbf{Aggregation across configurations.}
Finally, the procedure is repeated across multiple configurations by varying parameters such as the sample size $n$, the number of change-points $\kappa_{\tau^\star}-1$, their locations, the kernel $k$, and the data-generating distributions. The largest empirical value of $C_0$ obtained across all configurations is retained. This approach is consistent with the theoretical definition of $C_{0,\sup}$, which must satisfy the inequality for all admissible scenarios.

\end{enumerate}

\subsubsection{Curve of $C_0$ as a function of the data dimension}
The quantity plotted as a function of the dimension is defined as follows:
\begin{align*}
C_0(\text{fixed parameters}, d)
:=
\sup
\left\{
\min \left(
Q_{1-e^{-y}}
\!\left(
\frac{\frac{1}{n}d_\infty^{(1)}(\tau^\star,\hat{\tau})}
{\frac{\kappa_{\tau^\star} M^2}{\underline{\Delta}^2}\cdot\frac{y+\log(n)+1}{n}}
\right),
\frac{\frac{\underline{\Delta}^2}{M^2}\frac{\underline{\Lambda}_{\tau^\star}}{\kappa_{\tau^\star}}}
{(\kappa_{\tau^\star}+1)\frac{(y+\log(n)+1)}{n}}
\right)
\right. \\
\left.
\middle|\
\begin{array}{l}
\text{for any kernel } k\in\mathcal K \text{ and any time series of dimension } d \\
\text{with a change-point in the middle whose distributions belong to }\mathcal P
\end{array}
\right\}.
\end{align*}

We restrict the computation of this quantity to time series of length
\[
n \in \{100, 500, 1000\},
\]
with a single change-point located in the middle and whose distributions belong to $\mathcal P$. The fixed parameters are $\kappa_{\tau^\star}$, $\underline{\Lambda}_{\tau^\star}$, and $y$, whereas the non-fixed parameters are $n$, $k$, $M$, $\underline{\Delta}$, and $d$.

\begin{figure}[tbp]
\centering
\includegraphics[width=0.45\textwidth]{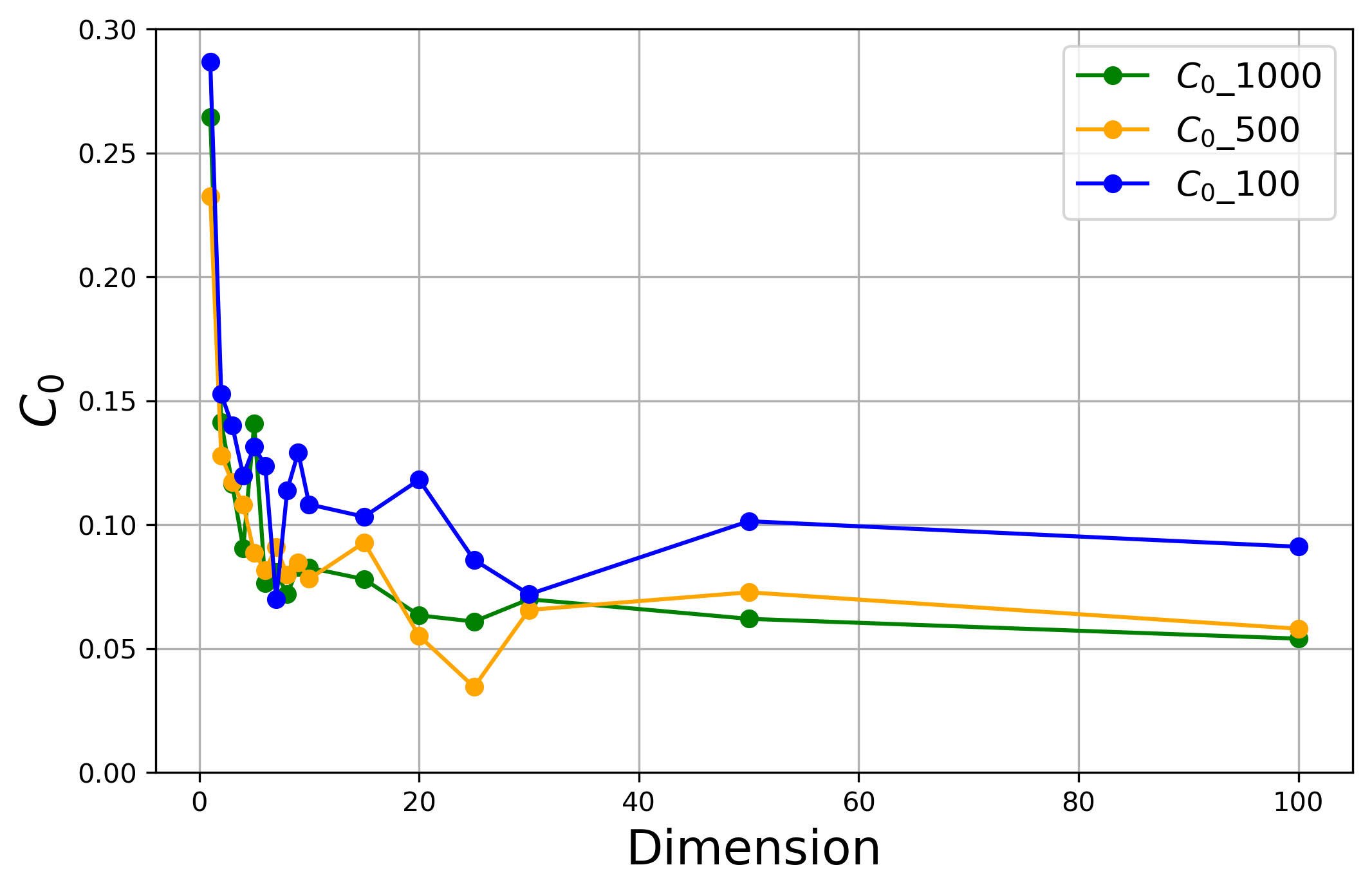}
\caption{$C_0$ as a function of the dimension of the time series.}
\label{fig:C0_dim}
\end{figure}

We observe in Figure~\ref{fig:C0_dim} that $C_0$ decreases as the dimension of the time series increases: the largest values, close to $0.3$, are attained in dimension $1$, and the values then drop quickly with the dimension. In particular, from dimension $2$ onward --- which corresponds to the multivariate setting considered in this work --- the maximal value of $C_0$ remains below $0.15$. For the remainder of the study, we therefore fix the dimension to $2$. We now investigate the influence of the parameter $\underline{\Lambda}_{\tau^\star}$ on $C_0$ in the case of a single change-point.

\subsubsection{Curve of $C_0$ as a function of the change-point position}
We fix the dimension of the time series to $2$. The quantity plotted as a function of the change-point position is defined as follows:
\begin{align*}
C_0(\text{fixed parameters}, \underline{\Lambda}_{\tau^\star})
:=
\sup
\left\{
\min \left(
Q_{1-e^{-y}}
\!\left(
\frac{\frac{1}{n}d_\infty^{(1)}(\tau^\star,\hat{\tau})}
{\frac{\kappa_{\tau^\star} M^2}{\underline{\Delta}^2}\cdot\frac{y+\log(n)+1}{n}}
\right),
\frac{\frac{\underline{\Delta}^2}{M^2}\frac{\underline{\Lambda}_{\tau^\star}}{\kappa_{\tau^\star}}}
{(\kappa_{\tau^\star}+1)\frac{(y+\log(n)+1)}{n}}
\right)
\right. \\
\left.
\middle|\
\begin{array}{l}
\text{for any kernel } k \in \mathcal{K} \text{ and any time series with a change-point} \\
\text{at different positions whose distributions belong to } \mathcal{P}
\end{array}
\right\}.
\end{align*}

We restrict the computation of these values to time series of lengths
\[
n \in \{100, 500, 1000\},
\]
with a single change-point and distributions belonging to $\mathcal{P}$. The fixed parameters are $d$, $\kappa_{\tau^\star}$, and $y$, while the non-fixed parameters are $n$, $k$, $M$, $\underline{\Delta}$, and $\underline{\Lambda}_{\tau^\star}$.

\begin{figure}[tbp]
\centering
\includegraphics[width=0.45\textwidth]{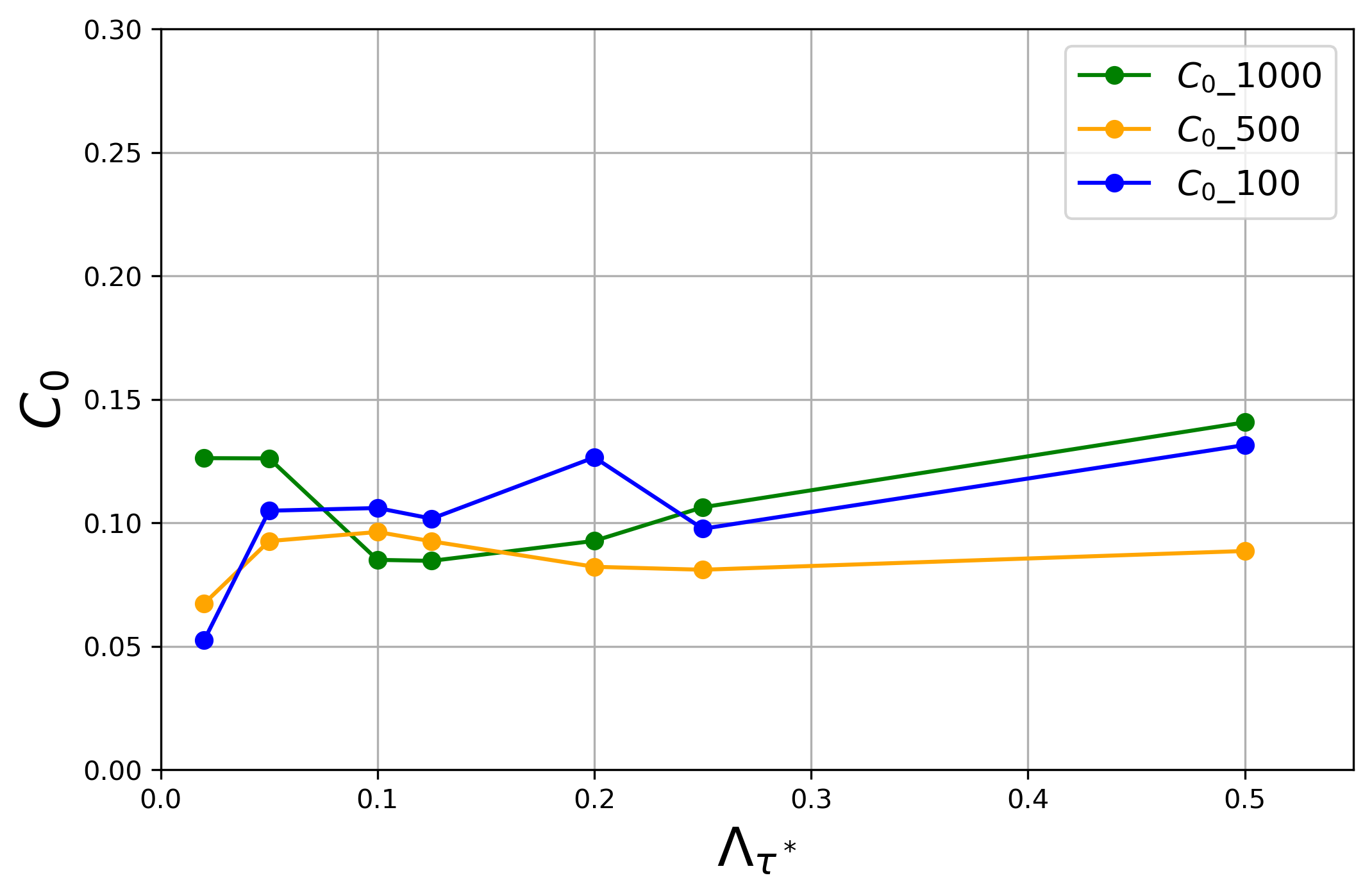}
\caption{$C_0$ as a function of $\underline{\Lambda}_{\tau^\star}$.}
\label{fig:C0_pos}
\end{figure}

We observe in Figure~\ref{fig:C0_pos} that the position of the change-point has little influence on $C_0$: whether the change-point lies near the boundary of the time series (small values of $\underline{\Lambda}_{\tau^\star}$) or near its center, the values of $C_0$ remain of the same order and stay below $0.15$, consistently with the previous experiment in dimension $2$.

\subsubsection{Curve of $C_0$ as a function of the number and locations of change-points}

We fix the dimension of the time series to $2$. For each number of segments $\kappa_{\tau^\star} \in \{3, 4\}$, we plot, as a function of $\underline{\Lambda}_{\tau^\star}$, the quantity defined as follows:

\begin{align*}
C_0(\text{fixed parameters}, \underline{\Lambda}_{\tau^\star})
:=
\sup
\left\{
\min \left(
Q_{1-e^{-y}}
\!\left(
\frac{\frac{1}{n}d_\infty^{(1)}(\tau^\star,\hat{\tau})}
{\frac{\kappa_{\tau^\star} M^2}{\underline{\Delta}^2}\cdot\frac{y+\log(n)+1}{n}}
\right),
\frac{\frac{\underline{\Delta}^2}{M^2}\frac{\underline{\Lambda}_{\tau^\star}}{\kappa_{\tau^\star}}}
{(\kappa_{\tau^\star}+1)\frac{(y+\log(n)+1)}{n}}
\right)
\right. \\
\left.
\middle|\
\begin{array}{l}
\text{for any kernel } k \in \mathcal{K} \text{ and any time series with } \kappa_{\tau^\star}-1 \text{ change-points} \\
\text{at different locations whose distributions belong to } \mathcal{P}
\end{array}
\right\}.
\end{align*}

We restrict the computation of this quantity to time series of lengths
\[
n \in \{100, 500\},
\]
with different numbers and locations of change-points, and whose distributions belong to $\mathcal{P}$. The fixed parameters are $d$, $\kappa_{\tau^\star}$, and $y$, while the non-fixed parameters are $n$, $k$, $M$, $\underline{\Delta}$, and $\underline{\Lambda}_{\tau^\star}$. The results are shown in Figures~\ref{fig:C0_D3} and~\ref{fig:C0_D4}.

\begin{figure}[tbp]
\centering
\includegraphics[width=0.45\textwidth]{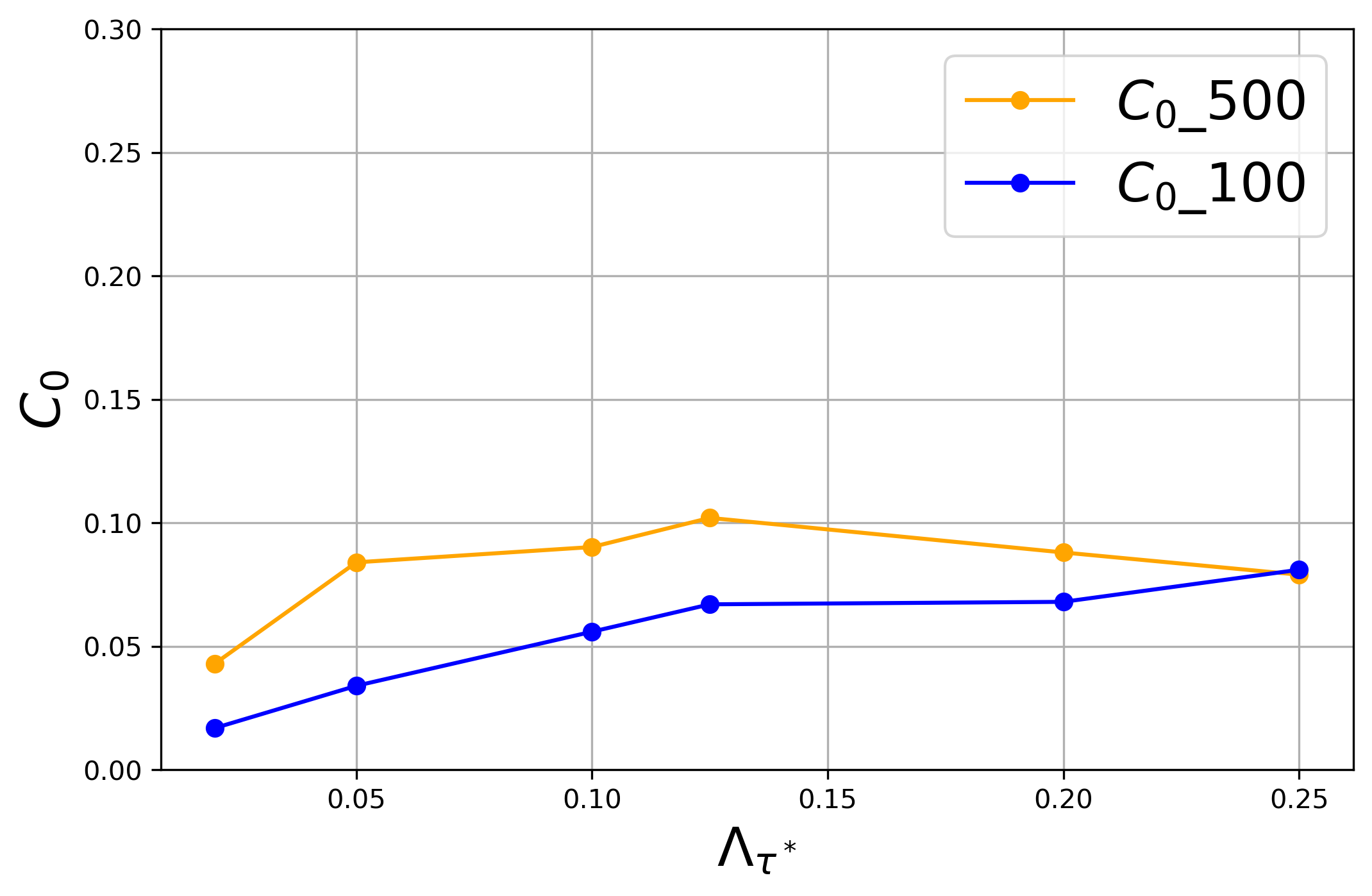}
\caption{$C_0$ as a function of $\underline{\Lambda}_{\tau^\star}$, with $\kappa_{\tau^\star} = 3$.}
\label{fig:C0_D3}
\end{figure}

\begin{figure}[tbp]
\centering
\includegraphics[width=0.45\textwidth]{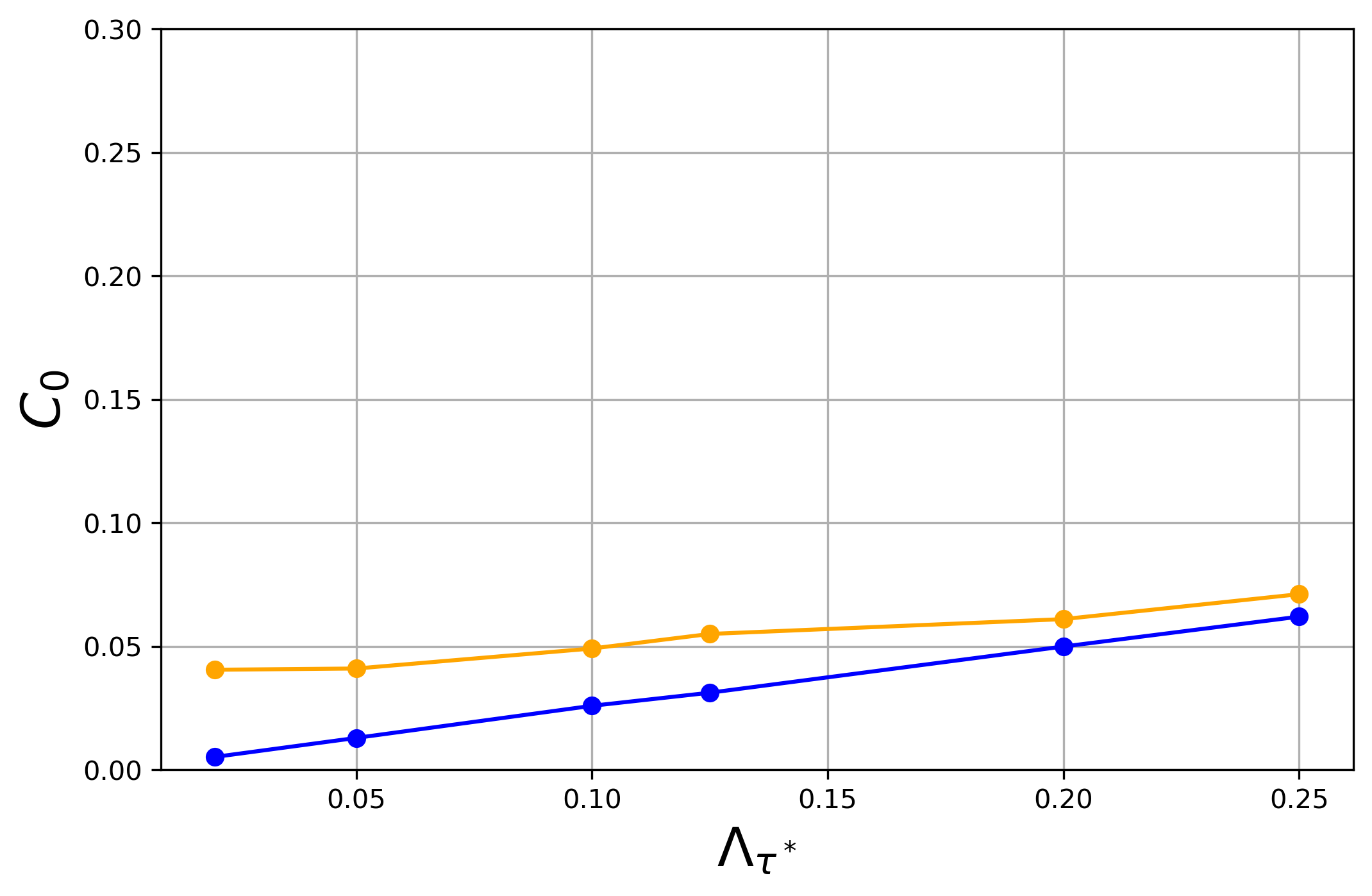}
\caption{$C_0$ as a function of $\underline{\Lambda}_{\tau^\star}$, with $\kappa_{\tau^\star} = 4$.}
\label{fig:C0_D4}
\end{figure}

We observe in Figures~\ref{fig:C0_D3} and~\ref{fig:C0_D4} that increasing the number of change-points does not lead to larger values of $C_0$; on the contrary, the values obtained for $\kappa_{\tau^\star} = 3$ and $\kappa_{\tau^\star} = 4$ are even smaller than those obtained in the single change-point case, remaining well below $0.15$ (below $0.10$ in these experiments).

\subsection{Comments and conclusions}
In this study, we explored several configurations in order to obtain an empirical value of $C_0$ that is close to its theoretical value.

The largest values of $C_0$ obtained in the simulations, around $0.3$, are only attained in dimension $1$. In the multivariate case ($d \geq 2$), which is the setting considered in this work, all the empirical values of $C_0$ remain below $0.15$. This suggests that the value of $C_{0,\sup}$ is much smaller than the constant $148$ appearing in the bound of Theorem~\ref{thm:thm_3_1}.

Based on these observations, we replace the highly conservative constant $148$ by the empirical tuning parameter $C_0 = 0.15$ in most of our simulations, unless specified otherwise. This value can be adapted to the scenario under consideration: for instance, when the time series contains more than one change-point, the empirical values of $C_0$ are even smaller, and a smaller value of $C_0$ can therefore be used.

Even when replacing $148$ by the empirical value of $C_0$ obtained from the simulations, the direct application of Theorem~\ref{thm:thm_3_1} remains difficult. Indeed, several parameters involved in the definition of $C_0$ are typically unknown in practice, namely $\kappa_{\tau^\star}$, $\underline{\Delta}$, and $\underline{\Lambda}_{\tau^\star}$.

\begin{remark}
If one wishes to use a different kernel or to consider other types of distributional changes that were not studied here, the entire experimental procedure can be repeated with the new kernel or the new types of distributional changes. This makes it possible to obtain a new empirical value of $C_0$ adapted to these choices by following the same methodology.
\end{remark}
\subsection{Complements for $C_0$}
\label{sec:complements}
\begin{lemma}\label{lem:C_0_sup}
Let $\mathcal{C} \subset \mathbb{R}_+^2$ be a set describing all admissible pairs $(R,Q)$ arising from the variation of the parameters of the model. Define
\[
C_{0,\mathrm{sup}} := \inf \left\{ C_0 > 0 \;\middle|\; 
\forall (R,Q) \in \mathcal{C}, \; R > C_0 \Rightarrow Q \le C_0
\right\}.
\]
We have
\[
C_{0,\mathrm{sup}} = \sup \left\{ \min(R,Q) \;\middle|\; (R,Q) \in \mathcal{C} \right\}.
\]
\end{lemma}

\begin{proof}
\[
C_{0,\mathrm{sup}} = \inf \left\{ C_0 > 0 \;\middle|\;
\forall (R,Q) \in \mathcal{C},\; R > C_0 \Rightarrow Q \le C_0
\right\}.
\]
We observe that for any $(R,Q)$ and any $C_0$,
\[
\left(R > C_0 \Rightarrow Q \le C_0\right)
\quad \Leftrightarrow \quad
\min(R,Q) \le C_0 .
\]
Indeed, if $R \le C_0$, then $\min(R,Q) \le C_0$. If $R > C_0$, the implication $R > C_0 \Rightarrow Q \le C_0$ gives $Q \le C_0$, and therefore $\min(R,Q) \le C_0$. 
Conversely, if $\min(R,Q) \le C_0$ and $R > C_0$, then necessarily $Q \le C_0$, so the implication $R > C_0 \Rightarrow Q \le C_0$ holds. Therefore,
\[
C_{0,\mathrm{sup}}
=
\inf \left\{ C_0 > 0 \;\middle|\;
\forall (R,Q) \in \mathcal{C},\; \min(R,Q) \le C_0
\right\}.
\]
Since this set consists of all upper bounds of the collection
\[
\{\min(R,Q) \mid (R,Q) \in \mathcal{C}\},
\]
its infimum is precisely the least upper bound of this collection. Hence
\[
C_{0,\mathrm{sup}} =
\sup \left\{ \min(R,Q) \;\middle|\; (R,Q) \in \mathcal{C} \right\}.
\]
\end{proof}

\section{Complementary simulations}
\label{sec:complementary_simulations}
\subsection{Signal}
\subsubsection{Influence of signal sparsity}
This experiment investigates the impact of signal sparsity on the performance of the \emph{post hoc} procedures. To properly assess how the dilution of the signal affects the test statistics, we fix the mean jump amplitude at $0.3$ and vary the number of affected coordinates. The total dimension is set to $d=50$, and a single change-point is placed at the center of the sequence.

The coordinates are partitioned into two distinct blocks:
$$ \text{Block 1} = \{1,\ldots,30\}, \qquad \text{Block 2} = \{31,\ldots,50\}. $$
The distributional change is restricted to Block~1, while Block~2 serves as a completely stationary control group. Specifically, the mean shift is applied to a subset of $s$ coordinates within Block~1, where $s$ varies across the experiments. The remaining $30-s$ coordinates in Block~1, along with all coordinates in Block~2, remain unchanged. For the localized tests (Hold-out and GTST), the uncertainty margin $\delta_n$ is calibrated using the constant $C_0 = 0.08$.

By evaluating the \emph{post hoc} procedures as a function of the number of affected coordinates $s$, we can precisely assess how transitioning from a highly sparse signal ($s$ close to $1$) to a denser signal ($s$ close to $30$) influences the empirical power, and the Type~I error control.

\begin{figure}[tbp]
\centering

\begin{subfigure}{0.48\textwidth}
    \centering
    \includegraphics[width=\textwidth]{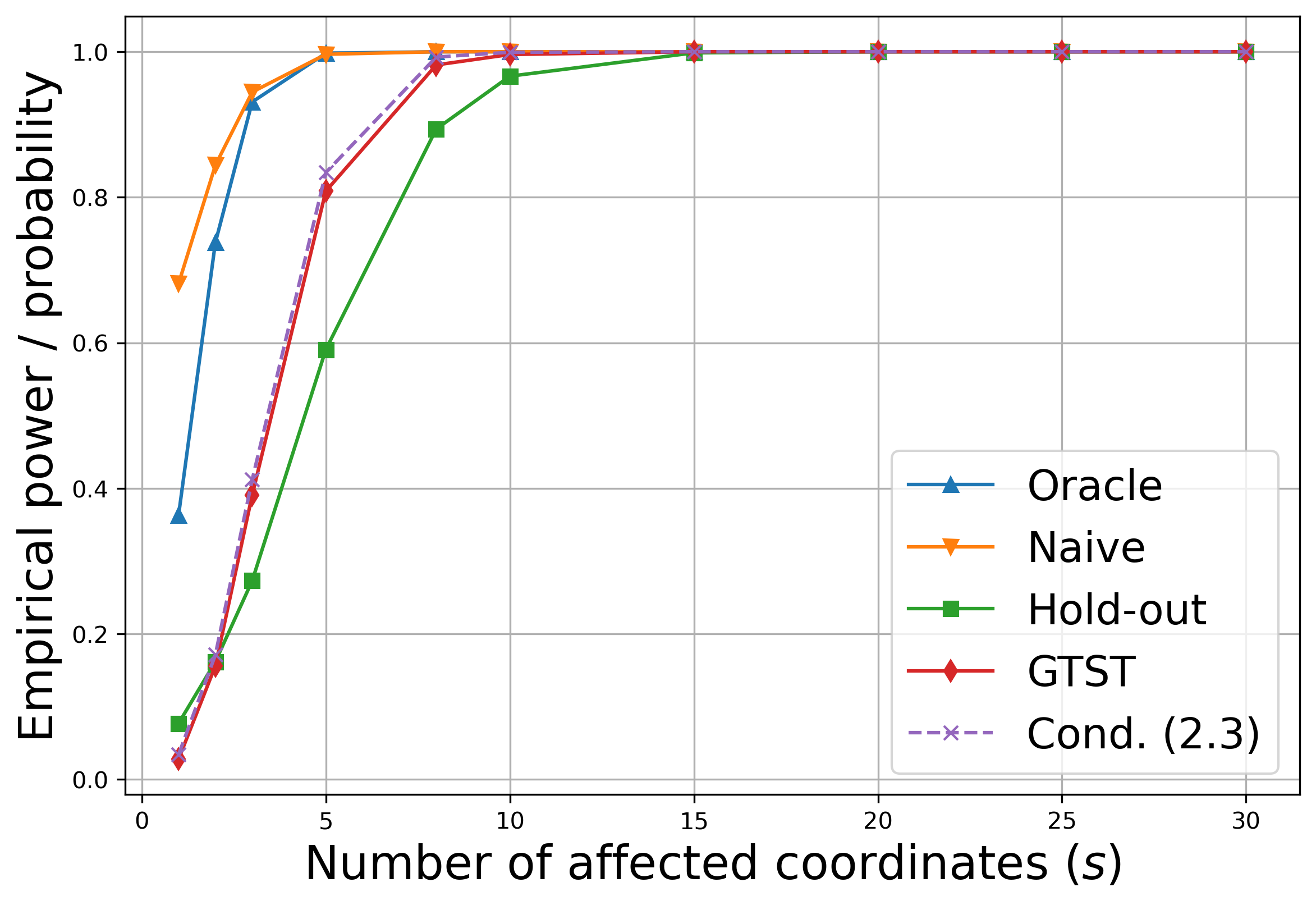}
    \caption{Power}
    \label{fig:sparcité_power}
\end{subfigure}
\hfill
\begin{subfigure}{0.48\textwidth}
    \centering
    \includegraphics[width=\textwidth]{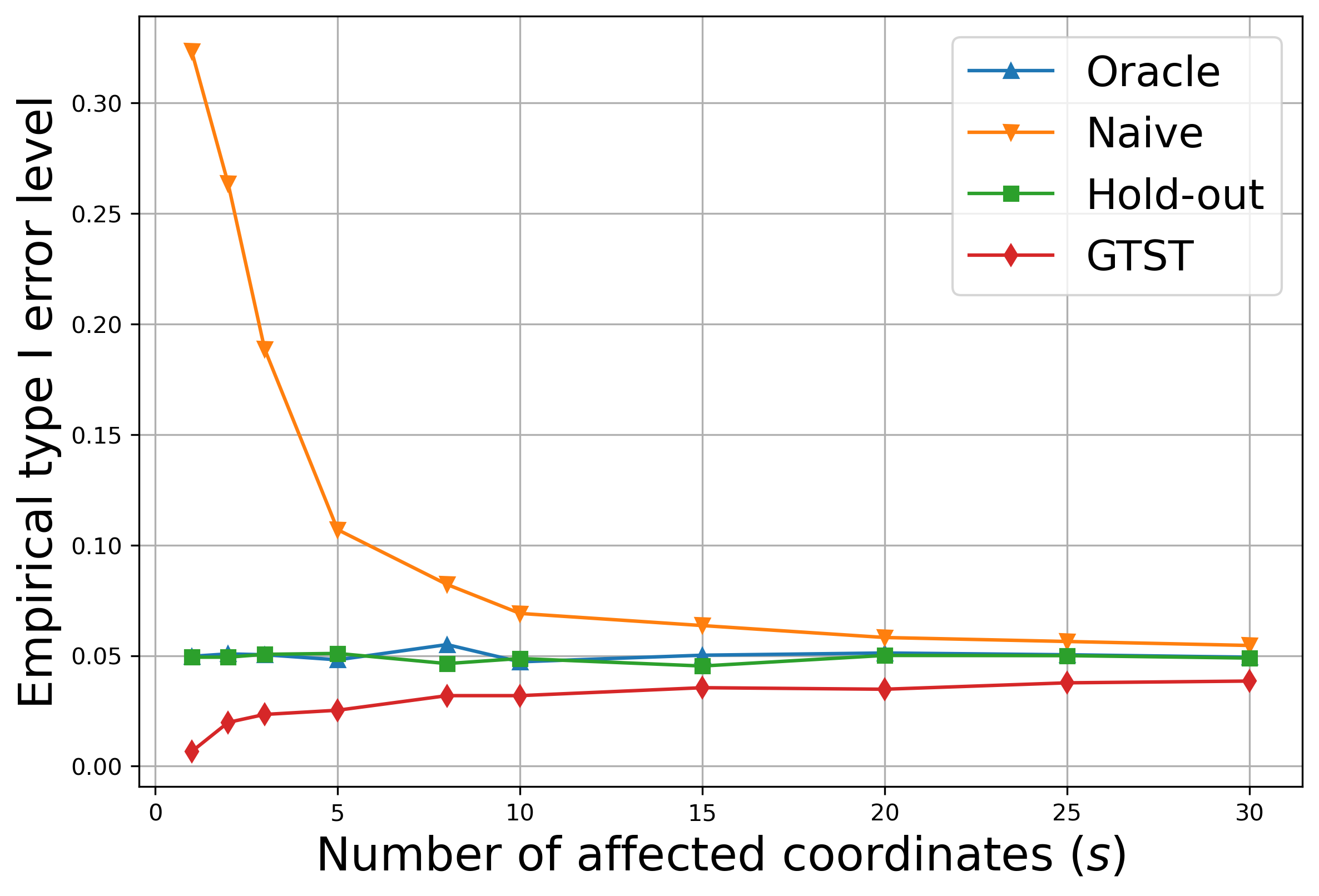}
    \caption{Type~I error level}
    \label{fig:sparcité_level}
\end{subfigure}

\caption{
Performance of the \emph{post hoc} procedures as a function of the signal sparsity. The total dimension is fixed at $d=50$ and the mean jump amplitude is set to $0.3$. The x-axis represents the number of affected coordinates $s$ within the active block. As $s$ increases, the signal transitions from a highly sparse regime to a denser regime. (a) Empirical power; (b) Empirical Type~I error level.
}
\label{fig:sparcité_results}
\end{figure}

\paragraph{Results.}
Figure~\ref{fig:sparcité_results} confirms that highly sparse signals (e.g., small $s$) are heavily diluted by the ambient noise of the stationary dimensions, leading to low empirical power across all methods. As the number of affected coordinates increases, the aggregated signal overcomes this noise, restoring optimal detection performance. 

Consistent with previous observations, the naive method fails to control the Type~I error. Notably, this vulnerability is considerably amplified in the extreme sparse regime, with the error rate exceeding $0.3$ for $s=1$. By contrast, the hold-out and GTST procedures safely maintain their nominal error control even when the change is virtually undetectable. This confirms that the proposed strategies inherently protect against false discoveries, regardless of the signal density.
\subsubsection{Influence of the dimension}
In this experiment, we investigate how the performance of the \emph{post hoc} procedures evolves with the dimension~$d$. A single change-point is considered, and the mean jump amplitude is fixed at $\text{jump}=0.3$. All simulation settings remain unchanged except for the dimension $d$, which varies over a predefined range while the proportion of affected coordinates is kept approximately constant. More precisely, for each dimension $d$, a mean shift is introduced on the first $\lfloor 0.6\,d \rfloor$ coordinates, while the remaining coordinates are left unchanged. Accordingly, Block~1 consists of the first $\lfloor 0.6\,d \rfloor$ coordinates and undergoes the distributional change, whereas Block~2 contains the remaining coordinates. This setting allows us to examine how increasing dimensionality influences both detection performance and exact block recovery while maintaining a fixed proportion of affected coordinates. For the localized tests (Hold-out and GTST), the uncertainty margin $\delta_n$ is calibrated using a dimension-dependent constant $C_0$, defined as follows:
$$
C_0 = 
\begin{cases} 
0.15 & \text{if } d \le 5, \\ 
0.10 & \text{if } 5 < d \le 15, \\ 
0.08 & \text{if } d \ge 15. 
\end{cases}
$$

\begin{figure}[tbp]
\centering

\begin{subfigure}{0.48\textwidth}
    \centering
    \includegraphics[width=\textwidth]{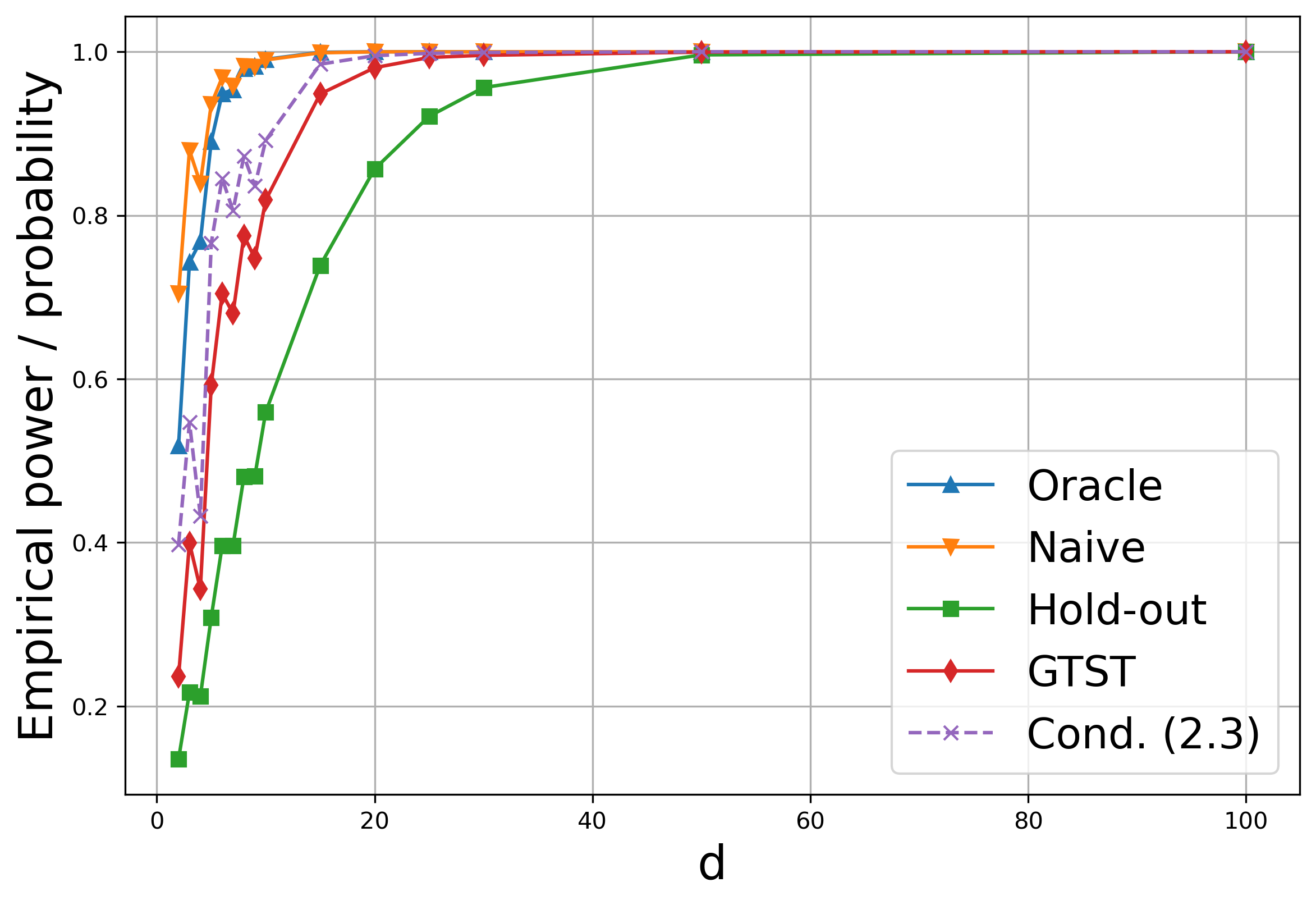}
    \caption{Power}
    \label{fig:d_power}
\end{subfigure}
\hfill
\begin{subfigure}{0.48\textwidth}
    \centering
    \includegraphics[width=\textwidth]{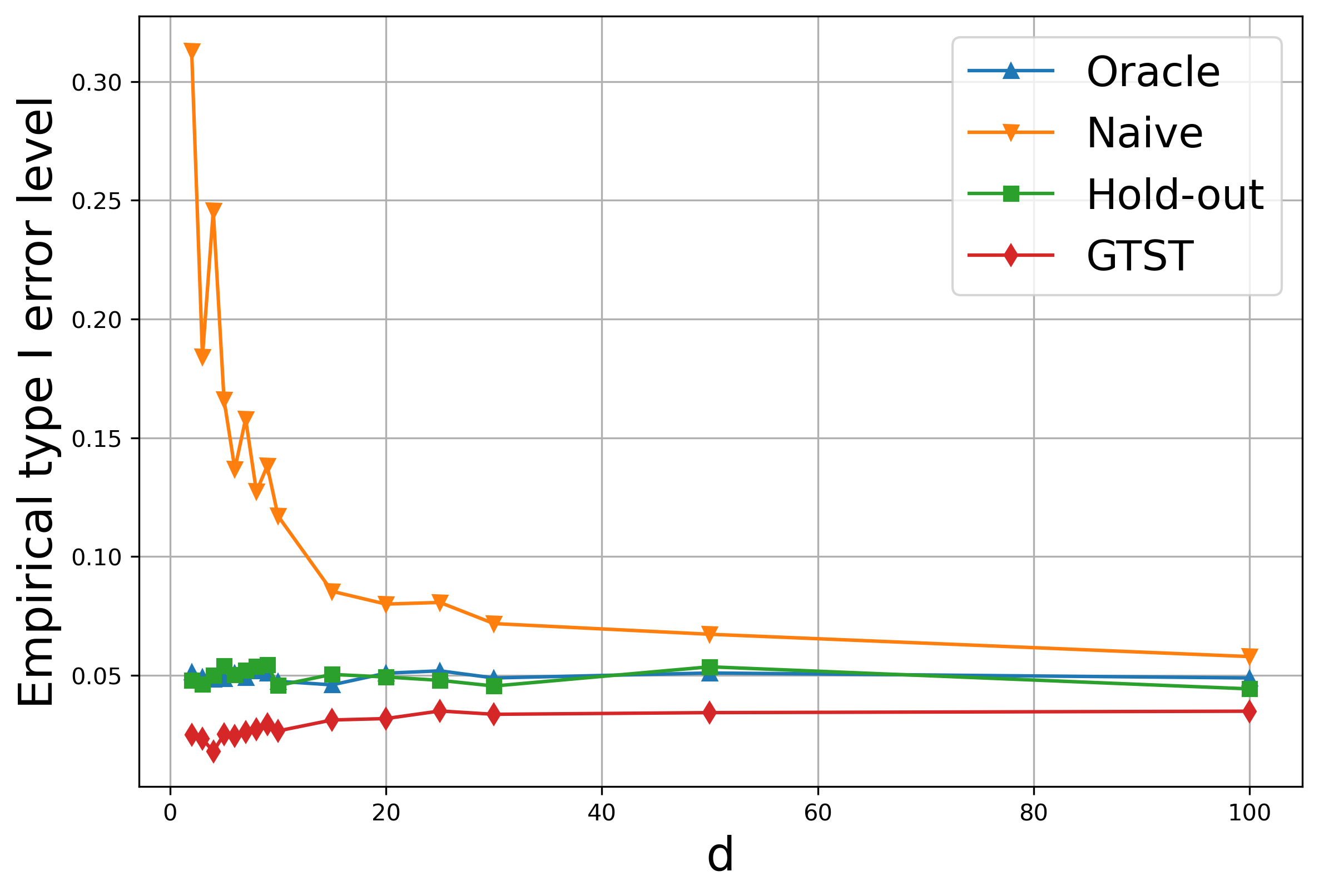}
    \caption{Type I error level}
    \label{fig:d_level}
\end{subfigure}

\vspace{0.5cm}

\caption{
Performance of the \emph{post hoc} procedures as a function of the dimension $d$. The proportion of affected coordinates is maintained at approximately $60\%$ ($\lfloor 0.6\,d \rfloor$) across all tested dimensions, corresponding to a dense signal regime. (a) Empirical power; (b) Empirical Type~I error level.
}
\label{fig:d_results}
\end{figure}

\paragraph{Results.}
Figure~\ref{fig:d_results} shows that, when using a Gaussian kernel with the bandwidth selected via the median heuristic, increasing the dimension does not deteriorate the performance of the \emph{post hoc} procedures. On the contrary, empirical power improves as the dimension grows.

We also observe small non-monotonic fluctuations in the performance curves as the dimension increases. These local variations are a direct consequence of the step-function behavior of $\lfloor 0.6\,d \rfloor$. When $d$ increases but the integer value $\lfloor 0.6\,d \rfloor$ remains temporarily unchanged, the newly added coordinate does not undergo any distributional change and acts as pure noise. This marginally degrades the signal-to-noise ratio, temporarily decreasing the power until the dimension is large enough to increment the number of affected coordinates, at which point the performance resumes its upward trend.

Overall, these results indicate that, when the proportion of changing coordinates remains fixed, the considered methods remain effective in moderate to high-dimensional settings under the median heuristic.

\subsubsection{Different types of distributional changes}
Beyond mean shifts, we consider alternative types of distributional changes in order to evaluate the robustness of the \emph{post hoc} procedures across different change mechanisms. In these experiments, the dimension is fixed to $d=5$, and only the first three coordinates are affected by the change, while the remaining coordinates remain unchanged.

For variance changes, the pre-change distribution has unit variance in all dimensions. After the change-point, the variance of the affected coordinates is modified while the remaining coordinates keep variance equal to one. The magnitude of the variance change is varied across experiments, and performance curves are reported as a function of the variance jump.

For covariance changes, the pre-change covariance structure assumes zero covariance between coordinates. After the change-point, the covariance between each pair of affected coordinates is increased from $0$ to a positive value, which varies across simulations. The performance curves are then reported as a function of this covariance jump.

These experiments allow us to assess how the different \emph{post hoc} procedures behave when the distributional change is not driven by a shift in the mean but by other types of distributional modifications.

\begin{figure}[tbp]
\centering

\begin{subfigure}{0.48\textwidth}
    \centering
    \includegraphics[width=\textwidth]{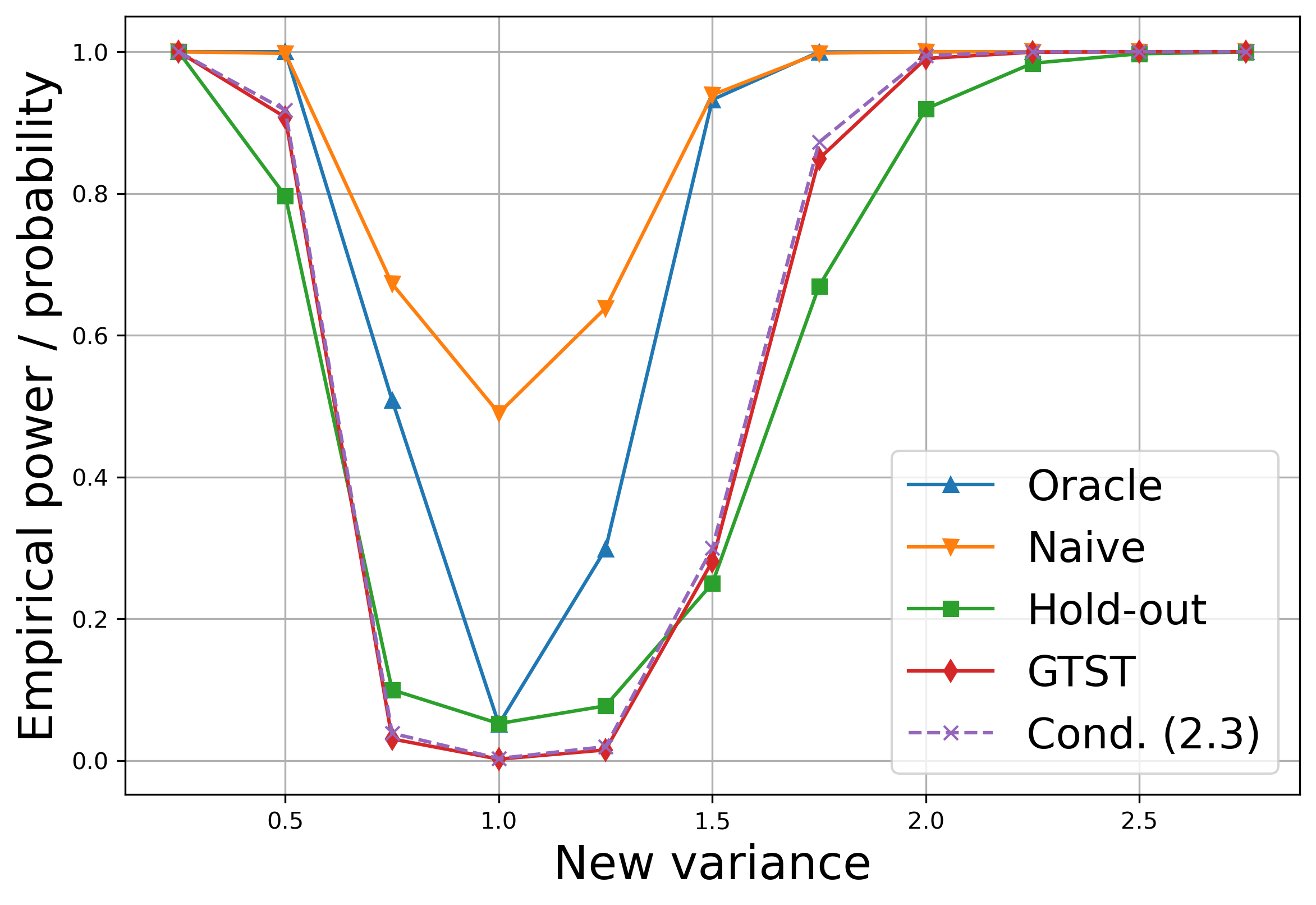}
    \caption{Power}
    \label{fig:var_power}
\end{subfigure}
\hfill
\begin{subfigure}{0.48\textwidth}
    \centering
    \includegraphics[width=\textwidth]{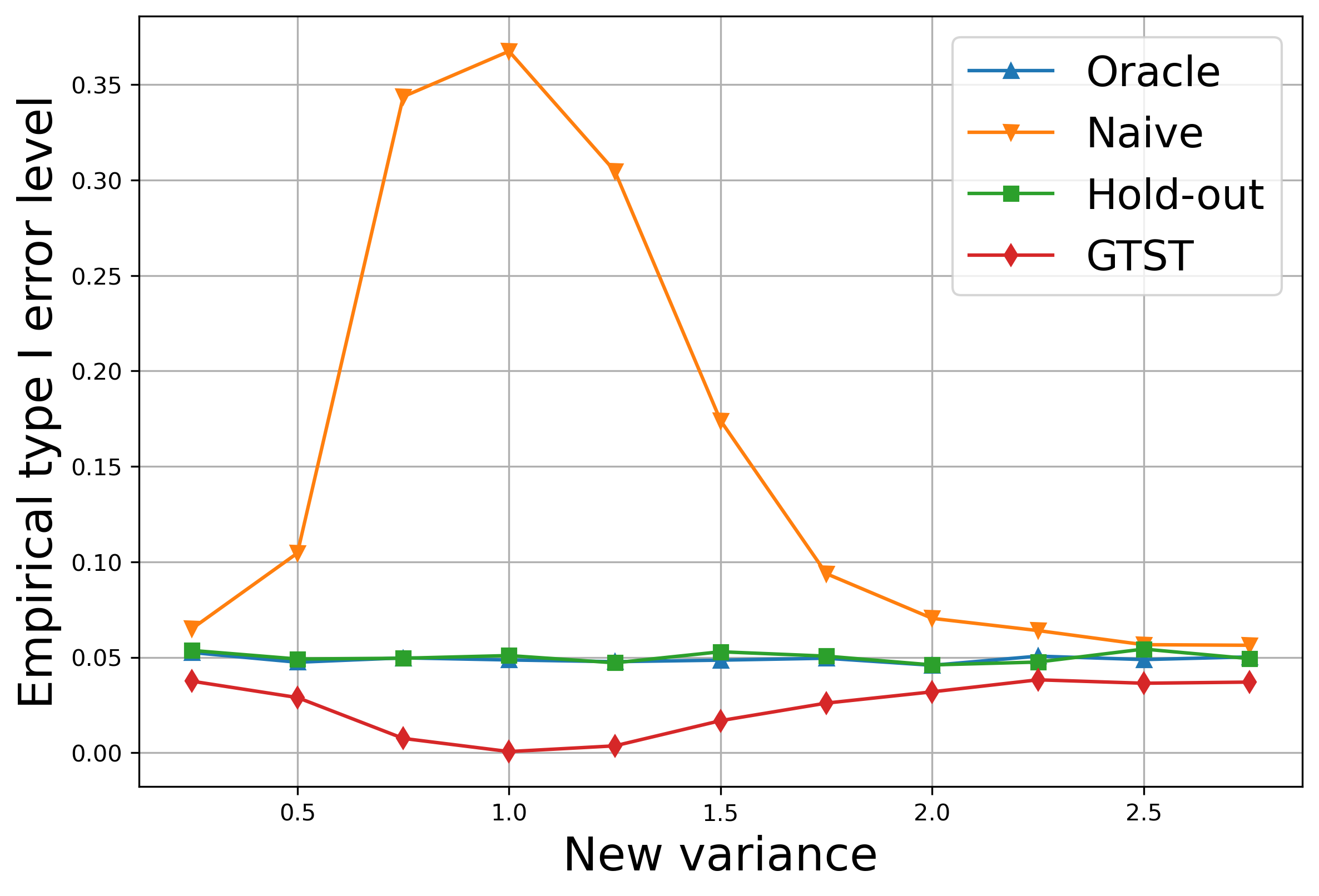}
    \caption{Type I error level}
    \label{fig:var_level}
\end{subfigure}

\caption{
Performance of the \emph{post hoc} procedures as a function of the new variance on the affected coordinates. The pre-change variance is exactly $1.0$. (a) Empirical power; (b) Empirical Type~I error level.
}
\label{fig:var_results}
\end{figure}

\paragraph{Results for variance changes.}
Figure~\ref{fig:var_results} shows that the overall behavior observed for mean shifts remains largely consistent for variance changes. Because the pre-change variance is exactly $1.0$, the empirical power curves exhibit a V-shape as the new variance diverges from this baseline.

\begin{figure}[tbp]
\centering

\begin{subfigure}{0.48\textwidth}
    \centering
    \includegraphics[width=\textwidth]{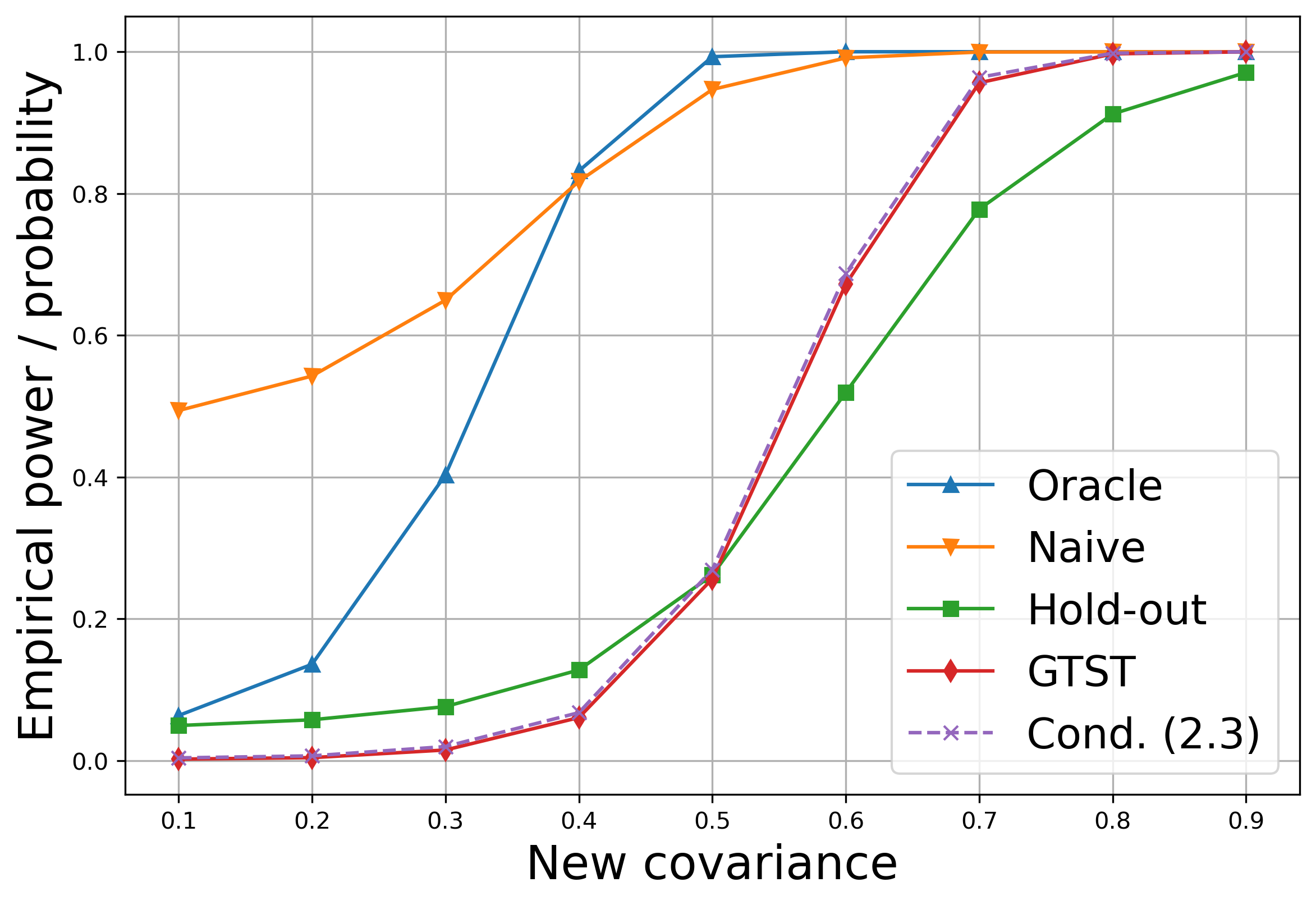}
    \caption{Power}
    \label{fig:cov_power}
\end{subfigure}
\hfill
\begin{subfigure}{0.48\textwidth}
    \centering
    \includegraphics[width=\textwidth]{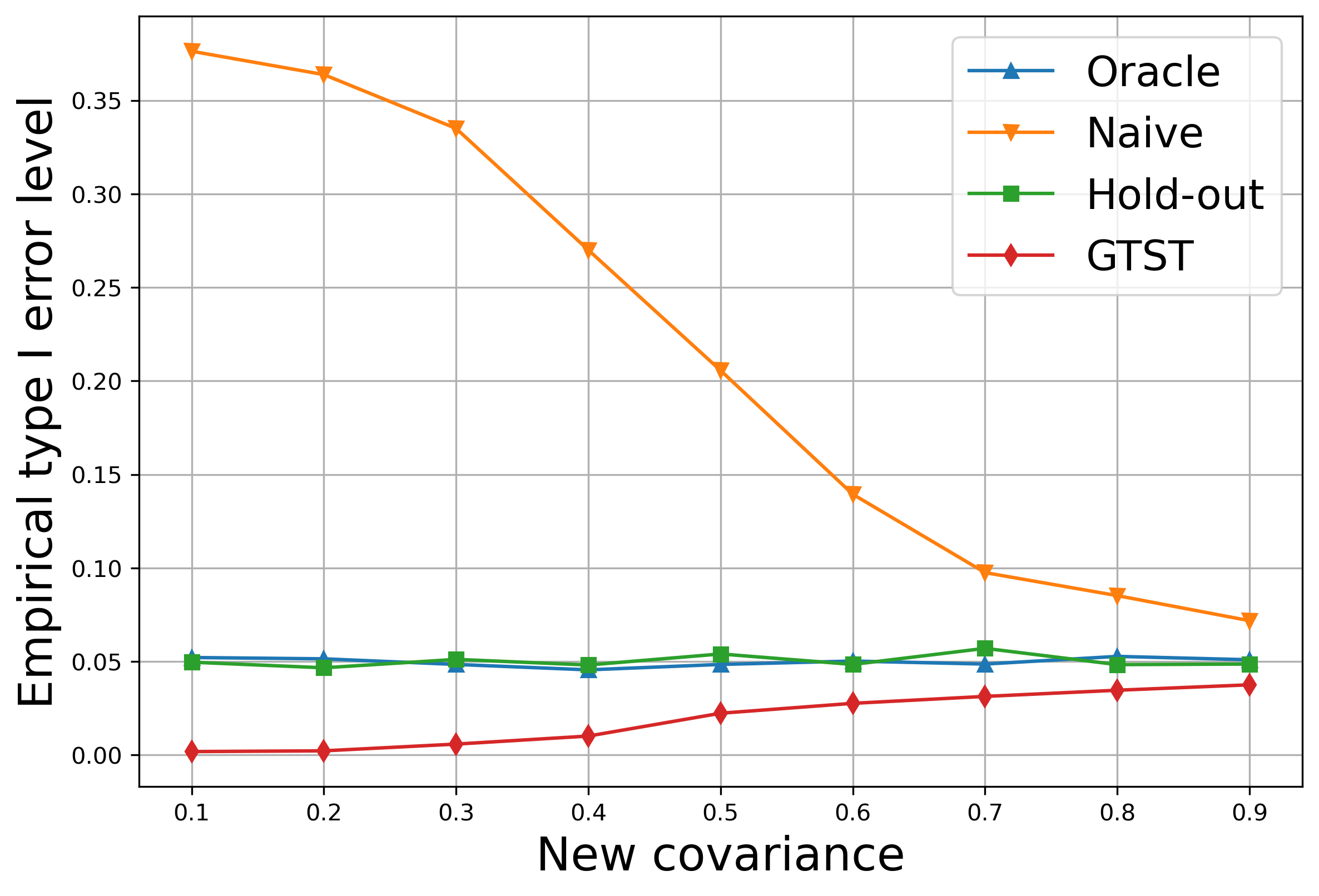}
    \caption{Type I error level}
    \label{fig:cov_level}
\end{subfigure}

\caption{
Performance of the \emph{post hoc} procedures as a function of the covariance jump magnitude. (a) Empirical power; (b) Empirical Type~I error level.
}
\label{fig:cov_results}
\end{figure}

\paragraph{Results for covariance changes.}
Figure~\ref{fig:cov_results} exhibits trends that are consistent with those observed in the mean-shift and variance-change experiments. The relative performance ranking of the \emph{post hoc} procedures remains unchanged. More importantly, these results confirm that the proposed framework is not limited to mean-shift scenarios. Similar patterns are preserved when the change affects other aspects of the distribution, such as the covariance structure. This robustness with respect to the nature of the change constitutes one of the main practical strengths of the proposed approach.

\subsection{Estimation / Noise}
\subsubsection{Influence of the time series length}
In this experiment, we study how the length of the time series influences the performance of the \emph{post hoc} procedures. A single change-point is placed at the center of the sequence and the mean jump amplitude is kept fixed at $\text{jump} = 0.3$, while all other simulation settings remain unchanged except for the sample size $n$, which varies over a predefined range of values. This setup allows us to analyze how the statistical performance of the methods changes as the amount of available data in the two segments increases.
\begin{figure}[tbp]
\centering

\begin{subfigure}{0.48\textwidth}
    \centering
    \includegraphics[width=\textwidth]{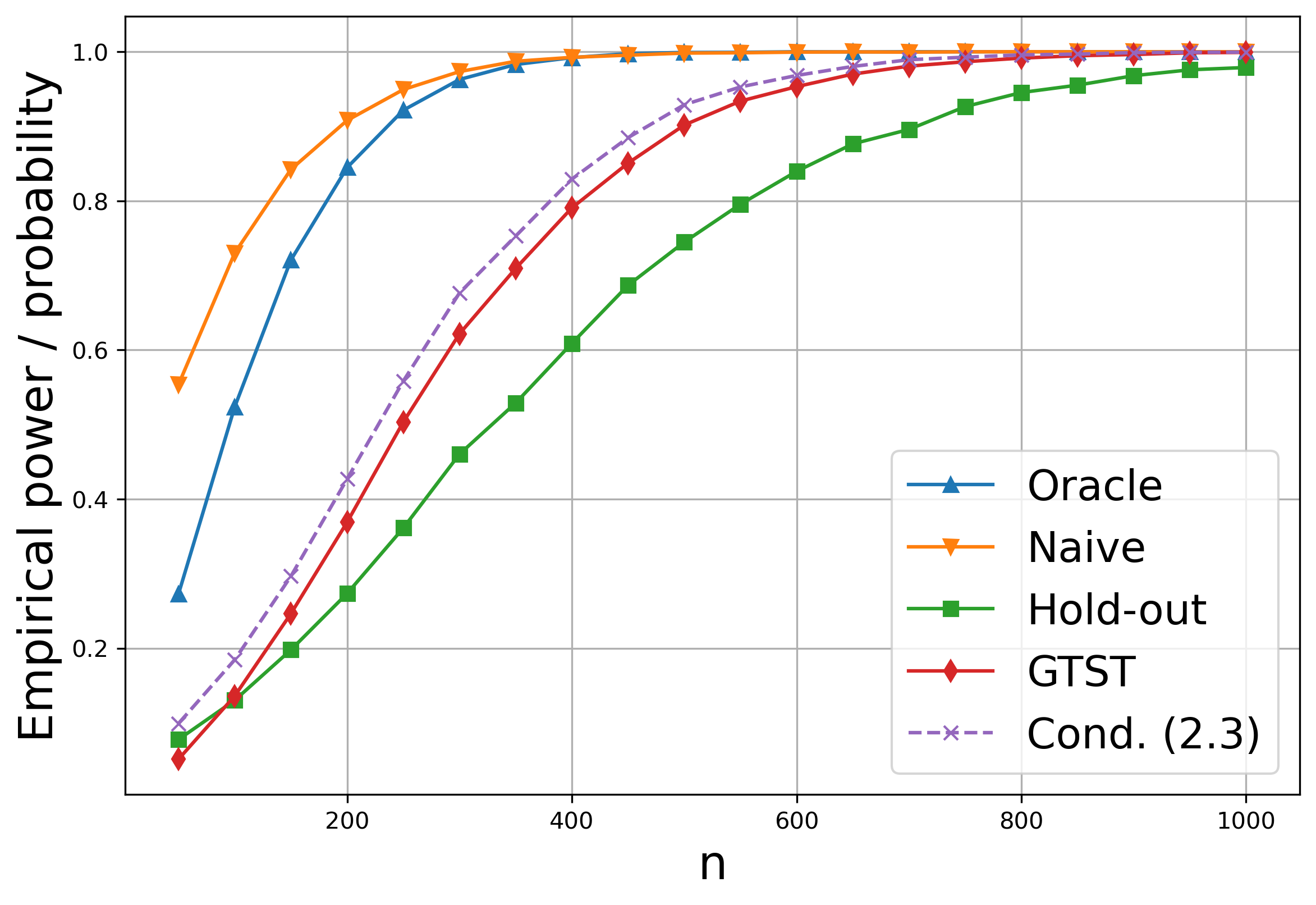}
    \caption{Power}
    \label{fig:n_power}
\end{subfigure}
\hfill
\begin{subfigure}{0.48\textwidth}
    \centering
    \includegraphics[width=\textwidth]{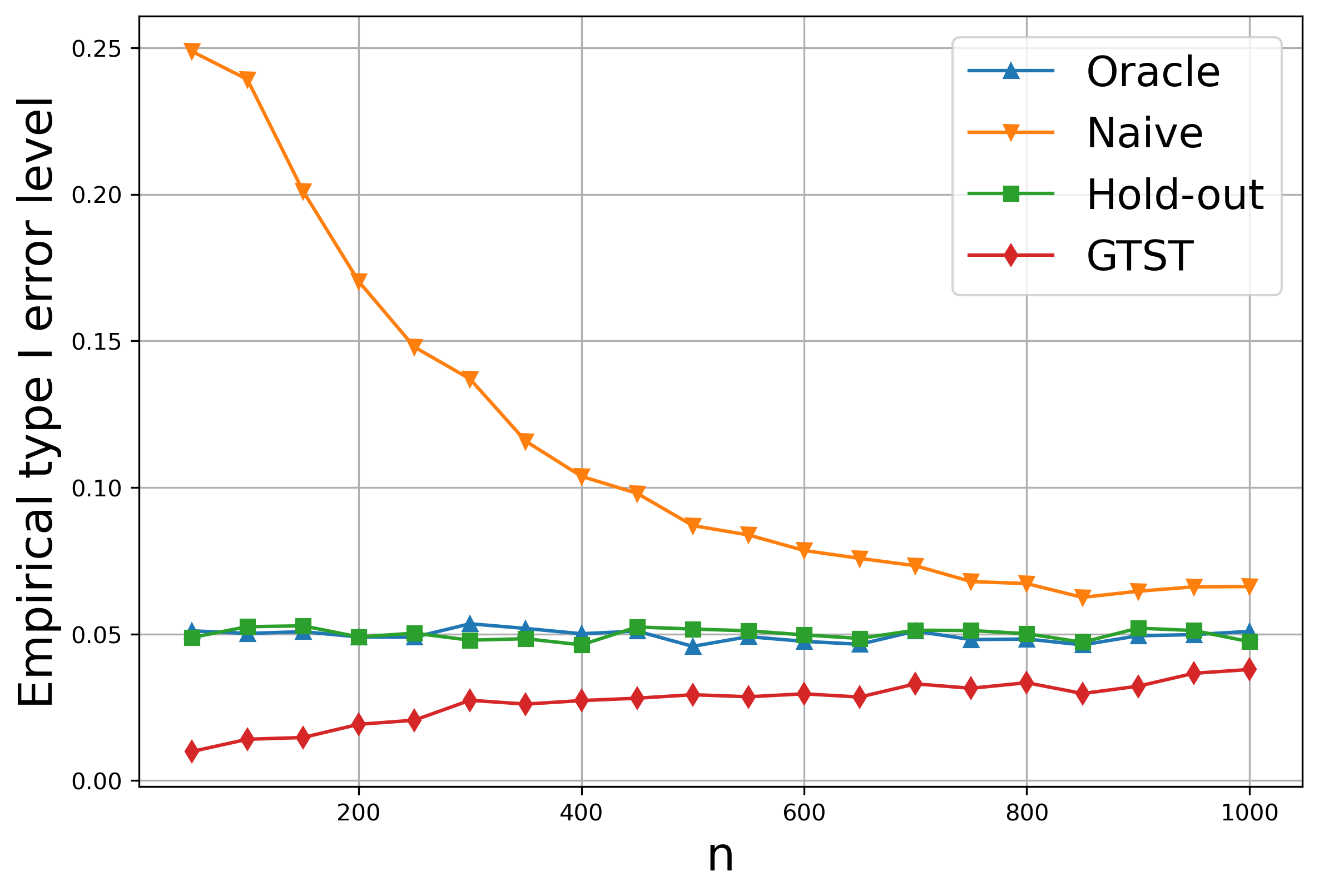}
    \caption{Type I error level}
    \label{fig:n_level}
\end{subfigure}

\caption{
Performance of the \emph{post hoc} procedures as a function of the time series length $n$. A single change-point is located at the center of the sequence, with a fixed mean jump magnitude of $0.3$. (a) Empirical power; (b) Empirical Type~I error level.
}
\label{fig:n_results}
\end{figure}
\paragraph{Results.}
Figure~\ref{fig:n_results} shows that increasing the length of the time series substantially improves the performance of all \emph{post hoc} procedures. For a fixed mean jump amplitude, larger sample sizes lead to higher power. This overall improvement is driven by two main effects. First, a larger $n$ increases the accuracy of the initial change-point detection step, yielding a tighter localization. Second, because the true change-point is located at the center of the sequence, the absolute size of the underlying segments grows proportionally with $n$. This yields larger test segments, which reduces the variance of the empirical $\text{MMD}_b^2$ estimates. Consequently, the global test statistic becomes more reliable, facilitating the correct identification of the active blocks for all methods. Furthermore, the empirical probability that Condition~\eqref{eq:delta_n'} is satisfied increases with the sample size. It is important to note that this is a direct consequence of the central location of the change-point: as the sequence grows, the minimal segment size also diverges, providing sufficient observations to satisfy the theoretical minimal spacing requirements.

\subsubsection{Influence of the number of change-points}

Finally, we investigate scenarios involving multiple change-points in order to assess the behavior of the \emph{post hoc} procedures when several distributional changes occur along the sequence. In these experiments, we consider mean-shift changes and evaluate performance as a function of the jump amplitude.

When $\kappa^\star=3$, two change-points are present. The time series initially follows a standard Gaussian distribution, then undergoes a mean shift affecting $3$ out of $5$ dimensions, and finally returns to the standard Gaussian distribution. The true change-points are located at positions $n/3$ and $2n/3$. \emph{Post hoc} inference is performed for the first change-point, corresponding to $i_0 = 1$.

When $\kappa^\star=4$, three change-points are introduced. The sequence starts from a standard Gaussian distribution, then switches to a mean-shifted regime affecting $3$ out of $5$ dimensions, returns to the standard Gaussian distribution, and finally switches again to the same mean-shifted distribution. In this case, \emph{post hoc} inference focuses on the middle change-point, i.e.\ $i_0 = 2$. For the localized tests (Hold-out and GTST), the uncertainty margin $\delta_n$ is calibrated using the constant $C_0 = 0.1$ in both scenarios. This setting allows us to study the impact of multiple nearby change-points and, in particular, the effect of imperfect localization of neighboring change-points on \emph{post hoc} inference. Performance curves are reported as a function of the mean jump amplitude.

\begin{figure}[tbp]
\centering

\begin{subfigure}{0.48\textwidth}
    \centering
    \includegraphics[width=\textwidth]{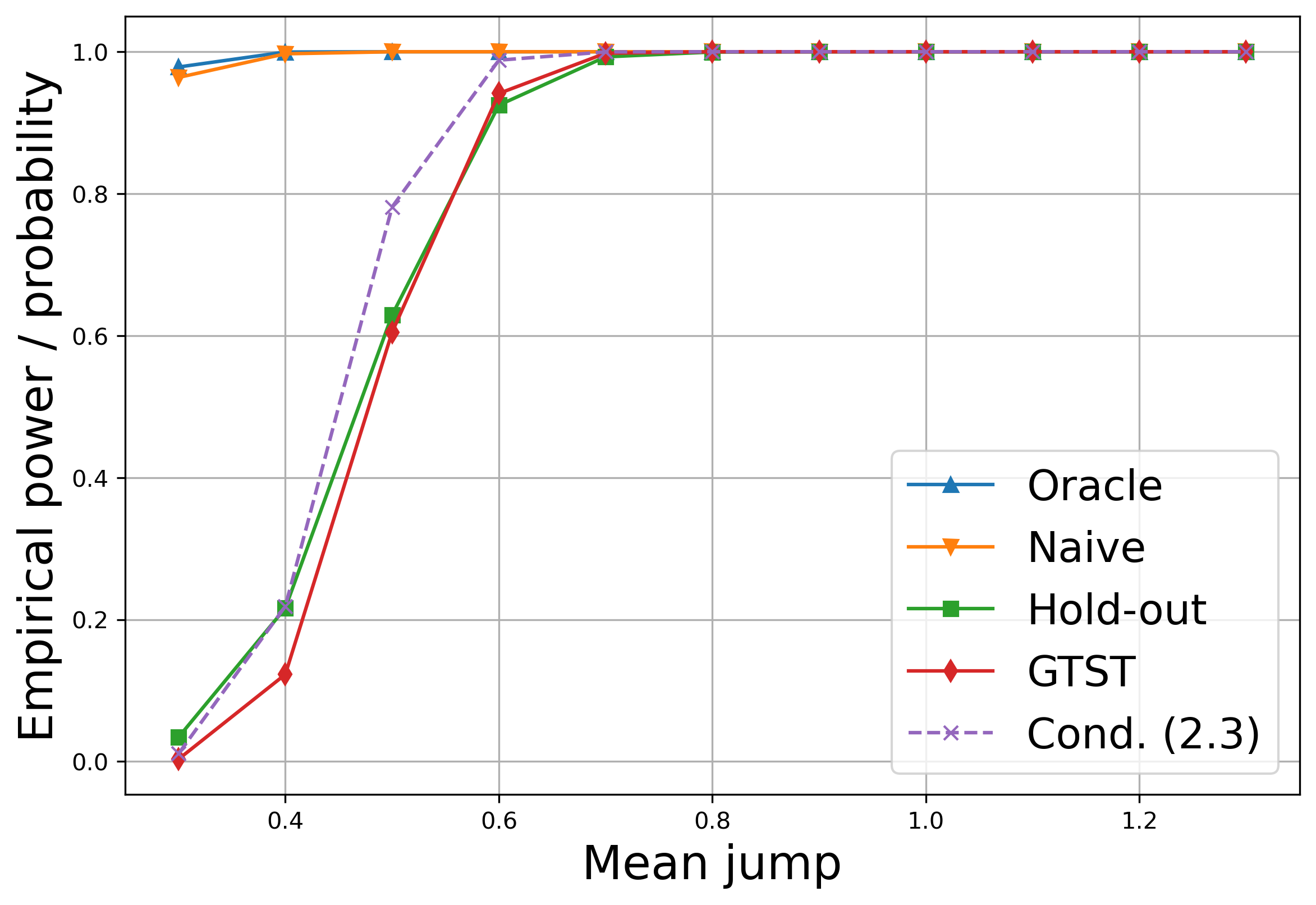}
    \caption{Power}
    \label{fig:D_3_power}
\end{subfigure}
\hfill
\begin{subfigure}{0.48\textwidth}
    \centering
    \includegraphics[width=\textwidth]{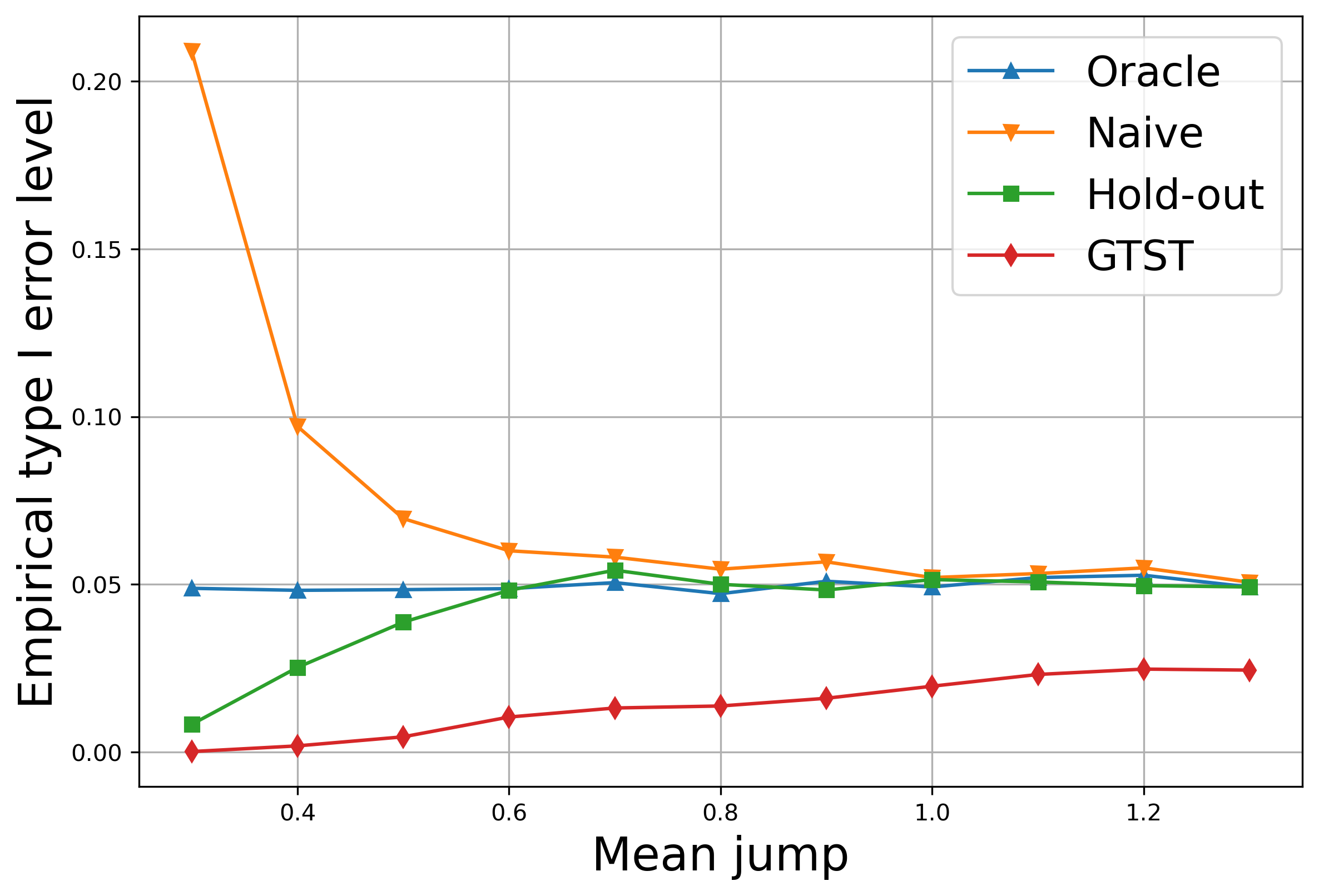}
    \caption{Type I error level}
    \label{fig:D_3_level}
\end{subfigure}

\vspace{0.5cm}

\caption{
Performance of the \emph{post hoc} procedures as a function of the mean jump magnitude in a three-segment scenario ($\kappa^\star=3$). The sequence contains two true change-points, and inference is specifically performed on the first one ($i_0 = 1$). (a) Empirical power; (b) Empirical Type~I error level.
}
\label{fig:D_3_results}
\end{figure}

\paragraph{Results for $\kappa^\star=3$.}

Figure~\ref{fig:D_3_results} reveals a limitation of GTST in the multiple change-point setting. In contrast to the single change-point case, the separability condition underpinning GTST becomes harder as soon as $\kappa^\star>2$. More precisely, the requirement changes from
\[
n\,\underline{\Lambda}_{\tau^\star} > \delta_n \qquad (\kappa^\star=2)
\]
to
\[
n\,\underline{\Lambda}_{\tau^\star} > 2\delta_n \qquad (\kappa^\star>2),
\]
which makes the condition harder to satisfy in practice.

This effect is clearly reflected in the empirical behavior of GTST: when the strengthened condition fails, the procedure conservatively retains the null hypothesis by construction, which leads to a substantial loss of power. In the present configuration with $\kappa^\star = 3$, the hold-out test still executes when $n\,\Lambda_{\tau^\star}$ falls between $\delta_n$ and $2\delta_n$---a regime in which the strengthened separability condition required by GTST fails---albeit without the formal $\kappa^\star > 2$ guarantee in Lemma~\ref{lem:IC}, as discussed in Remark~\ref{rem:holdout_multiple}. This explains its stronger empirical performance in that range.

An additional phenomenon appears in the multiple change-point setting. For the range of jump amplitudes considered here, the oracle and naive procedures quickly attain power close to one, whereas GTST and the hold-out approach may still fail for some of these values. This discrepancy is not due to a weaker signal, but rather to the increased difficulty of the inference problem when several change-points are present, reflecting the advantage of performing inference without explicitly accounting for localization uncertainty. In fact, even for large mean shifts, the presence of multiple change-points reduces the effective amount of data available within each segment and makes localization more challenging, since several boundaries must be estimated simultaneously instead of a single one. As a consequence, errors in change-point localization have a stronger impact on \emph{post hoc} inference. Moreover, the hold-out strategy further reduces the effective sample size by construction, as part of the data is removed to prevent contamination from neighboring segments during the comparison step. Similarly, GTST becomes more demanding in this setting: the uncertainty on both neighboring boundaries makes the testing region random, leading to a larger number of tests and a higher chance of remaining in the conservative regime where the null hypothesis is retained.

Taken together, these results show that multiple change-point configurations are intrinsically harder than the single change-point case, and lead to reduced empirical performance for methods that explicitly account for localization uncertainty. In such settings, the hold-out approach may provide a more robust alternative than GTST because it relies on a weaker applicability condition. However, once the theoretical condition is satisfied for both procedures, GTST tends to outperform the hold-out strategy while maintaining a lower Type~I error level, which ultimately makes it more precise.

\begin{figure}[tbp]
\centering

\begin{subfigure}{0.48\textwidth}
    \centering
    \includegraphics[width=\textwidth]{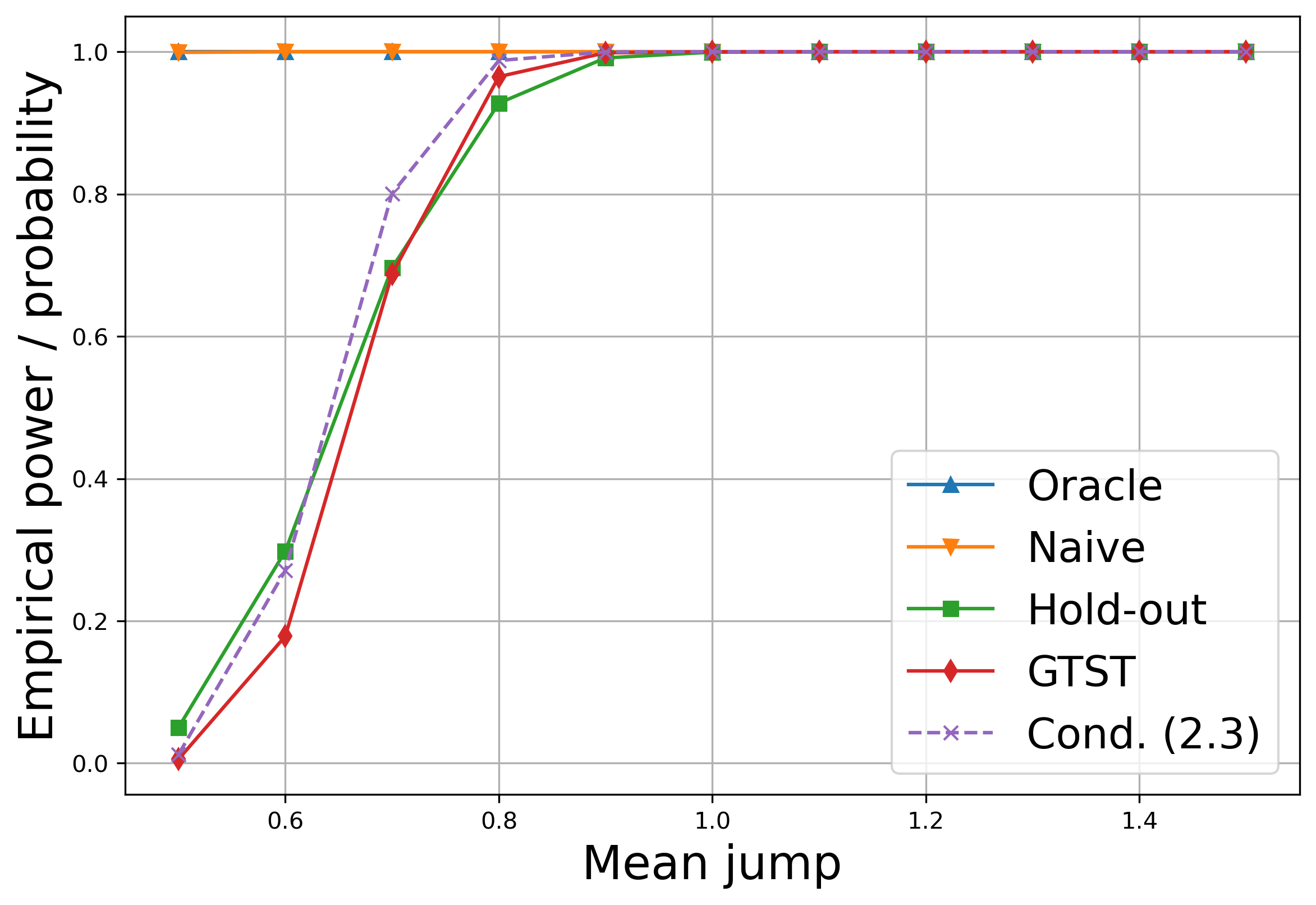}
    \caption{Power}
    \label{fig:D_4_power}
\end{subfigure}
\hfill
\begin{subfigure}{0.48\textwidth}
    \centering
    \includegraphics[width=\textwidth]{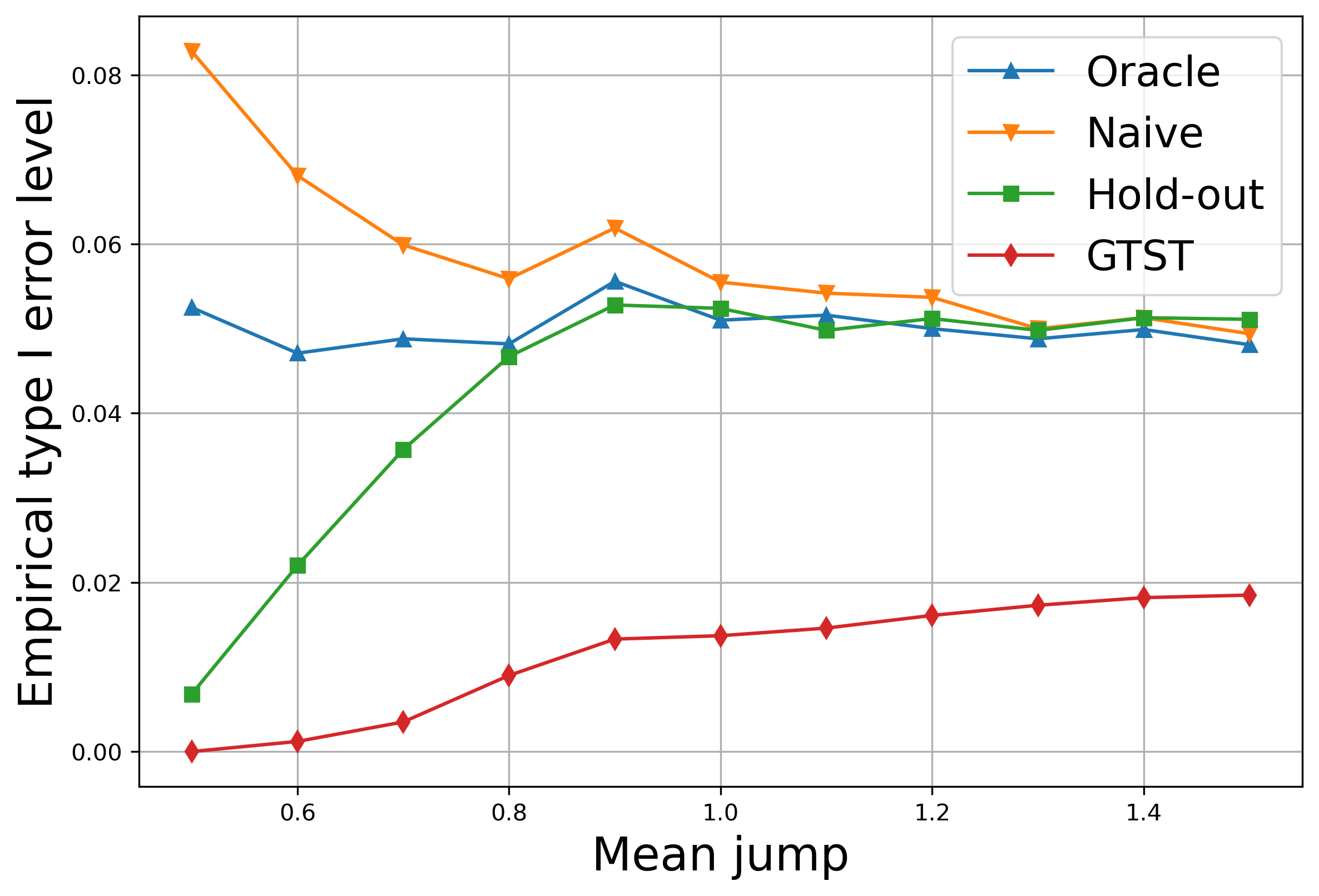}
    \caption{Type I error level}
    \label{fig:D_4_level}
\end{subfigure}

\caption{
Performance of the \emph{post hoc} procedures as a function of the mean jump magnitude in a four-segment scenario ($\kappa^\star=4$). The sequence contains three true change-points, and inference is explicitly focused on the middle one ($i_0 = 2$). (a) Empirical power; (b) Empirical Type~I error level.
}
\label{fig:D_4_results}
\end{figure}

\paragraph{Results for $\kappa^\star=4$.}
When $\kappa^\star=4$, we observe behavior similar to the case $\kappa^\star=3$. In particular, both GTST and the hold-out procedure are governed by the same separability requirement as before. 

However, the empirical power differs slightly between the two configurations. In the case $\kappa^\star=3$, both GTST and the hold-out procedure reach power close to one for mean jumps greater than approximately $0.7$, whereas when $\kappa^\star=4$, a larger jump (around $0.9$) is required to achieve the same level of power.

This difference can be explained by the structure of the segmentation problem. When $\kappa^\star=3$, the change-point under investigation lies between the true start/end of the sequence and an estimated boundary. In contrast, when $\kappa^\star=4$, the change-point that we choose lies between two estimated boundaries. As a consequence, the localization uncertainty affects both sides of the tested segments. To avoid introducing dependence between the localization step and the \emph{post hoc} test, additional precautions must be taken when constructing the testing segments. These adjustments effectively reduce the amount of usable data for the test, which in turn leads to a decrease in statistical power.

Overall, this illustrates that configurations with multiple change-points are intrinsically more challenging, since uncertainty in the localization of neighboring change-points propagates to the \emph{post hoc} inference step.
\end{appendix}

\bibliography{biblio}

\end{document}